\documentclass[showpacs,amsmath,amssymb,twocolumn,superscriptaddress,notitlepage,preprintnumbers,pra]{revtex4-1}
\usepackage{qcircuit}
\usepackage{graphicx}
\usepackage{amsmath,amssymb,amsthm,mathrsfs,amsfonts,dsfont}
\usepackage{mathtools}
\usepackage{epsfig}
\usepackage{braket}
\usepackage{bm}
\usepackage{bbm}
\usepackage{enumerate}
\usepackage{color}
\usepackage[normalem]{ulem}
\usepackage{comment}
\usepackage[colorlinks = true]{hyperref}

\usepackage{makecell} 
\usepackage{multirow} 

\hypersetup{colorlinks=true, linkcolor=blue, citecolor=blue, urlcolor=blue}

\usepackage{algorithm}  
\usepackage{algpseudocode}  

\graphicspath{{./figures/}}

\allowdisplaybreaks

\usepackage{varioref}
\labelformat{equation}{(#1)}

\labelformat{section}{#1}

\labelformat{figure}{#1}

\labelformat{proposition}{#1}

\labelformat{lemma}{#1}

\labelformat{theorem}{#1}

\labelformat{observation}{#1}

\labelformat{definition}{#1}

\newtheorem{theorem}{Theorem}
\newtheorem{proposition}{Proposition}
\newtheorem{lemma}{Lemma}
\newtheorem{corollary}{Corollary}
\newtheorem{definition}{Definition}

\newcommand{\mc}{\mathcal}

\newcommand{\mbb}{\mathbb}
\newcommand{\mr}{\mathrm}

\newcommand{\ad}{\mathrm{ad}}
\newcommand{\tr}{\mathrm{Tr}}

\newcommand{\chicago}{Pritzker School of Molecular Engineering, The University of Chicago, Illinois 60637, USA}

\begin{document}

\title{Simple and high-precision Hamiltonian simulation by compensating Trotter error with linear combination of unitary operations}

\begin{abstract}
Trotter and linear-combination-of-unitary (LCU) are two popular Hamiltonian simulation methods. 
The Trotter method is easy to implement and enjoys good system-size dependence endowed by commutator scaling, while the LCU method admits high accuracy simulation with a smaller gate cost.
We propose Hamiltonian simulation algorithms using LCU to compensate Trotter error, which enjoy both of their advantages. 
By adding few gates after the $K$th-order Trotter formula, we realize a better time scaling than 2$K$th-order Trotter.  
Our first algorithm exponentially improves the accuracy scaling of the $K$th-order Trotter formula. 
For a generic Hamiltonian, the estimated gate counts of the first algorithm can be 2 orders of magnitude smaller than the best analytical bound of fourth-order Trotter formula.
In the second algorithm, we consider the detailed structure of Hamiltonians and construct LCU for Trotter errors with commutator scaling. 
Consequently, for lattice Hamiltonians, the algorithm enjoys almost linear system-size dependence and quadratically improves the accuracy of the $K$th-order Trotter.
For the lattice system, the second algorithm can achieve 3 to 4 orders of magnitude higher accuracy with the same gate costs as the optimal Trotter algorithm.
These algorithms provide an easy-to-implement approach to achieve a low-cost and high-precision Hamiltonian simulation.
\end{abstract}
\date{\today}
            
\author{Pei Zeng}  
\email{peizeng@uchicago.edu}
\affiliation{\chicago}

\author{Jinzhao Sun}
\email{jinzhao.sun.phys@gmail.com}
\affiliation{Clarendon Laboratory, University of Oxford, Parks Road, Oxford OX1 3PU, United Kingdom}
\affiliation{Blackett Laboratory, Imperial College London, London SW7 2AZ, United Kingdom}

 \author{Liang Jiang}
\email{liang.jiang@uchicago.edu}
\affiliation{\chicago}

\author{Qi Zhao}
\email{zhaoqi@cs.hku.hk}
\affiliation{QICI Quantum Information and Computation Initiative, Department of Computer Science, Department of Computer Science, The University of Hong Kong, Pokfulam Road, Hong Kong}

\maketitle

\section{INTRODUCTION} \label{sec:intro}

Hamiltonian simulation, i.e., to simulate the real-time evolution $U(t)=e^{-iHt}$ of a physical Hamiltonian $H=\sum_l H_l$, is considered to be a natural and powerful application of quantum computing~\cite{feynman1982simulating}. 
It can also be used as an important subroutine in many other quantum algorithms like ground-state preparation~\cite{Lloyd1999ground,Alan05Science}, optimization problems~\cite{farhi2014quantum,zhou2020QAOA}, and quantum linear solvers~\cite{HHL}.
To pursue real-world applications of Hamiltonian simulation with near-term quantum devices, we need to design feasible algorithms with small space complexity (i.e., qubit number) and time complexity (i.e., circuit depth and gate number).

One of the most natural Hamiltonian simulation methods is based on Trotter formulas~\cite{lloyd1996universal,suzuki1990fractal,suzuki1991general,berry2007efficient,campbell2019random,childs2019fasterquantum,childs2019nearly,PhysRevA.99.012334,heyl2019quantum,chen2020quantum,sahinoglu2020hamiltonian,su2020nearly,Tran_2020,childs2021theory,layden2021first,PhysRevLett.129.270502}, which approximate the real-time evolution of $H=\sum_{l=1}^L H_l$ by 
the product of the simple evolution of its summands $e^{-iH_lt}$. 
Besides its prominent advantage of simple realization without ancillas, Trotter methods are recently rigorously shown to enjoy commutator scaling~\cite{childs2019nearly,childs2021theory}, i.e., the Trotter error is only related to the nested commutators of the Hamiltonian summands $\{H_l\}$. 
This is very helpful for the Hamiltonians with strong locality constraints. 
For example, when we consider $n$-qubit lattice Hamiltonians, the gate cost of high-order Trotter methods is almost linear to the system size $n$, which is nearly optimal~\cite{childs2019nearly}. 
The major drawback of the Trotter methods is its polynomial gate cost to the inversed accuracy $1/\varepsilon$, $\rm Poly(1/\varepsilon)$. 
This is unfavorable in many applications where high-precision simulation is demanded to obtain practical advantages over the existing classical algorithms~\cite{markus2017elucidating}.

In recent years, we have seen the developments of ``post-Trotter'' algorithms with exponentially improved accuracy dependence~\cite{berry2014exponential,berry2015simulating,Berry15optimal,low2019Hamiltonian,low2017optimal,low2019STOC}. 
Due to the smart choice of the expansion formulas (i.e., Taylor series~\cite{berry2015simulating} or Jacobi-Anger expansion~\cite{low2017optimal}), these post-Trotter methods are able to catch the dominant terms in the time evolution operator $U(t)$ with polynomially increasing gate resources, leading to a logarithmic gate-number dependence on the accuracy requirement $1/\varepsilon$. 
Unlike Trotter methods, these advanced algorithms are not able to utilize the specific structure of Hamiltonians due to the lack of commutator-based error form.
Consequently, for instance, for $n$-qubit lattice Hamiltonians, their gate complexities are $\mc{O}(n^2)$, which is worse than those in Trotter algorithms $\mc{O}(n^{1+o(1)})$.
Furthermore, these post-Trotter algorithms require the implementation of linear combination of unitary (LCU) formulas~\cite{childs2012Hamiltonian,long2011} or block encoding of Hamiltonians~\cite{low2019Hamiltonian} which often costs many ancillary qubits and multicontrolled Toffoli gates. 
This is still experimentally challenging in a near-term or early fault-tolerant quantum computer~\cite{lin2021heisenberglimited}.
To reduce the hardware requirement of compiling the LCU formulas, recent studies focus on a random-sampling implementation of a LCU formula~\cite{childs2012Hamiltonian,yang2021accelerated,wan2021randomized,faehrmann2021randomizing}, where the elementary unitaries are sampled to realize the LCU formula statistically. 
In this case, the Hamiltonian simulation is not performed by coherently implementing the unitary $U=e^{-iHt}$, but is instead realized through random sampling. This method remains effective for common applications, such as estimating the properties of the final state.
Similar ideas have also been studied in the ground-state preparation algorithms~\cite{lin2020nearoptimalground,zeng2021universal,zhang2022computingground}.

Here, we propose composite algorithms that combine the inherent advantages of Trotter and LCU methods---easy implementation, high precision, and commutator scaling---by performing the Trotter method and then compensating the Trotter error with the LCU formulas we construct. 
We primarily focus on the random-sampling implementation of the LCU formula~\cite{childs2012Hamiltonian,yang2021accelerated,wan2021randomized,faehrmann2021randomizing}, with the goal of estimating the properties of the target state after real-time evolution.
We demonstrate that optimal performance can be achieved by allowing the Trotter circuit to handle the majority of the simulation, with the LCU method completing the remainder.

In \autoref{sec:summary}, we provide a summary of our construction and results, aimed at a general audience. 
We explain the key ideas behind the constructions with intuitive examples. 
For readers interested in the technical aspects, we introduce the necessary preliminary knowledge of Hamiltonian simulation in \autoref{sec:preliminaries} to facilitate understanding of the technical results. Next, in \autoref{sec:PTSC} and \autoref{sec:NCC}, we present a detailed construction and gate complexity analysis of the two Trotter-LCU algorithms. Finally, in \autoref{sec:conclusion}, we conclude our discussion and outline possible future directions.

\section{SUMMARY of RESULTS} \label{sec:summary}

\subsection{General idea}

The major idea of the proposed Trotter-LCU algorithm is illustrated in \autoref{fig:illustration}. In a normal $K$th-order Trotter circuit, we decompose the time evolution $U(t)=e^{-iHt}$ to $\nu$ segments, each with a small evolution time $x=t/\nu$. For consistency, we denote the $0$th-order Trotter formula as $S_0(x)=I$. After we perform the $K$th-order Trotter formula $S_K(x) (K=0,1,2k, k\in\mbb{N}_+)$, there will be a remaining Trotter error $V_K(x):=U(x)S_K(x)^\dag$, which affects the simulation accuracy. To address this problem, we introduce a random LCU formula to compensate the Trotter error using one ancilla and simple gates, which achieves a high-precision Hamiltonian simulation with low cost.

\begin{figure}[htbp]
\centering
\includegraphics[width=\linewidth]{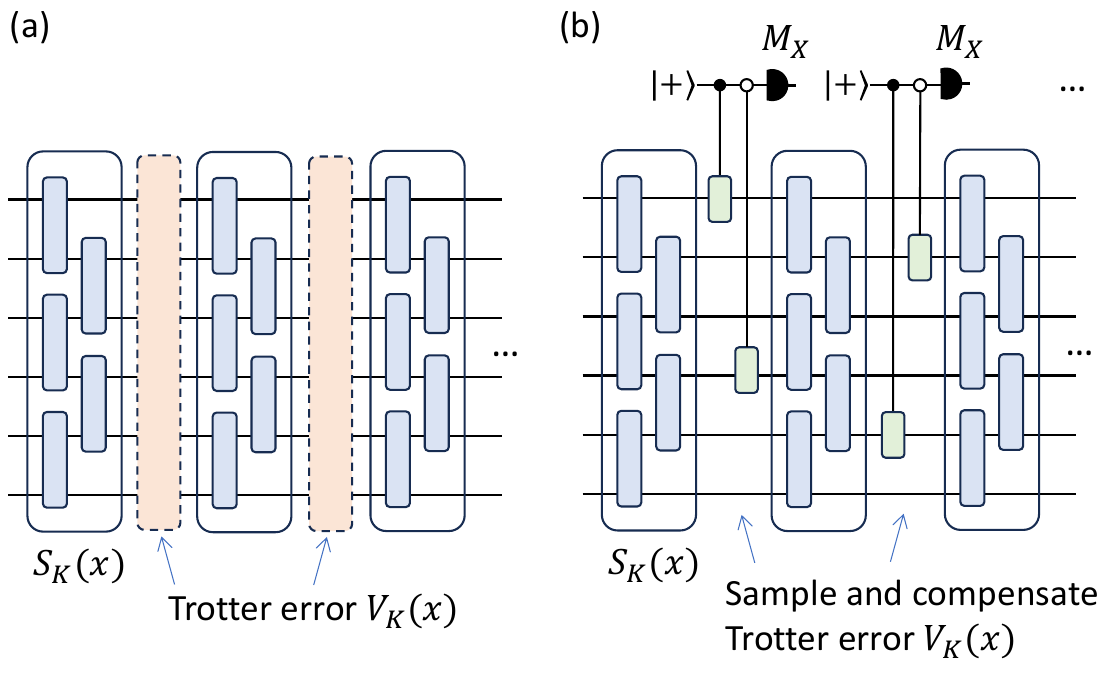}
\caption{(a) In the normal $K$th-order Trotter circuit, there will be a remaining Trotter error $V_K(x)$ after each segment. (b) We introduce random LCU formulas to compensate $V_K(x)$ with a single ancilla qubit and simple gates.}
\label{fig:illustration}
\end{figure}

Consider the following LCU formula of an operator $V$,
\begin{equation} \label{eq:LCUmain}
\tilde{V} = \mu \sum_{i=0}^{\Gamma-1} \Pr(i) V_i,
\end{equation}
such that the spectral norm distance $\|V-\tilde{V}\|\leq \varepsilon$. Here, $\mu>0$ is the $1$-norm (i.e., $l_1$-norm) of the coefficient vector, $\Pr(i)$ is a probability distribution over different unitaries $V_i$, and $\{V_i\}_{i=0}^{\Gamma-1}$ is a set of unitaries.
There are usually two ways to implement the LCU formula: the coherent implementation~\cite{berry2014exponential,berry2015simulating} and the random-sampling implementation~\cite{yang2021accelerated,wan2021randomized,faehrmann2021randomizing}.
Our major focus is on the random-sampling implementation, where we can estimate the properties of the target state $\rho=U(t) \rho_0 U(t)^\dag$ with only one ancillary qubit. In Appendix~\ref{sec:AppCoherent}, we discuss the potential use of the coherent implementation of our algorithm.
In the random-sampling implementation, we can use \autoref{eq:LCUmain} to estimate an arbitrary observable value $O$ on $\rho$,
\begin{equation}
\tr(O\rho) \approx \mu^2 \sum_{i,j} p_i p_j \tr(O V_i \rho_0 V_j^\dag). 
\end{equation}
As is shown in \autoref{fig:TrotterLCUidea}(b), since the estimation of $\tr(O V_i \rho_0 V_j^\dag)$ can be implemented using Hadamard-test-type circuits~\cite{kitaev1995quantum}, we only need to sample $V_i$ and $V_j$ based on the LCU formula in \autoref{eq:LCUmain} to estimate $\tr(O\rho)$ with $\varepsilon$ accuracy using $\mc \mc{O}(\mu^4/\varepsilon^2)$ sampling resource, which owns an extra $\mu^4$ overhead compared to the normal Hamiltonian simulation algorithms~\cite{faehrmann2021randomizing}. To make the algorithm efficient, we need to set $\mu$ to be a constant.
We also provide a variant in \autoref{fig:TrotterLCUidea}(c) where the ancillary qubit is measured and reset in each segment, which is equivalent to \autoref{fig:TrotterLCUidea}(b) for the observable estimation. In this case, the expectation value of $\mu^2 X^{(tot)}_A \otimes O_S$ provides an unbiased estimation of $\tr(O \rho)$ where $X^{(tot)}:=\bigotimes_{k=1}^\nu X_k$ is the multiplication of all the ancillary measurement values. This variant reduces the need to store the ancilla qubit, simplifying the implementation on a fault-tolerant quantum computer.

\begin{figure}[htbp]
\centering
\includegraphics[width=\linewidth]{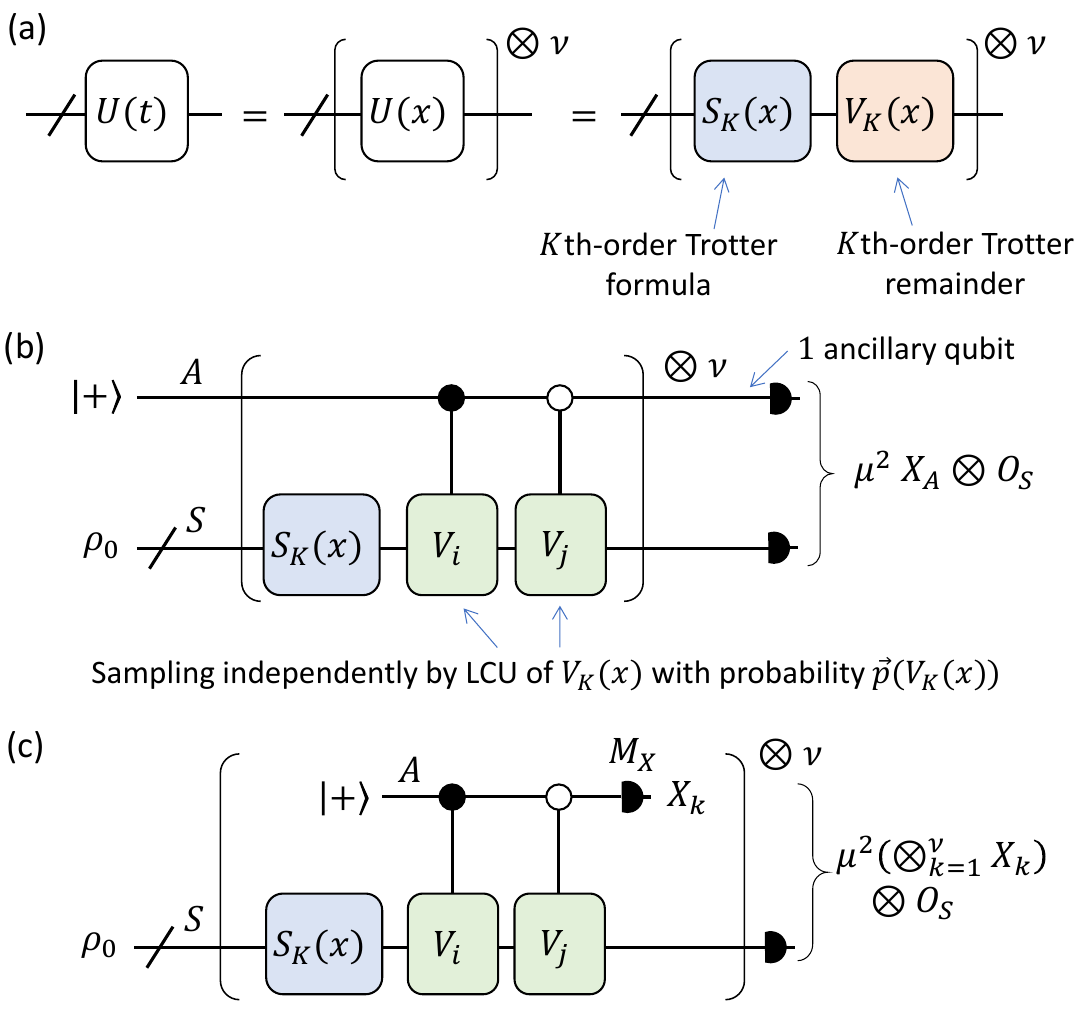}
\caption{(a) In the $K$th-order Trotter-LCU algorithm, we first implement $K$th-order Trotter formula, then compensate the remainder using the LCU formulas we construct. (b) Random-sampling implementation of the LCU formula. We sample the elementary unitaries $V_i$ and $V_j$ independently based on the LCU formula of $V_K(x)$ and implement the controlled $V_i$ and $V_j$ gate. Then the Trotter remainder $V_K(x)$ can be compensated in a Hadamard-test type circuit. (c) A variant of the implementation where the ancillary qubit is measured and reset in each segment. The detailed sampling procedure of $V_i$ and $V_j$ is shown in \autoref{fig:SamplingPTS} and \autoref{fig:SamplingNC}.
}
\label{fig:TrotterLCUidea}
\end{figure}

The construction of an appropriate LCU formula for the $K$th-order Trotter remainder, $V_K(x)$, is crucial for developing an efficient Hamiltonian simulation algorithm. 
In the following subsections, we briefly introduce two approaches for constructing the LCU formula for $V_K(x)$. The resulting composite Trotter-LCU algorithms are referred to as paired Taylor-series compensation (PTSC) and nested-commutator compensation (NCC), respectively. 
Detailed analysis and performance proofs for these two algorithms can be found in \autoref{sec:PTSC} and \autoref{sec:NCC}.

\subsection{Paired Taylor-series compensation: overview}

Without loss of generality, we focus on the case of an $n$-qubit Hamiltonian $H$, which can be written as
\begin{equation} \label{eq:H}
H = \sum_{l=1}^L H_l = \sum_{l=1}^L \alpha_l P_l = \lambda \sum_{l=1}^L p_l P_l, 
\end{equation}
where $\{P_l\}_l$ are different $n$-qubit Pauli operators. We set all the coefficients $\{\alpha_l\}_l$ to be positive and absorb the signs into Pauli operators $\{P_l\}_l$. 
$\lambda := \sum_l \alpha_l$ is the $l_1$-norm of the Hamiltonian coefficient vector.
We consider the Hamiltonians where $\lambda$ increases polynomially with respect to $n$.

\begin{figure*}[ht]
\centering
\includegraphics[width=0.85\linewidth]{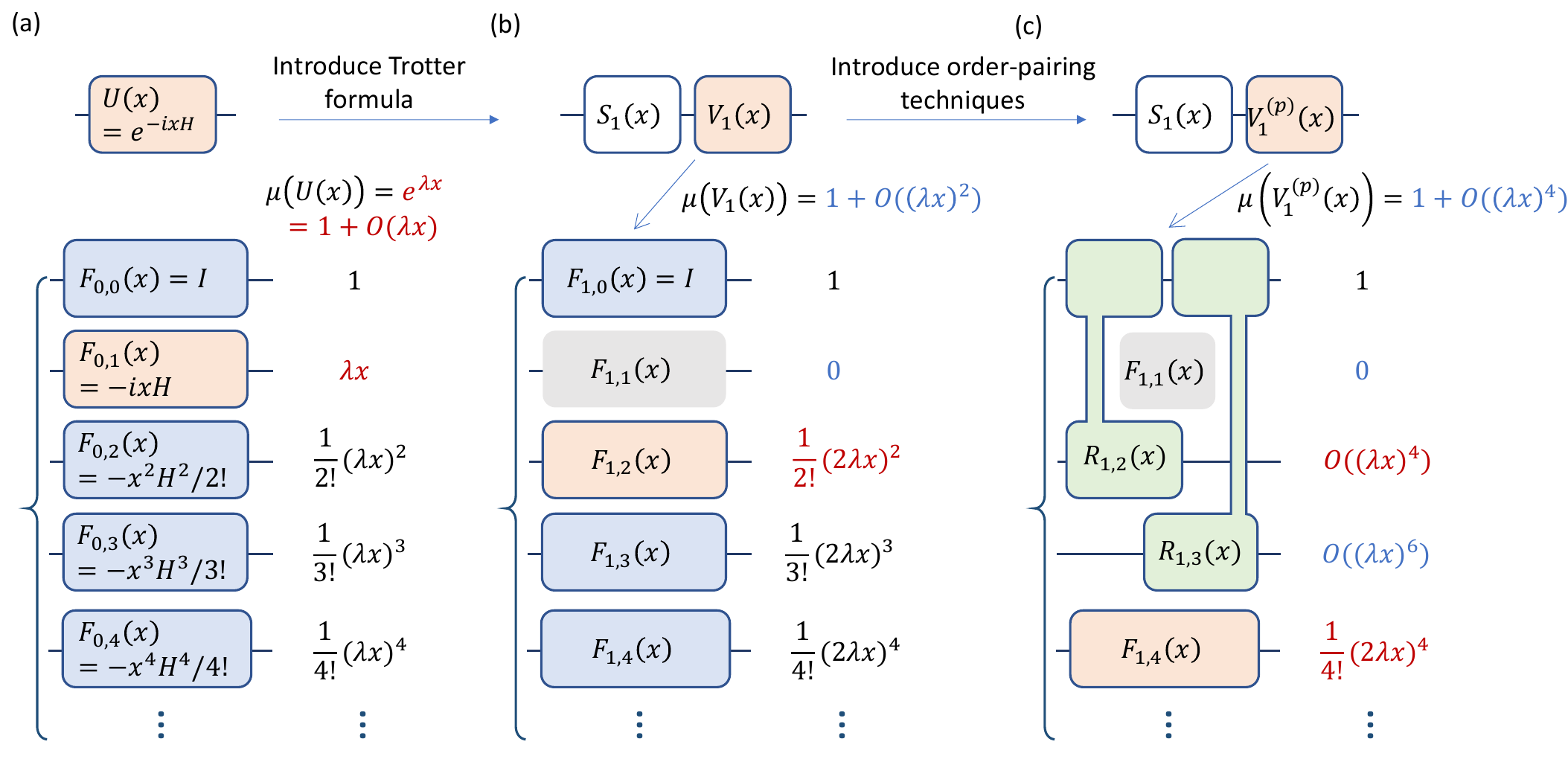}
\caption{Illustration of the idea of paired Taylor-series compensation (PTSC) algorithm. We take the first-order algorithm as an example. (a) Taylor-series expansion of the small-time evolution $U(x)$. The $1$-norm is $\lambda(U(x))=e^{\lambda x}=1+\mc{O}(\lambda x)$. The dominant term is contributed from the first-order expansion $F_{0,1}(x)$. (b) By introducing the first-order Trotter formula $S_1(x)$, the first-order expansion $F_{1,1}$ in the Trotter remainder $V_1(x)$ becomes $0$. As a result, the $1$-norm of $V_1(x)$ is suppressed to $1+\mc{O}((\lambda x)^2)$, limited by the second-order expansion term $F_{1,2}(x)$. (c) By further noticing that $F_{1,2}$ and $F_{1,3}$ are anti-Hermitian, we introduce Euler's formula in \autoref{eq:Euler} to pair the leading-order terms to $F_{1,0}:=I$. This will double the order of $\mu_x-1$ and further suppress the $1$-norm $\mu$ of the overall evolution $U(t)$.}
\label{fig:PairedTSidea}
\end{figure*}

We first consider to construct the LCU formulas for $V_K(x)$ from Taylor-series expansion~\cite{berry2015simulating}.
In the $0$th-order case, when no Trotter formula is introduced, the Trotter remainder $V_0(x)$ is the short-time evolution $U(x)$ itself, which can be expanded as
\begin{equation} \label{eq:DTS0th}
\begin{aligned}
V_0(x)&= e^{-ixH} = \sum_{s=0}^\infty F_{0,s}(x) =\sum_{s=0}^\infty \frac{(-ix)^s}{s!} H^s \\
&= e^{\lambda x} \sum_{s=0}^\infty \mr{Poi}(s;\lambda x) \sum_{l_1,...,l_s} p_{l_1}...p_{l_s} P_{l_1}...P_{l_s}.
\end{aligned}
\end{equation}
Here, $\mr{Poi}(s;a):=e^{-a} a^s/s!$ is the Poisson distribution. 
$F_{0,s}(x)$ denotes the $s$-order expansion term.
\autoref{eq:DTS0th} illustrated in \autoref{fig:PairedTSidea}(a) is a LCU formula of $V_0(x)$ with $1$-norm $\mu_x=e^{\lambda x}$. The $1$-norm of the overall evolution $U(t)=U(x)^\nu$ is then $\mu= e^{\lambda t}$, which, unfortunately, grows exponentially with respect to $t$ regardless of how much we increase the segment number $\nu$. This implies that the direct random-sampling implementation of the LCU formula $V_0(x)$ in \autoref{eq:DTS0th} is not feasible. In Ref.~\cite{berry2015simulating}, the authors discuss the coherent implementation of $V_0(x)$ instead and find that one can achieve good time and accuracy dependence in that scenario. 

When focusing on the random-sampling implementation, we need to suppress the $1$-norm $\mu_x$ of each segment.
To this end, we first consider the usage of the Trotter formula. For example, if we apply first-order Trotter formula $S_1(x)=\prod_{l=1}^L e^{-ixH_l}$ in each segment, the first-order remainder $V_1(x):=U(x)S_1(x)^\dag$ can be expanded as
\begin{equation} \label{eq:DTS1st}
\begin{aligned}
V_1(x)&= e^{-ixH}\prod_{l=L}^1 e^{ixH_l} = \sum_{s=0}^\infty F_{1,s}(x) \\
&=\sum_{r;r_1,r_2,...,r_L} \frac{(-ix)^r}{r!} H^r \prod_{l=L}^1 \frac{(ix)^{r_l}}{r_l!} H^{r_l}, \\
\end{aligned}
\end{equation}
where $\{r_1,...,r_L\}$ 
denotes $L$ expansion variables related to the Trotter formula and $s:=r+\sum_{l=1}^L r_l$. $F_{1,s}$ denotes the $s$-order expansion term of $V_1(x)$. 
The $1$-norm of $V_1(x)$ in \autoref{eq:DTS1st} is $e^{2\lambda x}$, which seems to be even larger than the $0$th-order case. However, since $V_1(x)$ denotes the (multiplicative) Trotter error, we have $F_{1,1}(x) = 0$. Using this condition, we can rewrite $V_1(x)$ as
\begin{equation} \label{eq:DTS1stNew}
V_1(x)= I + \sum_{s=2}^\infty F_{1,s}(x),
\end{equation}
as illustrated in \autoref{fig:PairedTSidea}.
From the Taylor-series expansion, we can bound the $1$-norm of the new LCU formula in \autoref{eq:DTS1stNew} by $\mu_x=e^{2\lambda x} - (2\lambda x)\leq e^{(2\lambda x)^2}$. In this way, we reduce $\mu_x$ from $1+\mc{O}(\lambda x)$ to $1+\mc{O}((\lambda x)^2)$. The $1$-norm of the overall time evolution becomes $\mu=\mu_x^\nu \leq \exp((2\lambda t)^2/\nu)$. 
As a result, by increasing the segment number $\nu$ — or equivalently reducing the unit evolution time $x$ — we can decrease the 1-norm $\mu$, leading to a lower sampling cost. If we choose the segment number as $\nu = \mathcal{O}((\lambda t)^2)$, $\mu$ will remain constant.

From the above discussion, it is clear that, to reduce the 1-norm $\mu$ of the overall LCU formula for $U = e^{-iHt}$, our main objective is to suppress the leading-order term of the $1$-norm remainder $\mu_x - 1$ for each segment, which determines the number of segments $\nu$ and hence circuit depth when $\mu$ is set to be a constant.

We can further reduce the $1$-norm $\mu_x$ of $V_1(x)$ by taking advantage of the structure of Trotter errors. 
For an anti-Hermian Pauli operator $\pm iP$ where $P\in\{I,X,Y,Z\}^{\otimes n}$, we have the following Euler's formula, 
\begin{equation} \label{eq:Euler}
I \pm iyP = \sqrt{1+y^2} e^{\pm i\theta P}.
\end{equation}
Here, $\theta=\tan^{-1}(y)$ and we suppose $0<y<1$. The 1-norm of the left-hand side of \autoref{eq:Euler} is $1+y$, while the 1-norm of the right-hand side is $\sqrt{1 + y^2} < 1 + \frac{y^2}{2} = 1 + \mc{O}(y^2)$. As a result, the exponent of $\mu_x-1$ is effectively doubled.
In \autoref{ssc:pairedTS1}, we prove that the expansion terms $F_{1,2}$ and $F_{1,3}$ in the LCU formula \autoref{eq:DTS1st} are anti-Hermitian. 
As a result, we can further suppress $\mu_x$ by pairing the terms in $F_{1,s}(x)$ with $F_{1,0}=I$ using Euler's formula in \autoref{eq:Euler}. 
When $y=x^s$, the paired formula $R_{1,s}(x)$ as a summation of Pauli rotation unitaries owns the $1$-norm of $\mu_x=1+\mc{O}(x^{2s})$, whose $x$ dependence is doubled, as illustrated in\autoref{fig:PairedTSidea}(c).

To generalize the discussion, for the $K$th-order Trotter error $V_K(x)$ ($K=2k$, $k\in\mbb{N}_+$), we have $F_{K,1}=...F_{K,K}=0$~\cite{suzuki1990fractal,childs2021theory}. Moreover, we can show that $F_{K,K+1}$, $F_{K,K+2}$,..., $F_{K,2K+1}$ are anti-Hermitian. We call the term with $s=K+1$ to $2K+1$ the leading-order terms. The algorithm utilizing $K$th-order Trotter formula and the paired idea in the LCU construction is called the $K$th-order PTSC algorithm. 
We provide the detailed algorithm description and gate-complexity analysis in \autoref{sec:PTSC}.

The PTSC algorithm is generic for the $L$-sparse Hamiltonian $H=\sum_{l=1}^L H_l = \lambda \sum_{l=1}^L p_l P_l$ with $\lambda=\sum_l\|H_l\|$ and $\{P_l\}$ are Pauli matrices. 
It can be implemented using a simple and universal classical random-sampling procedure: first, we sample the order $s$ from the Taylor-series expansion, and then we sample a Pauli string based on the Hamiltonian coefficients.
With the random-sampling implementation, we prove that by appending few gates after the $K$th-order Trotter formula with only one ancillary qubit, one can improve the time scaling from $1+1/K$ to $1+1/(2K+1)$ and exponentially improves the accuracy scaling of the $K$th-order Trotter formula compared to the $K$th-order Trotter. We have the following theorem.

\textbf{Theorem 1} (Informal, see \autoref{thm:PTS} in \autoref{sec:PTSC}) 
In a $K$th-order paired Taylor-series compensation algorithm ($K=0,1$ or $2k$, $k\in\mbb{N}_+$), the gate complexity in a single round is $\mc{O}\left((\lambda t)^{1+ \frac{1}{2K+1}}(\kappa_K L + \frac{\log(1/\varepsilon)}{\log\log(1/\varepsilon)} ) \right)$, where $\lambda=\sum_l\|H_l\|$, $\kappa_K=K$ when $K=0$ or $1$, $\kappa_K=2\times 5^{K/2-1}$ otherwise. 

From Theorem~1, we observe that setting $K = 0$, i.e., not using Trotter formulas, still yields a valid PTSC algorithm by pairing $F_{0,1}$ with $I$. In this case, the gate complexity is independent of the sparsity $L$ but quadratically dependent on $t$, similar to the algorithm in Ref.~\cite{wan2021randomized}. Conversely, when using a $K$th-order Trotter formula, the PTSC algorithm becomes $L$ dependent with an almost linear dependence on $t$. In both cases, the PTSC algorithms achieve high simulation accuracy $\varepsilon$. We expect the $0$th-order algorithm to be particularly useful for quantum chemistry Hamiltonians with large $L$, while higher-order algorithms are better suited for generic $L$-sparse Hamiltonians with long simulation times $t$.

\subsection{Nested-commutator compensation: overview}

The PTSC algorithms above are generic and applicable to any Hamiltonian. When we consider the detailed structure of Hamiltonians, we could make the compensation algorithms more efficient by taking advantage of the commutation relationship of the terms in the Hamiltonians, which was formerly also studied in the Trotter algorithms~\cite{childs2021theory}. 

We will take the first-order Trotter remainder $V_1(x)$ as an illustrative example. Following the Taylor-series expansion in \autoref{eq:DTS1st}, the second-order term $F_{1,2}(x)$ in $V_1(x)$ can be written as
\begin{equation} \label{eq:F12xExp}
\begin{aligned}
F_{1,2}(x) &= \sum_{\substack{r;\vec{r}\\ r+\sum \vec{r} = 2}} \frac{(-ix H)^r}{r!} \prod_{l=L}^1 \frac{(ix H_l)^{r_l}}{r_l!}.
\end{aligned}
\end{equation}
Since $F_{1,2}(x)$ is anti-Hermitian, all the Hermitian expansion terms in \autoref{eq:F12xExp} will cancel out. 
We can then simplify $F_{1,2}(x)$ as,
\begin{equation} \label{eq:F12xComm}
\begin{aligned}
F_{1,2}(x) 
&= (ix)^2 \left( \sum_{l=1}^L \frac{H_l^2}{2!} + \sum_{l,l':l>l'} H_l H_{l'} \right) \\
&\quad + (-ixH)\sum_{l=1}^L (ixH_l) + \frac{(-ixH)^2}{2!} \\
&= \frac{x^2}{2} \sum_{l,l':l>l'} [H_{l'}, H_l],
\end{aligned}
\end{equation}
which is a summation of $L(L-1)$ commutators. Since the commutators of Hermitian operators are always anti-Hermitian, this implies that the nested-commutator expansion of $F_{1,2}(x)$ in \autoref{eq:F12xComm} is compact enough since there is no Hermitian expansion terms in it.

For a common physical Hamiltonian with locality constraints, we can take advantage of the commutator-form expression like \autoref{eq:F12xComm}. For example, for an $n$-qubit lattice Hamiltonian with the form,
\begin{equation} \label{eq:HLattice}
H = \sum_{j=0}^{n-1} H_{j,j+1},
\end{equation}
where the summand $H_{j,j+1}$ acts on the $j$th and $(j+1)$th vertices. 
We can split the Hamiltonian to two components $H=A+B$ where 
$A := \sum_{j:\mr{even}} H_{j,j+1}, \quad B:= \sum_{j:\mr{odd}} H_{j,j+1}$ ,
so that the summands $H_{j,j+1}$ commute with each other in each component. We denote the norm of each Hamiltonian summand as 
\begin{equation} \label{eq:LatticeHnormbound}
\begin{aligned}
\Lambda = \max_j \| H_{j,j+1} \|, \quad \Lambda_1 &= \max_j \| H_{j,j+1} \|_1. 
\end{aligned}
\end{equation}

Now, suppose we estimate the $1$-norm of $F_{1,2}(x)$ of the lattice Hamiltonian based on \autoref{eq:F12xComm}, we can see that there are only $n$ nonzero terms: for any given Hamiltonian component $H_{j,j+1}$, only $H_{j-i,j}$ and $H_{j+1,j+2}$ do not commute with it. Then, the norm of $F_{1,2}(x)$ is bounded by
\begin{equation}
\|F_{1,2}(x)\|_1 \leq n\Lambda_1\frac{x^2}{2} = \mc{O}(n x^2).
\end{equation}
Comparing with the original bound in \autoref{eq:muF1sx}, $\|F_{1,2}(x)\|_1\leq \eta_2 = \mc{O}((\lambda x)^2) = \mc{O}(n^2 x^2)$, we improve the system-size-related factor $n$. 
The improved system-size $n$ dependence of the 1-norm $\|F_{1,2}(x)\|_1$ suggests a corresponding improvement in the system-size dependence of the gate complexity for the nested commutator algorithm.

\begin{figure*}[ht]
\centering
\includegraphics[width=0.85\linewidth]{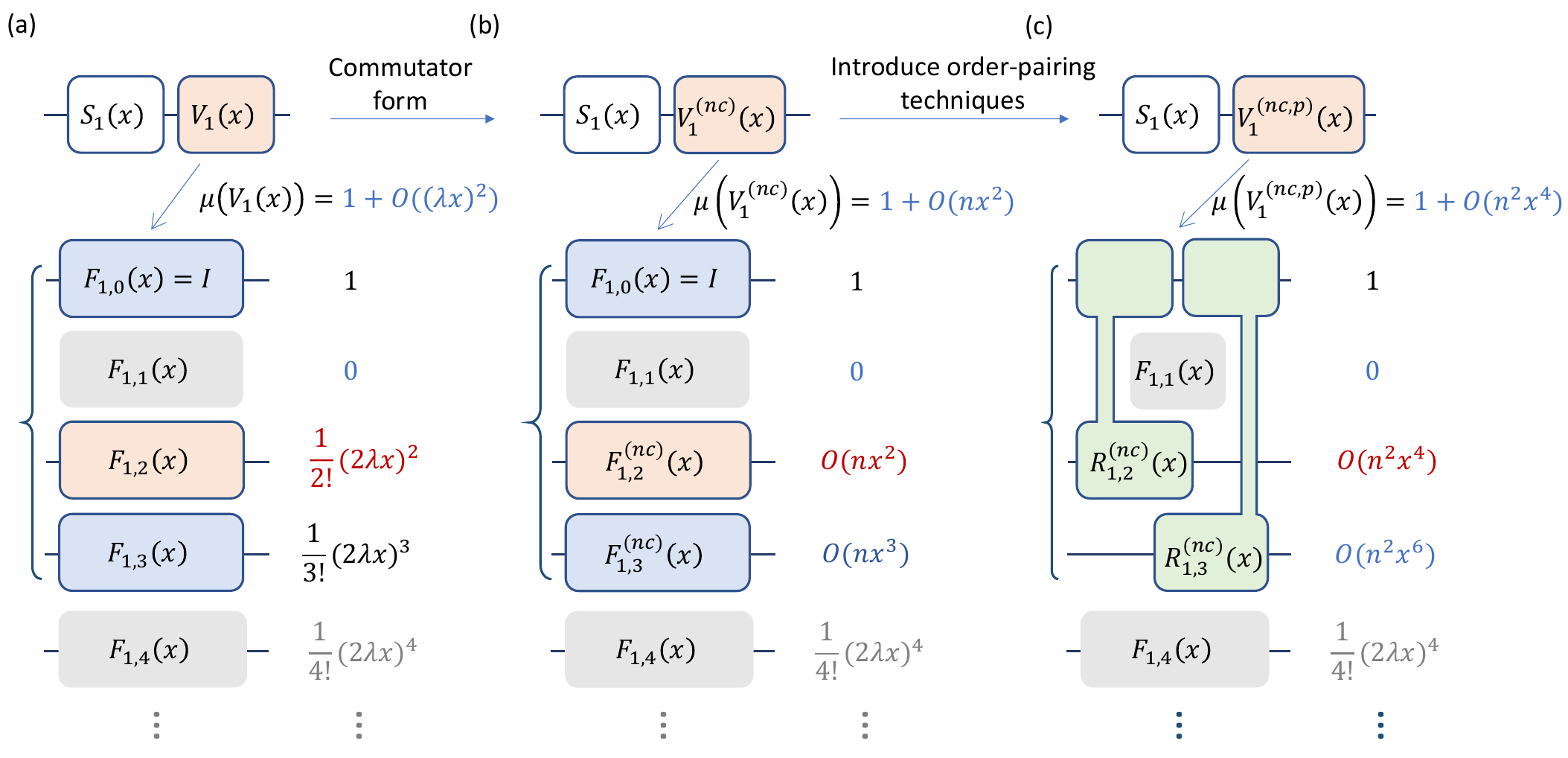}
\caption{Illustration of the idea of nested-commutator compensation (NCC) algorithm. We take the first-order algorithm $K=1$ as an example. (a) The Taylor-series expansion of the first-order Trotter remainder. We consider a truncation such that the higher-order terms with order $s\geq 2K+2$ will not be compensated. (b) We derive the nested-commutator form for the leading-order expansion terms $F_{1,2}^{(nc)}(x)$ and $F_{1,3}^{(nc)}(x)$, which own better system-size dependence. (c) Since $F^{(nc)}_{1,2}$ and $F^{(nc)}_{1,3}$ are anti-Hermitian, we introduce Euler's formula in \autoref{eq:Euler} to pair the leading-order terms to $F_{1,0}:=I$. This will double the order of $\mu_x-1$ and further suppress the $1$-norm $\mu$ of the overall evolution $U(t)$.}
\label{fig:NCidea}
\end{figure*}

Now, we are going to generalize the idea above. In \autoref{sec:NCC}, we show how to expand the second- and third-order terms of $V_1(x)$ as a summation of nested commutators,
\begin{equation} \label{eq:F1sxNCGen}
\begin{aligned}
F_{1,s}(x) &= \frac{x^s}{s!} C_{s-1} = \frac{(-ix)^s}{s!} \Big( \mc{A}_{\mr{com}}^{(s-1)}(H_L,...,H_1;H) \\
&\, - \sum_{l=1}^L \mc{A}_{\mr{com}}^{(s-1)}(H_L,...,H_{l+1};H_l) \Big), \quad s=2,3,
\end{aligned}
\end{equation}
where $\mc{A}_{\mr{com}}^{(s)}(A_L,...,A_1;B)$ is defined to be
\begin{equation} \label{eq:DefAcoms}
\begin{aligned}
\sum_{m_1+...+m_L=s} \binom{s}{m_1,...,m_L} \ad_{A_L}^{m_L} ... \ad_{A_1}^{m_1} B.
\end{aligned}
\end{equation}
We also use the adjoint notation $\ad_{A_L} ... \ad_{A_1} B:= [A_L, ...[ A_1, B]...]$. It is easy to check that the form of $F_{1,s}(x)$ ($s=2$ or $3$) in \autoref{eq:F1sxNCGen} is anti-Hermitian, which is consistent with the discussion in the paired Taylor-series algorithm in the previous section.
We can also generalize the method to the case of $K$th-order Trotter remainder, that is, to express the expansion terms $F_{K,s}$ of the $K$th-order Trotter remainder based on the nested commutators.

Based on the nested-commutator expansion, we propose the $K$th-order nested-commutator compensation (NCC) algorithm. 
As is illustrated in \autoref{fig:NCidea}, in the construction of NCC, we first utilize the nested-commutator forms of Trotter error terms from $K+1$ to $2K+1$ order, i.e., the leading-order terms; then we apply the order-pairing techniques similar to PTSC in \autoref{fig:PairedTSidea}(c) to further suppress the $1$-norm of $\mu_x$.
A key difference between NCC and PTSC algorithms is that in PTSC algorithms, we compensate the Trotter error $V_K(x)$ up to arbitrary order; while in NCC algorithms, we only compensate $V_K(x)$ for leading-order terms, which shrinks the error from $\mc \mc{O}(x^{K+1})$ to $\mc{O}(x^{2K+2})$ in one slice with the sampling cost $\mu=1+\mc{O}(  \frac{\|C_{K}\|_1}{(K+1)!} x^{2K+2})$.
The gate complexity estimation is then converted to the calculation of the $1$-norm of the commutator $\|C_K\|_1$. 
For instance, if we consider the $n$-qubit lattice Hamiltonian models in \autoref{eq:HLattice}, then we can prove that $\|C_{K}\|_1=\mc{O}(n)$.
We can then provide the following performance guarantee for the NCC algorithms.

\textbf{Theorem 2} (Informal, see \autoref{thm:NC} in \autoref{sec:NCC}) In a $K$th-order nested-commutator compensation (NCC) algorithm ($K=1$ or $2k$) with $n$-qubit lattice Hamiltonians, the gate complexity in a single round is $\mc{O}( n^{1+\frac{2}{2K+1}} t^{1+\frac{1}{2K+1}} \varepsilon^{-\frac{1}{2K+1}} )$. 

Compared to the performance of $K$th-order Trotter algorithm $\mc{O}((nt)^{1+1/K}\varepsilon^{-1/K})$~\cite{childs2019nearly}, we achieve $t$- and $\varepsilon$-dependence better than $2K$th-order Trotter using only $K$th-order Trotter formula with simple compensation gates of Pauli-rotation operators.
To generalize the result in Theorem~2, we also study the performance of Nested Commutator (NCC) algorithms when applied to a general Hamiltonian in \autoref{ssc:NestedCommSampling}.

\subsection{ Efficient random-sampling implementation }

A simple implementation of the Trotter-LCU algorithm in \autoref{fig:TrotterLCUidea}(b) or \autoref{fig:TrotterLCUidea}(c)  requires not only an easy-to-implement quantum circuit but also efficient classical random sampling of Pauli operators from the Trotter remainder $V_K(x)$. 
We now briefly discuss how to realize an efficient classical random sampling in PTSC and NCC algorithms, that is, with a space resource of $\mc{O}(\kappa_K K)$ and time resource of $\mc{O}(K(\log{L} + \log\kappa_K))$ where $L$ is the sparsity of the Hamiltonian.

A key idea for achieving efficient sampling is to use a multistage hierarchical sampling algorithm. Rather than fully expanding the Trotter remainder into a direct summation of Pauli operators, we structure the LCU formula into multiple layers. This allowss us to decompose the overall Pauli operator sampling process into a series of simpler, more manageable sampling tasks.

In the PTSC algorithm, the Trotter remainder in \autoref{eq:DTS1st} is derived by expanding the time-evolution $e^{-ix H}$ and each Hamiltonian summand term $e^{ix H_l}$ by Taylor series independently. As a result, the sampling can be done by first sampling the overall expansion order $s$, then sample the individual expansion order $r$ of Hamiltonian $H$ or the expansion order $\vec{r}:=[r_1,r_2,...,r_L]$ of the summands $\{H_l\}$. The sampling of $r$ and $\vec{r}$, following the analysis in \autoref{sec:PTSC}, can be done based on a multinomial distribution $\text{Mul}(r,\vec{r};\{\frac{1}{2},\frac{\vec{p}}{2}\};s)$ where $\vec{p}:=[p_1,p_2,...,p_L]$ denotes the normalized coefficient factor of $H$ defined in \autoref{eq:H}. For the sampled Hamiltonian $H$, we further sample the summands $H_l$ inside based on $\vec{p}$. We summarize the sampling algorithm in \autoref{fig:SamplingPTS}.

\begin{figure}[htbp]
\centering
\includegraphics[width=\linewidth]{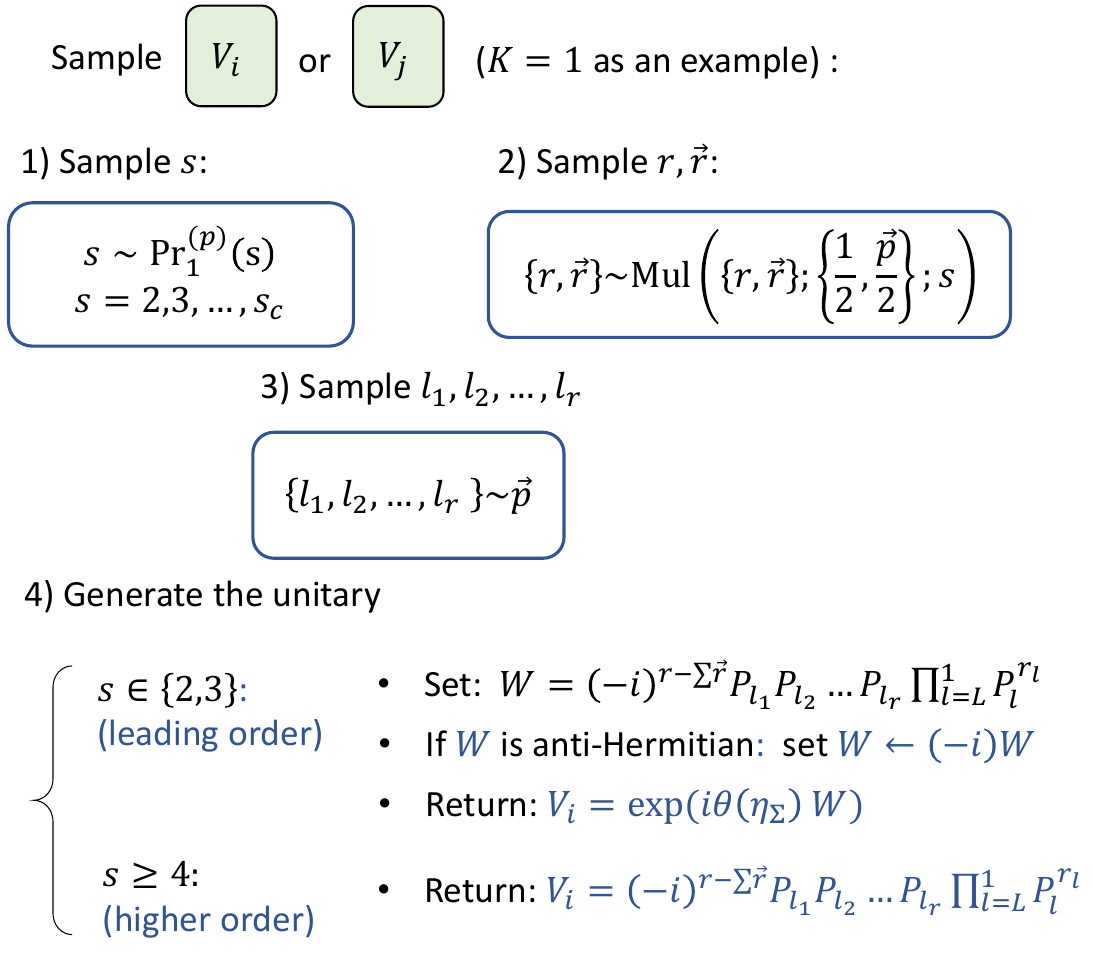}
\caption{The sampling procedure of $V_i$ or $V_j$ in \autoref{fig:TrotterLCUidea}(b,c) in the PTSC algorithm. We set $K=1$ as example. $\Pr_1^{(p)}(s)$ is a discrete probability distribution given in \autoref{eq:mu1_Pr1p}. $\eta_\Sigma:=\eta_2+\eta_3$ and $\eta_s:=\frac{(2\lambda x)^s}{s!}$. $\theta(y):=\tan^{-1}(y)$.}
 \label{fig:SamplingPTS}
\end{figure}

In the NCC algorithm, the Trotter remainder is expanded based on the adjoint operators. For example, for the lattice Hamiltonian $H=A+B$ in \autoref{eq:HLattice}, we can write the second- and third-order Trotter remainder of $V_1(x)$ as,
\begin{equation}
\begin{aligned}
F_{1,2}^{(nc)}(x) &= -\frac{x^2}{2!} \ad_A B, \\
F_{1,3}^{(nc)}(x) &= i\frac{x^3}{3!} (2 \ad_B\ad_A B + \ad_A^2 B).
\end{aligned}
\end{equation}
We can first sample the expansion order $s=2$ or $3$. If $s=3$, we then sample the specific commutator, i.e., $\ad_B\ad_A B$ or $\ad_A^2 B$. For the given commutator, for example, $\ad_B\ad_A B = [B, [A, B]]$, we first randomly sample a summand $H_{j,j+1}$ for the rightmost $B$ as the starting point of the adjoint operator. The action of the subsequent $\ad_A$
and $\ad_B$ will enlarge the support of $H_{j,j+1}$, but within a ``light-cone'' region shown in \autoref{fig:AdjExpandSamp}. We then sample the Hamiltonian summand $H_{j_1, j_1+1}$ and $H_{j_2, j_2+1}$ for the adjoint operators $\ad_B$ and $\ad_A$, but within the light-cone region. This will ensure our sampling to be efficient and with nested-commutator scaling.

\begin{figure}[htbp]
\centering
\includegraphics[width=\columnwidth]{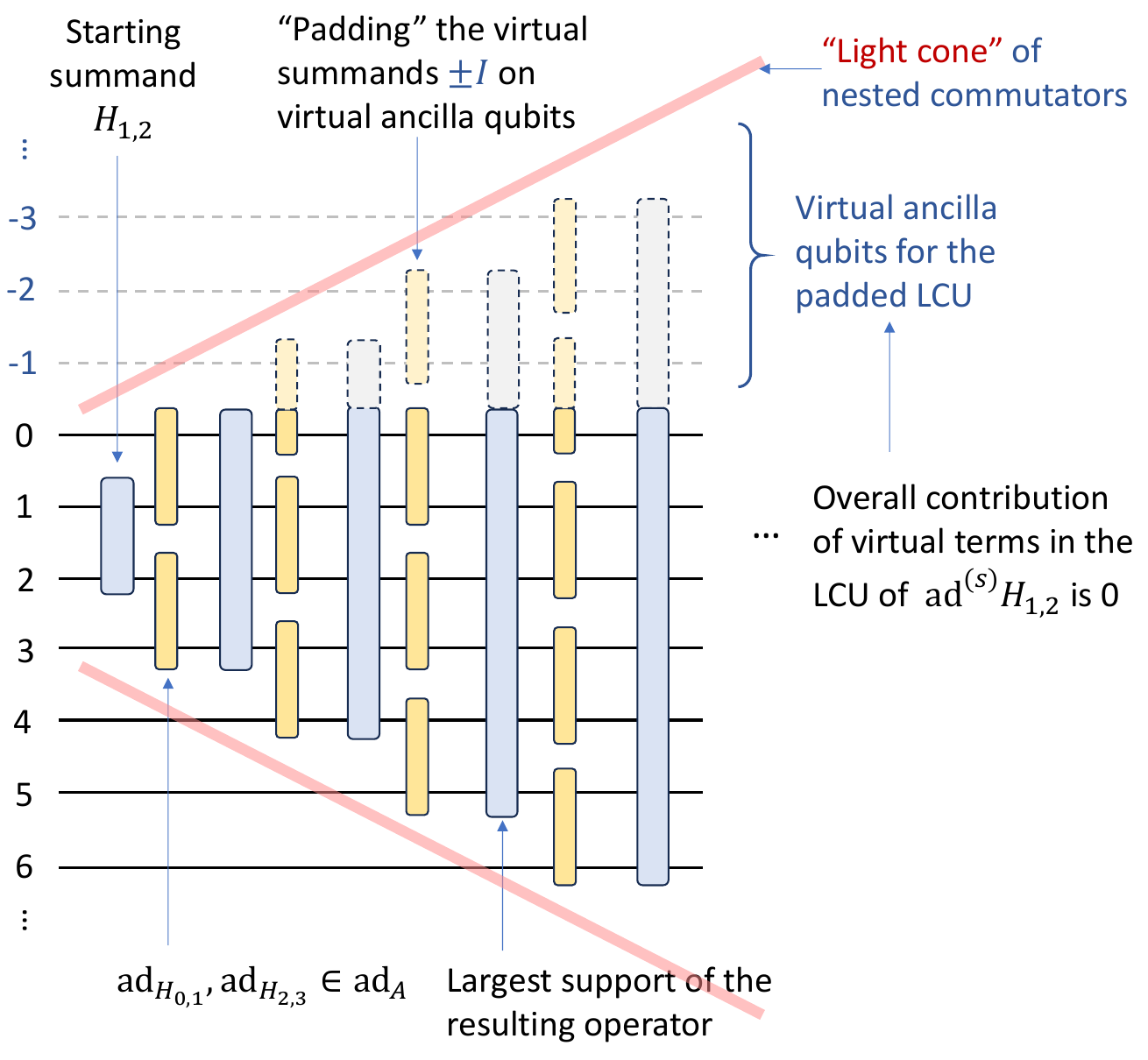} 
\caption{ 
Illustration of the generation of elementary commutators in a given commutator form $\ad_B\ad_A\ad_B\ad_A H_{j,j+1}$. Here we set $j=1$. Start from $H_{1,2}$, we sequentially apply the adjoint operation of elementary summands $H_{j_l,j_l+1}$ (the yellow blocks) on the operator. The blue blocks indicates the largest possible support of the resulting operators. When the support reach the boundary of the Hamiltonian, we introduce extra virtual ancillary qubits to construct padded LCU. The operators on the virtual qubits are always $\pm I$, so that we only need to implement a phase gate on the control qubit without the need to perform operations on the virtual qubits.}
\label{fig:AdjExpandSamp}
\end{figure}

A problem of sampling the Hamiltonian summands $H_{j,j+1}$ in the commutator is that, since the Hamiltonian $H$ may be nonhomogeneous, the $1$-norm of $H_{j,j+1}$ with different $j$ may be different. This will complicate the sampling algorithm, since we need to first calculate the $1$-norm of all the elementary commutators with the form of
\begin{equation} \label{eq:commutator_summand}
\ad_{H_{j_2,j_2+1}} \ad_{H_{j_1,j_1+1}} H_{j,j+1},
\end{equation}
for all $j, j_1$ and $j_2$. When the Trotter order $K$ and the expansion order $s=K+1,...,2K+1$ gets larger, the number of elementary adjoint operators will increase exponentially. Consequently, we cannot estimate the $1$-norm of all the elementary commutator.

To solve this problem but still keep the advantage of the NCC algorithm, we introduce the following Hamiltonian ``padding'' technique to ensure all the elementary nested commutator own the same $1$-norm. Consider a Hamiltonian summand $H_{j,j+1}$, which can be expanded to some Pauli operators,
\begin{equation}
H_{j,j+1} = \sum_{\omega} \alpha_j^{(\omega)} P_{j,j+1}^{(\omega)},
\end{equation}
where $\alpha_j^{(\omega)}$ is a positive number and $P_{j,j+1}^{(\omega)}$ is a normalized Pauli operator whose support is on qubit $j$ and $j+1$. 
Recall that $\Lambda_1:= \max_j \|H_{j,j+1}\|_1$ and $\|H_{j,j+1}\|_1:= \sum_{\omega} \alpha_j^{(\omega)}$.
When $\|H_{j,j+1}\|_1$ is smaller than $\Lambda_1$, we add extra trivial terms $\pm I$ in the Pauli decomposition of $H_{j,j+1}$,
\begin{equation} \label{eq:PadHj}
\begin{aligned}
\bar{H}_{j,j+1} &= \sum_{\omega} \alpha_j^{(\omega)} P_{j,j+1}^{(\omega)} + \frac{\delta\Lambda_1}{2} I + \frac{\delta\Lambda_1}{2} (-I) \\
&= \Lambda_1 \sum_{\omega} \bar{p}_j^{(\omega)} P_{j,j+1}^{(\omega)},
\end{aligned}
\end{equation}
where $\delta\Lambda_1:= \Lambda_1 - \|H_{j,j+1}\|_1$. \autoref{eq:PadHj} holds naturally, but now with a manually predetermined $1$-norm value $\Lambda_1$. Similarly, we can pad all the elementary commutators with the form 
\begin{equation}
\ad_{H_{j_1,j_1+1}} H_{j,j+1} = H_{j_1,j_1+1} H_{j,j+1} - H_{j,j+1} H_{j_1,j_1+1}, 
\end{equation}
so that their $1$-norms are all $2\Lambda_1^2$. In this way, we ignore the commutator relationship between $H_{j,j+1}$ and $H_{j_1,j_1+1}$ as long as they are in the light-cone region.

After the padding procedure described above, all elementary nested commutators of the same order will have the same $1$-norm. This property allows us to uniformly sample these commutators: the starting summand $H_{j,j+1}$ is sampled uniformly from those in $B$, and the subsequent $H_{j_1,j_1+1}$ and $H_{j_2,j_2+1}$ are sampled uniformly within the light-cone region.  
However, when the starting summand $H_{j,j+1}$ is near the boundary, applying a few adjoint operators may cause it to touch the boundary. This reduces the number of possible elementary nested commutators compared to those starting from the center, resulting in different $1$-norms for $\text{ad}_A H_{j,j+1}$ depending on $j$. This complicates the sampling of the starting summand $H_{j,j+1}$.  
To resolve this issue and ensure uniform sampling of $H_{j,j+1}$, we introduce virtual qubits at the boundary, as illustrated in \autoref{fig:AdjExpandSamp}, and pad the virtual qubits with $0$-summed $\pm I$ terms. 
Since we only perform $\pm I$ operation on the virtual qubits, we do not need to introduce it in the real experiments.

We remark that, our padding method preserve the locality structure. As a result, the performance guarantee in \autoref{thm:NC} still holds. We summarize the sampling algorithm in \autoref{fig:SamplingNC}. If we consider the Heisenberg Hamiltonian 
\begin{equation} \label{eq:HeisenbergH}
H = \sum_{i} \vec \sigma_i \vec \sigma_{i+1} +  \sum_{i} Z_i,
\end{equation}
as an example, where  $\vec \sigma_i:= (X_i, Y_i, Z_i)$ is the vector of Pauli operators on the $i$th qubit, we can define the summand $H_{j,j+1}$ to be
\begin{equation}
H_{j,j+1} = 4 \left( \frac{1}{4} X_j X_{j+1} + \frac{1}{4} Y_j Y_{j+1} + \frac{1}{4} Z_j Z_{j+1} + \frac{1}{4} Z_j \right).
\end{equation}
In this case, $\Lambda_1 = \|H_{j,j+1}\| = 4$. The probability distribution $\bar{p}_j^{(\omega)}$ in \autoref{fig:SamplingNC} is to uniformly sample the $XX$, $YY$, $ZZ$ and $ZI$ term. As a demonstration, we explicitly present the algorithm to sample the Pauli-rotation operator in the first-order NCC algortihm for the Heisenberg Hamiltonian in \autoref{eq:HeisenbergH} in Algorithm~\ref{Alg:Demonstration_1stNCC}.

\begin{figure}
\begin{algorithm}[H]
\caption{Demonstration: sampling of $V_i$ or $V_j$ of first-order NCC algorithm for the Heisenberg Hamiltonian in \autoref{eq:HeisenbergH}}
\label{Alg:Demonstration_1stNCC}
\begin{algorithmic}[1]
    \Require
    An $n$-qubit Heisenberg Hamiltonian $H$ in \autoref{eq:HeisenbergH}; unit evolution time $0<x<1$ for each Trotter segment; 
    \Ensure
    Sampling of a Pauli-rotation operator $V_i$ from the Trotter remainder $\tilde{V}_1^{(nc)}(x)$.
    \State Sample the expansion order $s\in\{2,3\}$ with the probability $\{\frac{1}{1+24x}, \frac{24x}{1+24x}\}$.
    \If{$s=3$}
        \State Sample the adjoint form $\ad^{(2)}B$ from $\{\ad_B\ad_A B, \ad_A^2 B\}$ with the probability $\{1/3, 2/3\}$.
    \EndIf
    \State Sample the starting index $j$ uniformly from all odd indices in $0,...,n-1$. Sample the Pauli operator $P_{j,j+1}$ uniformly from $\{XX,YY,ZZ,ZI\}$ on qubit $j$ and $j+1$.
    \State Set $W:= P_{j,j+1}$.
    \State Sample the adjoint index $j_1$ uniformly from $\{j-1, j+1\}$. Sample the Pauli operator $P_{j_1,j_1+1}$ uniformly from $\{XX,YY,ZZ,ZI\}$ on qubit $j_1$ and $j_1+1$. Sample the multiplication order $b_1$ uniformly from $\{0,1\}$. \footnote{If $j-1$ or $j+1$ exceeds the index range $0,...,n-1$, we pad extra virtual qubits similar to \autoref{fig:AdjExpandSamp}.}
    \If{$b_1=0$}
        \State Set $W:= P_{j_1,j_1+1} W$.
    \Else \Comment{$b_1=1$}
        \State Set $W:= -W P_{j_1,j_1+1}$.
    \EndIf
    \If{$s=3$}
        \If{$\ad^{(2)}B = \ad_B\ad_A B$}
            \State Sample the adjoint index $j_2$ uniformly from $\{j-2,j,j+2\}$. Sample the Pauli operator $P_{j_2,j_2+1}$ uniformly from $\{XX,YY,ZZ,ZI\}$ on qubit $j_2$ and $j_2+1$. Sample the multiplication order $b_2$ uniformly from $\{0,1\}$.
        \Else  \Comment{$\ad^{(2)}B = \ad_A^2 B$}
            \State Sample the adjoint index $j_2$ uniformly from $\{j-3,j-1,j+1\}$. Sample the Pauli operator $P_{j_2,j_2+1}$ uniformly from $\{XX,YY,ZZ,ZI\}$ on qubit $j_2$ and $j_2+1$. Sample the multiplication order $b_2$ uniformly from $\{0,1\}$.
        \EndIf
        \If{$b_2=0$}
            \State Set $W:= P_{j_2,j_2+1} W$.
        \Else \Comment{$b_2=1$}
            \State Set $W:= -W P_{j_2,j_2+1}$.
        \EndIf
    \EndIf
    \State Return $V_i:= \exp\left( i\theta W \right)$ as the sampled Pauli-rotation operator. Here $\theta:= \tan^{-1}(16n x^2(1+24x))$.
\end{algorithmic}
\end{algorithm}
\end{figure}

\begin{figure}[htbp]
\centering
\includegraphics[width=\linewidth]{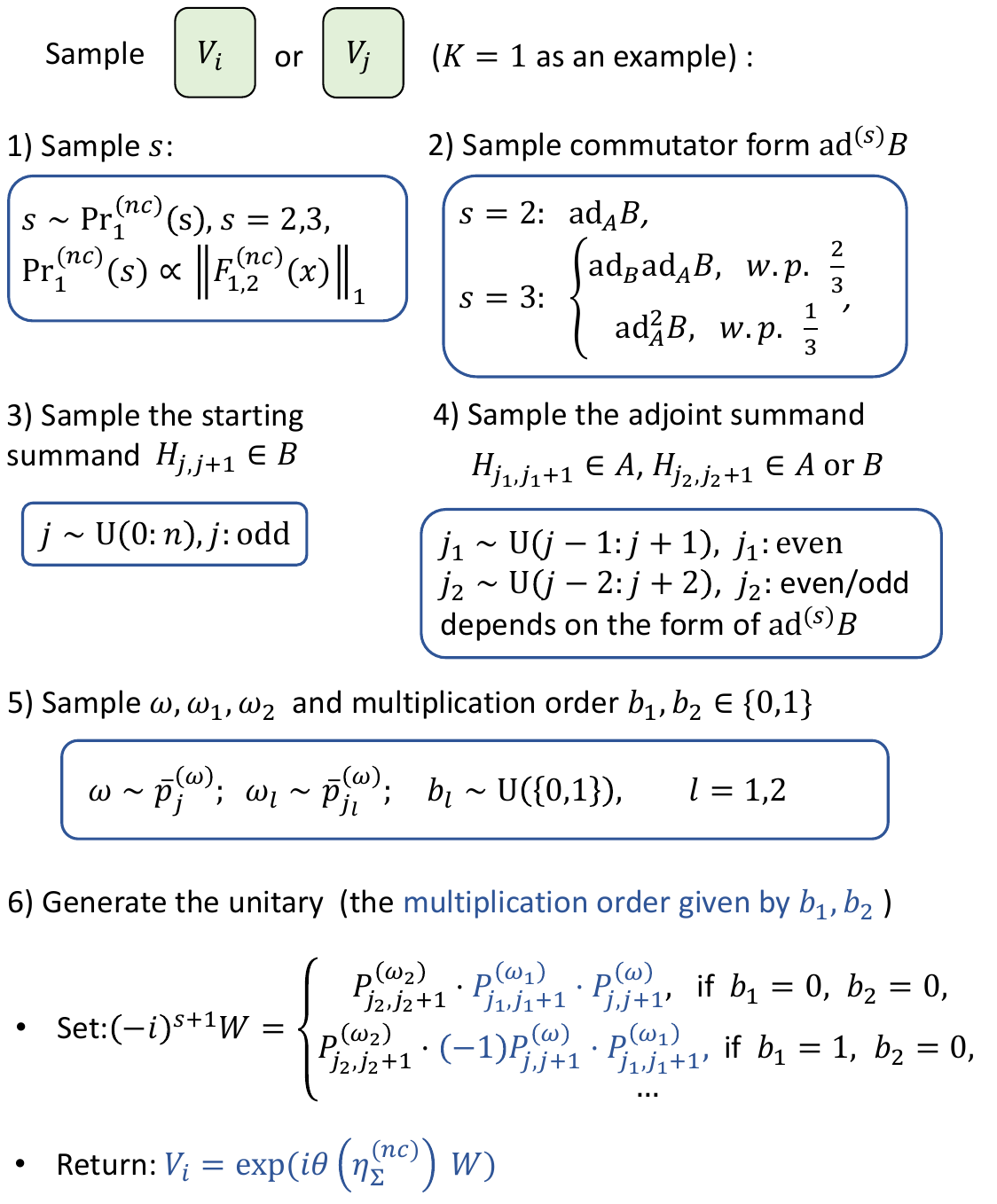}
\caption{The sampling procedure of $V_i$ or $V_j$ in \autoref{fig:TrotterLCUidea}(b) in the NCC algorithm. We set $K=1$ as example. The notation $l:l+k$ indicates the number array of $l,l+1,...,l+k$. $b_l$ determines whether to multiply $P_{j_l,j_l+1}^{(\omega_l)}$ to the left or the right side of the current Pauli operator. $\mr{U}(\bullet)$ refers to a uniform distribution in the set. $\eta_\Sigma^{(nc)}:= \|F_{1.2}^{(nc)}\|_1 + \|F_{1.3}^{(nc)}\|_1$. $\theta(y):=\tan^{-1}(y)$. The probability $\bar{p}_j^{(\omega)}$ is given by the Hamiltonian information in \autoref{eq:PadHj}.}
 \label{fig:SamplingNC}
\end{figure}

As a final remark, the sampling procedure in both PTSC and NCC algorithms are independent of the implementation of the quantum circuit and the measurement outcome. Thanks to this property, we can perform the classical sampling during the quantum circuit implementation or even generate the sampled Pauli matrices before the implementation of the quantum circuits.

\subsection{Performance comparison}

In \autoref{tab:SumComp}, we compare the implementation complexity and the gate complexity in a single round of experiment of the $0$th-order PTSC, $K$th-order PTSC, and $K$th-order NCC algorithms to previous Hamiltonian simulation algorithms. For a fair comparison, we set the 1-norm of all LCU formulas $\mu$ to be constant. 
In this case, the sample complexity of PTSC and NCC algorithms incurs a $\mu^4$ overhead compared to standard sampling from the $K$th-order Trotter or post-Trotter algorithms.

We show that by inserting a few randomly sampled Pauli-rotation gates after each Trotter segment, as illustrated in \autoref{fig:illustration}, both PTSC and NCC achieve improved accuracy and time dependence. The gate counts for PTSC exhibit logarithmic dependence on accuracy, $\log(1/\varepsilon)$, while the NCC gate counts show improved system-size dependence.

\begin{table*}[htbp]

\centering
\begin{tabular}{c|c|c|c|c}
\hline
\hline
Algorithm & Implementation hardness & Accuracy & Size scaling (lattice Hamiltonian) & Time dependence \\
 \hline
$K$th-order Trotter~\cite{suzuki1990fractal} & Easy & $\mc{O}(\varepsilon^{-1/K})$ &$\mc{O}(n^{1+ \frac{1}{K} })$ & $\mc{O}(t^{1+\frac{1}{K}})$ \\
\hline
\makecell[c]{Post-Trotter~\cite{berry2015simulating,low2019Hamiltonian} } &  Hard  & $\mc{O}(\tilde{\log}(1/\varepsilon))$  & $\mc{O}(n^2)$ & $\mc{O}(t)$ \\
\hline
\hline
\makecell[c]{$0$th-order PTSC} & Easy & $\mc{O}(\tilde{\log}(1/\varepsilon))$  & $\mc{O}(n^2)$ & $\mc{O}(t^2)$ \\
\hline
\makecell[c]{$K$th-order PTSC} & Easy & $\mc{O}(\tilde{\log}(1/\varepsilon))$  & $\mc{O}(n^{2+\frac{1}{2K+1} } )$ & $\mc{O}(t^{1+\frac{1}{2K+1}})$ \\
\hline
\makecell[c]{$K$th-order NCC} & Easy & $\mc{O}(\varepsilon^{-1/(2K+1)})$  & $\mc{O}(n^{1+\frac{2}{2K+1} })$ & $\mc{O}(t^{1+\frac{1}{2K+1}})$ \\
\hline
\hline
\end{tabular} 
\caption{Comparison of the implementation hardness and gate complexity in a single round of the circuit for Trotter-LCU methods versus previous algorithms. Here, $\tilde{\log}(x):=\log(x)/\log\log(x)$. The (implementation) hardness refers to whether one needs to implement multicontrolled gates with plenty of ancillary qubits. In the comparison for the system-size scaling of lattice Hamiltonians, we use the fact that $\lambda=\mc{O}(n)$ and $L=\mc{O}(n)$.
}  
\label{tab:SumComp}
\end{table*}

To demonstrate how the Trotter-LCU algorithms can help reduce gate costs in practical scenarios, we estimate the single-shot gate count of random-sampling Trotter-LCU algorithms and compare it with the state-of-the-art Trotter algorithm, i.e., fourth-order Trotter~\cite{suzuki1991general,childs2019fasterquantum,childs2021theory}. 
To ensure a fair comparison, we set $1$-norm of the LCU formula to be $\mu=2$. Based on \autoref{prop:randomLCU}, this implies that the random-sampling Trotter-LCU algorithm will require an additional factor of $16$ in the sample number to estimate the properties of the target state $U\rho U^\dag$ where $U:=e^{iHt}$ to a given precision compared to the normal Trotter or coherent implementation of LCU algorithms.

We compile their quantum circuits to $\mc{CNOT}$ gates, single-qubit Clifford gates, and single-qubit $Z$-axis rotation gates $R_z(\theta)=e^{i\theta Z}$.
Here, we mainly compare the number of $R_z(\theta)$ gates since they are the most resource-consuming part on a fault-tolerant quantum computer~\cite{Litinski2019gameofsurfacecodes}. 
The $\mc{CNOT}$ gate-number comparison results can be found in Appendix~\ref{sec:AppNumerics}, which is similar to the $R_z(\theta)$ gate comparison.
We also compare the gate counts of the Trotter-LCU algorithms to the coherent-implementation of LCU~\cite{berry2014exponential,berry2015simulating} and QSP~\cite{low2017optimal,low2019Hamiltonian} in Appendix~\ref{sec:AppNumerics}.

In the first comparison, we consider the simulation of generic Hamiltonians without the usage of commutator information, for which we choose the $2$-local Hamiltonian, $H = \sum_{i,j} X_iX_{j} + \sum_i Z_i$ where $X_i$ and $Z_i$ are the Pauli matrices on the $i$th qubit. 
\autoref{fig:L_NC_abs}(a,b) show the gate counts for the fourth-order Trotter formula and the PTSC Trotter-LCU algorithms with different orders with an increasing time $t$ and increasing system size $n$, respectively.
The gate counting method for fourth-order Trotter with random permutation is based on the analytical bounds in Ref.~\cite{childs2019fasterquantum}.

\begin{figure*}[htbp]
\centering
\includegraphics[width=1\textwidth]{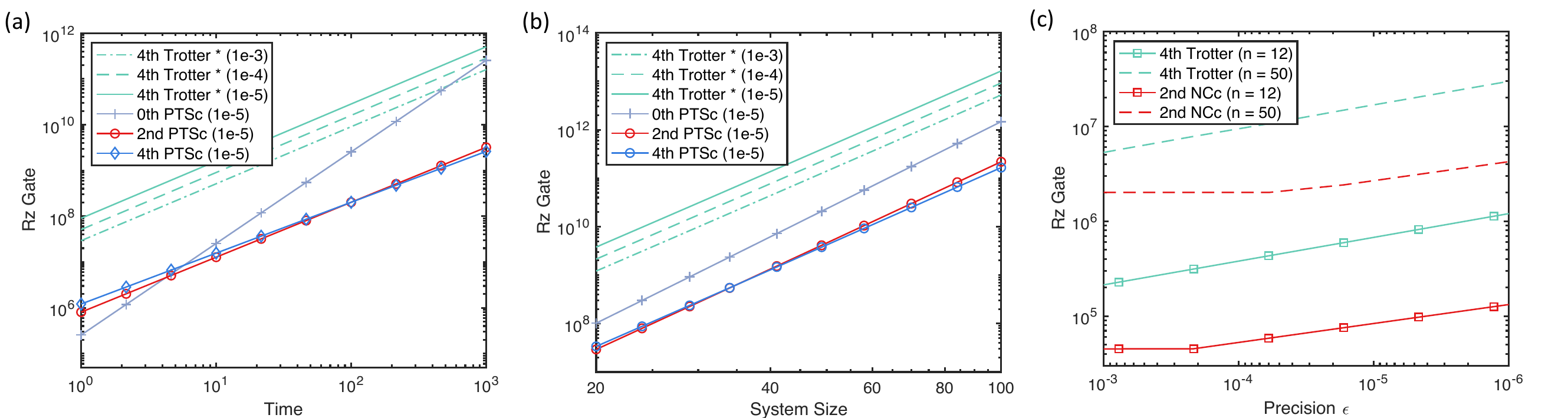}
\caption{ 
Non-Clifford $R_z(\theta)$ gate-number estimation for simulating real-time dynamics with Trotter and Trotter-LCU algorithms.
(a) and (b) show the $R_z(\theta)$ gate counts when simulating the Hamiltonian $H = \sum_{i,j} X_i X_j + \sum_i Z_i$.
(a) shows the $R_z(\theta)$ gate count with an increasing time and fixed system-size $n = 20$.
(b) shows the $R_z(\theta)$ gate count with an increasing system size with the time $t = n$.
We list the performance of zeroth, second, and fourth-order PTSC algorithms with $\mu=2$.
The fourth-order Trotter analytical bound is from Ref.~\cite{childs2019fasterquantum}.
(c) The $R_z(\theta)$ gate count for the nearest-interaction Heisenberg model using the nested-commutator bound with $n=12$ and $n=50$. The fourth-order Trotter commutator bound is from Ref.~\cite{childs2021theory}.
}
\label{fig:L_NC_abs}
\end{figure*}

From \autoref{fig:L_NC_abs}(a,b), we can clearly see the advantage of composition of Trotter and LCU formulas: if we do not use LCU and merely apply Trotter formulas, the gate resource of fourth-order Trotter suffers from a large overhead that is 2 orders of magnitudes larger than the PTSC algorithms. Moreover, if we increase $\varepsilon$ from $10^{-3}$ to $10^{-5}$, we can see a clear increase of the gate resources for the Trotter algorithm. For the PTSC algorithms, however, the gate number is almost not affected by $\varepsilon$ since they enjoy a logarithmic $\varepsilon$-dependence.

On the other hand, if we do not use Trotter formula, the $0$th-order PTSC algorithm owns a quadratically worse $t$-dependence ($\mc{O}(t^2)$) than the fourth-order Trotter ($\mc{O}(t^{1.25})$), second-order PTSC ($\mc{O}(t^{1.2})$) and fourth-order PTSC ($\mc{O}(t^{1.11})$) algorithms in \autoref{fig:L_NC_abs}(a).
For a short-time evolution $t=n$, the system-size dependence of second- or fourth-order PTSC in \autoref{fig:L_NC_abs}(b) outperforms $0$th-order PTSC algorithm. In the case when long-time Hamiltonian simulation is required, for example, $t$ should be set to $10^3$ for the phase estimation~\cite{campbell2019random}, the advantage of $2$nd or fourth-order PTSC to $0$th-order PTSC will be more obvious.
The composition of Trotter and LCU formulas enables $2$nd or fourth-order PTSC to enjoy good $t$ and $\varepsilon$ dependence and small gate-resource overhead simultaneously.

Next, we compare the gate count when simulating the lattice models,  where the commutator analysis will help remarkably reduce the gate count. We consider the Heisenberg Hamiltonian $H =  \sum_{i} \vec \sigma_i \vec \sigma_{i+1} +  \sum_{i} Z_i$ using the nested-commutator bounds.
In \autoref{fig:L_NC_abs}(c), we choose $n=t=12$ and $50$ and show the gate number with respect to the accuracy requirement $\varepsilon$. The fourth-order Trotter error analysis is based on the nested-commutator bound (Proposition~M.1 in Ref.~\cite{childs2021theory}), which is currently the tightest Trotter error analysis. 
The performance of our second-order NCC algorithm is based on the analytical bound in Sec.~H in the Appendix. 
We mainly present results for the second-order NCC algorithm due to its simplicity and leave precise higher-order NCC gate count analysis for future study. 

From \autoref{fig:L_NC_abs}(c), we can see that, while enjoying near-optimal system-size scaling similar to the fourth-order Trotter algorithm which is currently the best one for lattice Hamiltonians~\cite{childs2021theory,childs2019nearly}, the second-order NCC algorithm shows better accuracy dependence than fourth-order Trotter algorithm. Particularly, using the same gate number as the fourth-order Trotter, we are able to achieve a 3 to 4 orders of magnitudes higher accuracy $\varepsilon$.

\section{PRELIMINARIES} \label{sec:preliminaries}

In this section, we review the Hamiltonian simulation algorithms based on Trotter and LCU formulas. 

In all the Hamiltonian simulation algorithms discussed in this work, we divide the real-time evolution $U(t)$ into $\nu$ segments,
\begin{equation}
U(t) = e^{-iHt} = (U(x))^\nu = \left( e^{-iHx} \right)^\nu,
\end{equation}
with $x:=t/\nu$, and consider the construction of each small segment $U(x)$. Without loss of generality, we assume $x>0$.

\subsection{Trotter formulas} \label{ssc:Trotter}

The most natural way to approximate $U(x)$ is to apply the Lie-Trotter-Suzuki formulas~\cite{suzuki1990fractal,suzuki1991general}. Hereafter, we refer to them as Trotter formulas. 
The first-order Trotter formula is
\begin{equation} \label{eq:S1x}
\begin{aligned}
S_1(x) = \prod_{l=1}^L e^{-ix H_l} = e^{-ix H_L}... e^{-ix H_2}e^{-ix H_1} =: \prod^{\leftarrow} e^{-ix \vec{H}}.
\end{aligned}
\end{equation}
Here, $\vec{H}:=[H_1, H_2, ..., H_L]$. In \autoref{eq:S1x}, we simplify the notation of the sequential products from $l=1$ to $L$ with the same form, using the arrow to denote the ascending direction of the dummy index $l$. 
The Hermitian conjugation of $S_1(x)$ can similarly be written as $S_1(x)^\dag = \overset{\rightarrow}{\prod} e^{ix \vec{H}}$.

The second-order Trotter formula is
\begin{equation} \label{eq:S2x}
\begin{aligned}
S_2(x) &= D_2(-x)^\dag D_2(x) = S_1(-\frac{x}{2})^\dag S_1(\frac{x}{2}) \\
&= \prod^{\rightarrow} e^{-i\frac{1}{2} x \vec{H}} \prod^{\leftarrow} e^{-i\frac{1}{2} x \vec{H}},
\end{aligned}
\end{equation}
where $D_2(x):= S_1(\frac{x}{2})$. We have $S_2(-x)^\dag = S_2(x)$.

The general $2k$th-order Trotter formula is~\cite{suzuki1991general}
\begin{equation} \label{eq:S2kx}
\begin{aligned}
S_{2k}(x) &= D_{2k}(-x)^\dag D_{2k}(x) \\
&= [S_{2k-2}(u_k x)]^2 S_{2k-2}((1-4u_k)x)[S_{2k-2}(u_k x)]^2,
\end{aligned}
\end{equation}
with $u_k:=1/(4-4^{1/(2k-1)})$ for $k\geq 1$. The operator $D_{2k}(x)$ is defined recursively,
\begin{equation} \label{eq:D2kx}
D_{2k}(x) := D_{2k-2}( (1-4u_k)x  ) S_{2k-2}(u_k x)^2.
\end{equation}
By induction from \autoref{eq:S2kx} and \autoref{eq:D2kx} we can show that $S_{2k}(-x)^\dag = S_{2k}(x)$.

We also denote the zeroth-order Trotter formula to be $S_0(x) = I$ for consistency. We denote the multiplicative remainder of the Trotter formulas as
\begin{equation} \label{eq:Vx}
\begin{aligned}
V_K(x) &= U(x) S_K(x)^\dag,
\end{aligned}
\end{equation}
for $K=0,1,2k$. In what follows, we name this Trotter error term $V_K(x)$ as the $K$th-order Trotter remainder to avoid ambiguity to other error effects such as truncation error.

In Ref.~\cite{suzuki1990fractal}, Suzuki proves the following order condition for the Trotter formulas,
\begin{equation} \label{eq:suzuki}
S_K(x) = U(x) + \mc{O}(x^{K+1}) = e^{-iHx + \mc{O}(x^{K+1})},
\end{equation}
for $K=1$ or even positive $K$. As a result, the remainder $V_K(x)$ will only contain the terms of $x^{q}$ with $q\geq K+1$. We will use the following order condition of Trotter formulas in the later discussion.

\begin{lemma}[Order condition of Trotter formulas, Theorem~4 in \cite{childs2021theory}] \label{lem:TrotterOrder}
Let $H$ be an Hermitian operator and $U(x):=e^{-iHx}$ to be the time-evolution operator with $x\in\mbb{R}$. For the Trotter formula $S_K(x)$ defined in \autoref{eq:S1x} and \autoref{eq:S2kx}, where $K=1$ or positive even number, we have the following
\begin{enumerate}
\item Additive error: $A_K(x):= S_K(x) - U(x) = \mc{O}(x^{K+1})$,
\item Multiplicative error: $M_K(x):= U(x)S_K(x)^\dag - I = \mc{O}(x^{K+1})$,
\item For the exponentiated error $E_K(x)$ such that $S_K(x) = \mc{T}\exp\left( \int_0^x d\tau (-iH + E_K(\tau) )\right)$, we have $E_K(x)=\mc{O}(x^K)$.
\end{enumerate}
\end{lemma}

\subsection{LCU formulas} \label{ssc:LCU}

Instead of decomposing $U(x)$ as a product of elementary unitaries, another way is to decompose $U(x)$ to a summation of elementary unitaries. We now provide a formal definition of a LCU formula.

\begin{definition}[\cite{childs2012Hamiltonian}] \label{def:LCU}
A $(\mu,\varepsilon)$ (LCU) formula of an operator $V$ is defined to be 
\begin{equation} \label{eq:LCU}
\tilde{V} = \sum_{i=0}^{\Gamma-1} c_i V_i = \mu \sum_{i=0}^{\Gamma-1} \Pr(i) V_i,
\end{equation}
such that the spectral norm distance $\|V-\tilde{V}\|\leq \varepsilon$. Here, $\mu>0$ is the $l_1$-norm of the coefficient vector, $\Pr(i)$ is a probability distribution over different unitaries $V_i$, and $\{V_i\}_{i=0}^{\Gamma-1}$ is a set of unitaries. Here, we assume $c_i>0$ for all $i$ and absorb the phase into the unitaries $V_i$. We call $\mu$ the $1$-norm of this $(\mu,\varepsilon)$-LCU formula. 
\end{definition}
In what follows, we define the $1$-norm $\|V\|_1$ of an operator $V$ to be its smallest $1$-norm over all possible $(\mu,0)$-LCU formulas for $V$. Note that, $\|U\|_1=1$ for any unitary $U$. One can easily check the validity of this norm definition. 

We may consider two ways to implement the LCU formula of the operator $V$. 
In the first way, we coherently implement the LCU formula by introducing an ancillary system $A$ with the dimension $\Gamma$ which costs $\lceil \log_2\Gamma \rceil$ qubits. For the LCU lemma defined in \autoref{eq:LCU}, we define the amplitude-encoding unitary $U_{AE}$ and the select gate $\mr{Sel}(\tilde{V})$ to be,
\begin{equation} \label{eq:UAE_SelV}
\begin{aligned}
U_{AE}\ket{0}_A &:= \sum_{i=0}^{\Gamma-1} \sqrt{\Pr(i)} \ket{i}, \\
\mr{Sel}(\tilde{V}) &:= \sum_{i=0}^{\Gamma-1} \ket{i}\bra{i}\otimes V_i.
\end{aligned}
\end{equation}
Then, the following controlled-gate $W_{AB}$ acting on ancillary $A$ and system $B$ defines a way to realize LCU coherently when we prepare $\ket{0}$ on $A$ and measure it on computational basis to get $\ket{0}$,
\begin{equation} \label{eq:W}
W := (U_{AE}^\dag \otimes I) \mr{Sel}(\tilde{V}) (U_{AE} \otimes I).
\end{equation}
More precisely, if we set the initial state on $B$ to be $\ket{\psi}$, then we have
\begin{equation} \label{eq:WAB}
W\ket{0}_A\ket{\psi}_B = \frac{1}{\mu} \ket{0}_A \left(\tilde{V}\ket{\psi}\right)_B + \sqrt{1 - \frac{1}{\mu^2}} \ket{\Phi_\perp}_{AB},
\end{equation}
where $\ket{\Phi_\perp}_{AB}$ is a state whose ancillary state is supported in the subspace orthogonal to $\ket{0}$.

If we perform computation-basis measurement directly on $A$, the successful probability to obtain $\tilde{V}\ket{\psi}$ is $\frac{1}{\mu^2}$. To boost the successful probability to nearly deterministic, we can introduce the amplitude amplification techniques~\cite{berry2014exponential}. To this end, we consider the following isometry
\begin{equation} \label{eq:isoLCU}
\mc{V}:= -W R W^\dag R W \left( \ket{0}\otimes I \right)_{AB}.
\end{equation}
Here, $R$ is a reflection over $\ket{0}$ on the system $A$,
\begin{equation} \label{eq:R}
R = I - 2\ket{0}_A\bra{0}.
\end{equation}

Consider the case of $\mu=2$ and $\varepsilon$ is small for the LCU formula. If $\tilde{V}$ is a unitary, then we have $(\bra{0}\otimes I)\mc{V} = \tilde{V}$. For a general $\tilde{V}$, we can verify the resulting operator
\begin{equation}
(\bra{0}\otimes I)\mc{V} = \frac{3}{2}\tilde{V} - \frac{1}{2}\tilde{V}\tilde{V}^\dag \tilde{V},
\end{equation}
is close to $V$ with the spectral norm distance bound $\xi:=(\varepsilon^2+3\varepsilon+4)\varepsilon/2$.
Also, the successful probability to project the ancilla to $\ket{0}$ is larger than $(1-\xi)^2$.

In the second way to implement the LCU formula, we randomly sample the terms $\{V_i\}$ in \autoref{eq:LCU} based on the probability $\Pr(i)$~\cite{childs2012Hamiltonian}. This way saves the ancillary qubit number and the cost to implement the multiqubit Toffoli gates in $U_{AE}$ and $\mr{Sel}(\tilde{V})$ in \autoref{eq:UAE_SelV}, which is more suitable to implement in the near-term devices. 
Instead of implementing the operator $V$ directly, we now focus on the task to estimate the properties of the target state $\sigma = V \rho V^\dag$ where $\rho$ is the initial state of the Hamiltonian simulation task. 
Suppose we want to estimate expectation value of a given observable $O$ on $\sigma$, we can then embed the task to Hadamard test~\cite{kitaev1995quantum} shown in \autoref{fig:HadamardTest}.

\begin{figure}[htbp]
\centering
\includegraphics[width=0.7\linewidth]{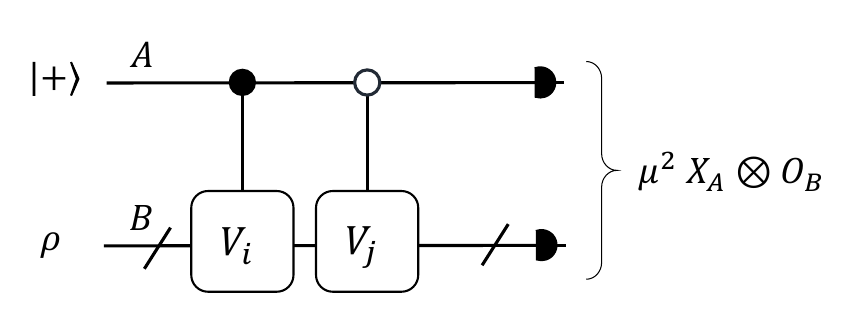}
\caption{Estimate the observable value $\tr(V\rho V^\dag O)$ based on the random sampling implementation of $(\mu,\varepsilon)$-LCU formula of $V$ defined in \autoref{eq:LCU}.}
\label{fig:HadamardTest}
\end{figure}

We first prepare a $\ket{+}$ state on a single ancillary qubit. After that, we implement two controlled operations $\ket{0}_A\bra{0}\otimes I_B + \ket{1}_B\bra{1}\otimes (V_i)_B$ and $\ket{0}_A\bra{0}\otimes (V_j)_B + \ket{1}_B\bra{1}\otimes I_B$ where $V_i$ and $V_j$ are sampled independently from the LCU formula in \autoref{eq:LCU}. The measured expectation value of $\langle \mu^2 X_A\otimes O_B \rangle$ is then a nearly unbiased estimation of $\tr(V\rho V^\dag O)$. We will use the following performance guarantee for the observable estimation.

\begin{proposition}[Performance of the random-sampling LCU implementation, Theorem~2 from \cite{faehrmann2021randomizing}] \label{prop:randomLCU}
For a target operator $V$ and its $(\mu,\varepsilon)$-LCU formula defined in Definition~\ref{def:LCU}, if we estimate the value $\braket{O}_V:= \tr(V\rho V^\dag O)$ with an initial state $\rho$ and observable $O$ using the circuit in \autoref{fig:HadamardTest} for $N$ times, then the distance between the mean estimator value $\hat{O}$ and the true value $\braket{O}_V$ is bounded by
\begin{equation} \label{eq:Obound}
|\hat{O} - \braket{O}_V| \leq \|O\|(3\varepsilon + \varepsilon_n),
\end{equation}
with successful probability $1-\delta$. Here, $N=2\mu^4\ln(2/\delta)/\varepsilon_n^2$, $\|O\|$ is the spectral norm of $O$.
\end{proposition}

From \autoref{prop:randomLCU} we can see that, the $1$-norm $\mu$ of the LCU formula affects the sample complexity while the accuracy factor $\varepsilon$ introduces extra bias in the observable estimation. 
To estimate $\langle O\rangle_V$ using Hadamard-test-type circuits with $\varepsilon$ accuracy, we need $\mc O(\mu^4/\varepsilon^2)$ sampling resource, which owns an extra $\mu^4$ overhead compared to the normal Hamiltonian simulation algorithms~\cite{faehrmann2021randomizing}. 
To make the algorithm efficient, we need to set $\mu$ to be a constant.

In the later discussion, we will construct new LCU formulas from the product of LCU formulas. We will use following proposition. 

\begin{proposition}[Product of LCU formulas] \label{prop:ProductRLCU}
Suppose we have a $(\mu,\varepsilon)$-LCU formula $\tilde{V}$ for an operator $V$ with the form of \autoref{eq:LCU}.
Then the product formula
\begin{equation}
\tilde{V}^\nu = \mu^\nu \sum_{i_1,i_2,...,i_\nu} \Pr(i_1)\Pr(i_2)...\Pr(i_\nu) V_1 V_2...V_\nu, \quad \nu\in\mbb{N},
\end{equation}
is a $(\mu',\varepsilon')$-LCU formula for $V^\nu$ with 
\begin{equation}
\mu' = \mu^\nu, \quad \varepsilon' \leq \nu\mu' \varepsilon.
\end{equation}

\end{proposition}

\begin{proof}

The $1$-norm is obvious. We now bound the distance $\varepsilon'$ between $\tilde{U}^\nu$ and $U^\nu$.
\begin{align*}
\|U^\nu - \tilde{U}^\nu\| &= \left\| \sum_{k=1}^\nu U^{k-1} (U - \tilde{U}) \tilde{U}^{\nu-k} \right\| \\
&\leq \sum_{k=1}^\nu  \| U^{k-1}\|  \|U - \tilde{U}\| \| \tilde{U}^{\nu-k} \| \\
&\leq \nu \|U-\tilde{U}\| \sum_{k=1}^\nu \max\{ \|U\|, \|\tilde{U}\| \}^{\nu-1} \stepcounter{equation}\tag{\theequation}\label{eq:tildeUnuUnu} \\
&\leq \nu \varepsilon \mu^{\nu-1} \leq \nu \mu' \varepsilon.
\end{align*}
\end{proof}

We remark that, when there are common unitary components for all $\{V_i\}$ in the LCU formula \autoref{eq:LCU}, we can simplify the circuit in \autoref{fig:HadamardTest} by removing the ancillary control on the common unitary components. Suppose we have the following LCU formula for an operator $V$,
\begin{equation} \label{eq:LCUmixed}
\tilde{V} = \mu \sum_{i_1,i_2,...,i_\nu} \Pr(i_1,i_2,...,i_\nu)  V_{i_\nu} W_\nu ... V_{i_2} W_2 V_{i_1} W_1,
\end{equation}
such that $\|V-\tilde{V}\|\leq \varepsilon$, where $W_1,W_2,...,W_\nu$ are some fixed unitaries. Then according to Definition~\ref{def:LCU}, \autoref{eq:LCUmixed} is a $(\mu,\varepsilon)$-LCU formula of $V$ with the elementary unitaries to be $V_{\vec{i}}=V_{i_\nu} W_\nu ... V_{i_2} W_2 V_{i_1} W_1$ for $\vec{i}=\{i_1,i_2,...,i_\nu\}$. Instead of naively apply the Hadamard test circuit in \autoref{fig:HadamardTest}, we can introduce an equivalent circuit implementation shown in \autoref{fig:RLCUmixed}(b). The performance guarantee in Proposition~\ref{prop:randomLCU} also holds for the improved implementation.

\begin{figure}
\centering
\includegraphics[width=\linewidth]{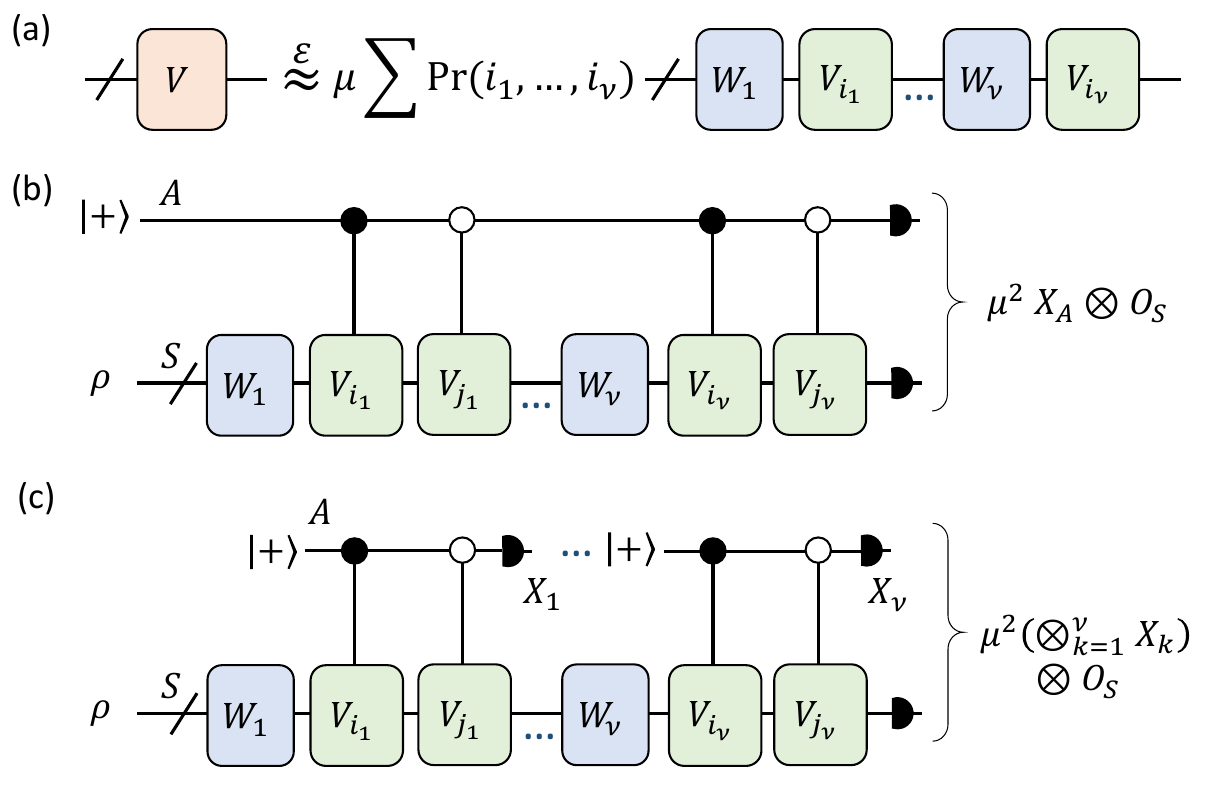}
\caption{ Improve the random-sampling implementation of Hamiltonian simulation based on $(\mu,\varepsilon)$-LCU formula in \autoref{eq:LCUmixed}. 
(a) The LCU formula of $V$ in \autoref{eq:LCUmixed} with deterministic components.
(b) The improved application of Hadamard test. We remove the control on the fixed unitaries and apply them only once. 
(c) A variant of Hadamard test, where the ancillary qubit is measured and reset for each randomly sampled unitary $V_{i_k}$ for $k=1,...,\nu$.
}
\label{fig:RLCUmixed}
\end{figure}

Furthermore, we notice that $\tr(O \tilde{V} \rho \tilde{V}^\dagger )$ can be written as,
\begin{equation}
\begin{aligned}
\mbb{E}_{i_1,j_1;...;i_\nu,j_\nu} \tr[O \, \mc{E}_{i_\nu, j_\nu} \circ \mc{W}_\nu \circ \,...\,\circ \mc{E}_{i_1, j_1} \circ \mc{W}_1(\rho)],
\end{aligned}
\end{equation}
where $\mc{W}_k (\bullet) := W_k \bullet W_k^\dagger$ and $\mc{E}_{i_k,j_k}(\bullet) := \frac{1}{2}( V_{i_k} \bullet V_{j_k}^\dagger + V_{j_k} \bullet V_{i_k}^\dagger)$ for $k=1,2,...,\nu$. 
As a result, we can implement each channel $\mc{W}_k$ by a unitary and each map $\mc{E}_{i_k,j_k}$ by a Hadamard-test-type circuit.  
This leads to a variant circuit shown in \autoref{fig:RLCUmixed}(c), where the ancillary qubit is measured and reset for every segment. While this circuit owns the same gate complexity as \autoref{fig:RLCUmixed}(b), it is beneficial in a fault-tolerant quantum computer since we do not need to store the ancillary qubit for a long time-the ancillary qubit is activated only in the compensation stage and is quickly measured.

\section{Trotter-LCU algorithms with paired Taylor-series compensation} \label{sec:PTSC}

In this section, we construct a $(\mu,\varepsilon)$-LCU formula for the remainder $V_K(x)$ of the $K$-th order Trotter formula based on the idea to perform Taylor-series (TS) expansion on all the exponential terms in $V_K(x)$. Although TS expansion is naturally a LCU formula of $V_K(x)$, it usually owns poor $\mu$ performance.
To further suppress the $1$-norm of the expansion, we will modify the introduce a ``pairing'' idea to combine the terms that correspond to different TS expansion orders.
We first consider a simple case without Trotter formula in \autoref{ssc:pairedTS0} to illustrate the major idea to construct TS-based LCU formula. Then we take the construction of LCU for the first-order Trotter formula as an example in \autoref{ssc:pairedTS1}. Finally, in \autoref{ssc:pairedTSsampling}, we discuss the random-sampling implementation of the algorithm and analyze its sample and gate complexity.

\subsection{Zeroth-order case} \label{ssc:pairedTS0}

We begin our discussion from the case where the Trotter formula is trivial, i.e., $S_0(x)=I$ for each segment. In this case, the Trotter remainder is $V_0(x)=U(x)=e^{-ixH}$. To construct LCU formula, we expand $V_0(x)$ by Taylor series~\cite{berry2015simulating},
\begin{equation} \label{eq:UxTSF}
\begin{aligned}
V_0(x) &= \sum_{s=0}^{\infty} \frac{(-i\lambda x)^s}{s!} \sum_{l_1, ..., l_s} p_{l_1} p_{l_2}... p_{l_s} P_{l_1} P_{l_2} ... P_{l_s} \\
&= \sum_{s=0}^{\infty} F_{0,s}(x), 
\end{aligned}
\end{equation}
where
\begin{equation} \label{eq:F0sx}
F_{0,s}(x) = \frac{(-i\lambda x)^s}{s!} \sum_{l_1, ..., l_s} p_{l_1} p_{l_2}... p_{l_s} P_{l_1} P_{l_2} ... P_{l_s}.
\end{equation}

The $1$-norm of $V_0(x)$ is
\begin{equation}
\|V_0(x)\|_1 = \sum_{s=0}^{\infty} \|F_{0,s}(x)\|_1 = 1 + (\lambda x) + \frac{1}{2!}(\lambda x)^2 + ...= e^{\lambda x}.
\end{equation}
That is, the $1$-norm of $V_0(x)$ is exponentially large with respect to $\lambda x$. 
Suppose we use $V_0(x)$ directly for the random-sampling implementation of LCU following \autoref{fig:TrotterLCUidea}(b), the composite LCU formula for $U(t)=V_0(x)^\nu$ is the product of $V_0(x)$. Based on \autoref{prop:ProductRLCU}, the $1$-norm of the product formula is $\mu=(e^{\lambda x})^\nu=e^{\lambda t}$, which increases exponentially with the simulation time $t$.
Based on \autoref{prop:randomLCU}, this implies an exponentially increasing sample cost $N=\mc{O}(\mu^4)$. 
To make the TS expansion practical for the random-sampling implementation, we need to reduce $1$-norm of $V_0(x)$.

When $\lambda x$ is a small value, the major contribution to $\|V_0(x)\|_1$ comes from the low-order terms of $F_{0,s}(x)$.
Note that the first-order term $F_{0,1}(x)$ is anti-Hermitian. This allows us to utilize the following Euler's formula on Pauli operators: for a Pauli matrix $P$ and $y\in\mbb{R}$,
\begin{equation} \label{eq:IiyP}
I + iyP = \sqrt{1+y^2} e^{i\theta(y) P},
\end{equation}
where $\theta(y):=\tan^{-1}(y)$. 

To suppress $\|V_0(x)\|_1$, we rewrite $V_0(x)$ as follows:
\begin{equation} \label{eq:Ux0thLOP}
\begin{aligned}
& V_0(x) = \sum_{s=0}^{\infty} F_{0,s}(x) \\
&= I - i\lambda x \sum_{l=1}^L p_l P_l + \sum_{s=2}^{\infty} F_{0,s}(x) \\
&= \sum_{l=1}^L p_l (I - i\lambda x P_l) + \sum_{s=2}^{\infty} F_{0,s}(x). 
\end{aligned}
\end{equation}
Then, we can apply \autoref{eq:IiyP} on \autoref{eq:Ux0thLOP} to convert the first-order term to Pauli rotation unitaries
\begin{equation} \label{eq:Ux0thLOP2}
\begin{aligned}
V_0^{(p)}(x) = \sqrt{1+(\lambda x)^2} \sum_{l=1}^L p_l e^{i\theta_0 P_l} + \sum_{s=2}^{\infty} F_{0,s}(x),
\end{aligned}
\end{equation}
where $\theta_0:=\tan^{-1}(-\lambda x)$. We call the formula in \autoref{eq:Ux0thLOP2} the $0$th-order PTS formula.

In practice, to avoid the sampling in the infinitely large space, we introduce a truncation $s_c\geq 2$ to the expansion order of $x$. After this truncation, the approximated LCU formula of $V_0^{(p)}(x)$ is
\begin{equation} \label{eq:V0TildePTSFormula}
\begin{aligned}
&\tilde{V}_0^{(p)}(x) := \tilde{U}_0^{(p)}(x) = \sqrt{1+(\lambda x)^2} \sum_{l=1}^L p_l e^{i\theta_0 P_l} + \sum_{s=2}^{s_c} F_{0,s}(x) \\
&= \mu_0^{(p)}(x) \left(\mr{Pr}_0^{(p)}(1)V_{0,1}^{(p)} + \sum_{s=2}^{s_c} \mr{Pr}_0^{(p)}(s) V_{0,s}^{(p)} \right),
\end{aligned}
\end{equation}
where 
\begin{align*}
\mu_0^{(p)}(x) &:= \sqrt{1+(\lambda x)^2} + \sum_{s=2}^{s_c} \|F_{0,s}(x)\|_1  \\
&\leq \sqrt{1+(\lambda x)^2} + \left( e^{\lambda x} - 1 - \lambda x \right), \\
\mr{Pr}_0^{(p)}(s) &:= \frac{1}{\mu_0^{(p)}(x)}
\begin{cases}
\sqrt{1+(\lambda x)^2}, & s=1, \\
\|F_{0,s}(x)\|_1 = \frac{(\lambda x)^s}{s!}, & s=2,3,...,s_c,
\end{cases}  \stepcounter{equation}\tag{\theequation}\label{eq:0thLOPPara} \\
V_{0,s}^{(p)} &:= 
\begin{cases}
\sum_{l} p_l \,e^{i\theta_0 P_l}, & s=1, \\
\sum_{l_{1:s}} p_{l_{1:s}}^{(s)} (-i)^s P_{l_{1:s}}^{(s)}, & s=2,3,...,s_c.
\end{cases} 
\end{align*}
After ``pairing'' the terms with $s=0$ and $s=1$, we obtain \autoref{eq:V0TildePTSFormula} which is a new LCU formula with the $1$-norm $\mu_0^{(p)}(x)$. 
We have the following proposition to characterize the LCU formula in \autoref{eq:V0TildePTSFormula}.
\begin{proposition}[0th-order Trotter-LCU formula by paired Taylor-series compensation] \label{prop:PTS0}
For $x\geq 0$ and $s_c\geq 2$, $\tilde{V}_0^{(p)}(x)$ in \autoref{eq:V0TildePTSFormula} is a $(\mu_0^{(p)}(x),\varepsilon_0^{(p)}(x))$-LCU formula of $V_0(x)=U(x)$ with
\begin{equation}
\mu_0^{(p)}(x) \leq e^{\frac{3}{2}(\lambda x)^2}, \quad 
\varepsilon_0^{(p)}(x) \leq \left(\frac{e\lambda x}{s_c +1}\right)^{s_c+1}.
\end{equation}
\end{proposition}

\begin{proof}
For the normalization factor, we have
\begin{align*}
& \mu_0^{(p)}(x) \leq \sqrt{1+(\lambda x)^2} + \left( e^{\lambda x} - 1 - \lambda x \right) \\
&\leq 1 + \frac{1}{2} (\lambda x)^2 + \left( e^{\lambda x} - 1 - \lambda x \right) \stepcounter{equation}\tag{\theequation}\label{eq:mu0Ixbound} \\
&\leq e^{(\lambda x)^2} + \frac{1}{2} (\lambda x)^2 \leq e^{\frac{3}{2} (\lambda x)^2}.
\end{align*}
The fourth inequality is due to $e^x - x \leq e^{x^2}$ for $x\in\mbb{R}$.

For the distance bound, we have 
\begin{equation} \label{eq:VboundPTS0}
\begin{aligned}
& \|\tilde{V}_0(x)- V_0(x)\| \leq \sum_{s>s_c} \|F_{0,s}(x)\| \\
&\leq \sum_{s>s_c} \frac{(\lambda x)^s}{s!} \|V_{0,s}(x)\| \\
&\leq \sum_{s>s_c} \frac{(\lambda x)^s}{s!} \leq \left( \frac{e\lambda x}{s_c+1} \right)^{s_c+1}.
\end{aligned}
\end{equation} 
In the second inequality, we use the fact that
\begin{equation}
\|V_{0,s}\| \leq \sum_{l_1,l_2,...,l_s} p_{l_1}p_{l_2}...p_{l_s} \|(-i)^s P_{l_1}P_{l_2}...P_{l_s}\| \leq 1.
\end{equation}
In the fourth inequality of \autoref{eq:VboundPTS0}, we apply the following Poisson tail bound formula,
\begin{equation} \label{eq:ExpTailPower}
\sum_{s=k+1}^\infty \frac{x^s}{s!} \leq \left(\frac{e x}{k+1}\right)^{k+1},
\end{equation}
which can be proven from Theorem~1 in Ref.~\cite{canonne2016poisson}.
\end{proof}

From \autoref{prop:PTS0} we see that, the $1$-norm of the LCU formula $V_0^{(p)}(x)$ in \autoref{eq:Ux0thLOP2} is $\exp(\frac{3}{2}(\lambda x)^2)\approx 1 + \mc{O}((\lambda x)^2)$, whose leading term is quadratically smaller than the one of $\|V_0^{(p)}(x)\|_1 = \exp(\lambda x) \approx 1 + \mc{O}(\lambda x)$ when $\lambda x\ll 1$. We will later see that $V_0^{(p)}(x)$ provides us an efficient random-sampling implementation of the Trotter-LCU algorithm.

\subsection{First-order case} \label{ssc:pairedTS1}

Following similar ideas in \autoref{ssc:pairedTS0}, we now study the PTS compensation of the Trotter remainder $V_K(x)$. We will take first-order case as an example. The Trotter remainder $V_1(x)$ is
\begin{equation}
V_1(x) = U(x) S_1(x)^\dag.
\end{equation}

From \autoref{eq:S1x} and \autoref{eq:UxTSF} we have
\begin{equation} \label{eq:UxS1x}
\begin{aligned}
S_1(x)^\dag &= \prod^{\rightarrow} e^{ix \vec{H}} 
= \left( \prod e^{\vec{\alpha} x} \right)  \sum_{\vec{r}} \mr{Poi}(\vec{r};\vec{\alpha} x) \prod^{\rightarrow} (i \vec{P})^{\vec{r}} \\
U(x) &= \sum_{r=0}^{\infty} \frac{(\lambda x)^r}{r!} \sum_{l_1, ..., l_r} p_{l_1} p_{l_2}... p_{l_r} (-i)^r P_{l_1} P_{l_2} ... P_{l_r} \\
&= e^{\lambda x} \sum_{r=0}^{\infty} \mr{Poi}(r;\lambda x) \sum_{l_{1:r}} p^{(r)}_{l_{1:r}} (-i)^r P^{(r)}_{l_{1:r}}.
\end{aligned}
\end{equation}
Here, we adopt the vector notations introduced in \autoref{ssc:Trotter} to simplify the expressions. In \autoref{eq:UxS1x}, we also extend the notation of Poisson distribution,
\begin{equation}
\mr{Poi}(\vec{r};\vec{\alpha}x) := \prod_{l=1}^L \mr{Poi}(r_l;\alpha_l x).
\end{equation}

Based on \autoref{eq:UxS1x}, we then write the Taylor-series expansion of first-order remainder as follows:
\begin{equation} \label{eq:V1xFull}
\begin{aligned}
& V_1(x) = U(x) S_1(x)^\dag \\
&= e^{2\lambda x} \sum_{r;\vec{r}} \mr{Poi}(r,\vec{r};\lambda x, \vec{\alpha} x) \sum_{l_{1:r}} p^{(r)}_{l_{1:r}} (-i)^{r-\sum \vec{r}} P^{(r)}_{l_{1:r}} \prod^{\rightarrow} \vec{P}^{\vec{r}},
\end{aligned}
\end{equation}
which is a LCU formula with $1$-norm $e^{2\lambda x}$. Here, $r$ and $\vec{r}$ are two groups of independent variables. 

Now, we utilize the order condition of the Trotter formula in \autoref{lem:TrotterOrder} to reduce $1$-norm of \autoref{eq:V1xFull}. To this end, we first rewrite \autoref{eq:V1xFull} by classifying the terms based on the power of $x$, which is determined by the value $s=r+\sum \vec{r}$,
\begin{equation} \label{eq:V1xExpand}
\begin{aligned}
V_1(x) = \sum_{s=0}^{\infty} F_{1,s}(x) = \sum_{s=0}^\infty \|F_{1,s}(x)\|_1 V_{1,s}.
\end{aligned}
\end{equation}
Here, $\|F_{1,s}(x)\|_1$ is the $1$-norm of the $s$-order expansion formula $F_{1,s}(x)$. $V_{1,s}$ denotes the normalized LCU formula for the $s$-order terms. We have
\begin{equation} \label{eq:V1xExpand2}
\begin{aligned}
F_{1,s}(x) &:= \sum_{r+\sum \vec{r} = s} \frac{(\lambda x)^r}{r!} \left(\prod \frac{(\vec{\alpha} x)^{\vec{r}}}{\vec{r}!} \right) \cdot \\ 
& \quad\quad\quad \sum_{l_{1:r}} p^{(r)}_{l_{1:r}} (-i)^{2r-s} P^{(r)}_{l_{1:r}} \prod^{\rightarrow} \vec{P}^{\vec{r}}, \\
V_{1,s} &:= \sum_{r;\vec{r}} \Pr(r,\vec{r}|s) \sum_{l_{1:r}} p^{(r)}_{l_{1:r}} (-i)^{2r-s} P^{(r)}_{l_{1:r}} \prod^{\rightarrow} \vec{P}^{\vec{r}}.
\end{aligned}
\end{equation}
In the expression of $V_{1,s}$, we use $r-\sum\vec{r} = r - (s - r) = 2r-s$.
The conditional probability $\Pr(r,\vec{r}|s)$ indicates the probability to sample $r$ and $\vec{r}$ when their summation $s$ is given,
\begin{equation} \label{eq:PrConMul}
\begin{aligned}
\Pr(r,\vec{r}|s) &= \frac{\Pr\{(\hat{r}=r,\hat{\vec{r}}=\vec{r})\cap (\hat{s}=s) \}}{\Pr\{\hat{s}=s\}} \\
&= \frac{s!}{r!\prod\vec{r}!} \left(\frac{1}{2}\right)^r \prod \left(\frac{\vec{p}}{2} \right)^{\vec{r}} \\
&= \mr{Mul}(\{r,\vec{r}\};\{\frac{1}{2},\frac{\vec{p}}{2}\};s),
\end{aligned}
\end{equation}
which is a multinomial distribution $\mr{Mul}(\cdot;\cdot;s)$ with $s$ trials and $(L+1)$ outcomes. In each trial, the $(L+1)$ outcomes $\{r;r_1,r_2,...,r_L\}$ occur with the corresponding probability $\{\frac{1}{2}, \frac{1}{2}p_1, \frac{1}{2}p_2,...,\frac{1}{2}p_L\}$. Recall that $\vec{p}=\{p_1, p_2,...,p_L\}$ is the normalized Hamiltonian coefficients defined in \autoref{eq:H}.

We first estimate the normalization cost $\|F_{1,s}(x)\|_1$ for different $s$ orders,
\begin{equation} \label{eq:muF1sx}
\begin{aligned}
\|F_{1,s}(x)\|_1 &= \sum_{s;\vec{r}} \mathbbm{1}\left[ r+\sum \vec{r} = s\right] \frac{(\lambda x)^r}{r!} \left(\prod \frac{(\vec{\alpha} x)^{\vec{r}}}{\vec{r}!} \right) \\
&= \frac{(\lambda x + \sum \vec{\alpha} x)}{s!} = \frac{ (2\lambda x)^s }{s!} =: \eta_s.
\end{aligned}
\end{equation}
In the second inequality, we use the following equation
\begin{equation} \label{eq:sum_twoPoission}
\sum_{s_1,s_2=0}^r \mathbbm{1}[s_1+s_2=r] \frac{(x_1)^{s_1}}{s_1!} \frac{(x_2)^{s_2}}{s_2!} = \frac{(x_1+x_2)^r}{r!}.
\end{equation}
We denote $\eta_s := \frac{(2\lambda x)^s}{s!}$, which will be frequently used in the following discussion.

From \autoref{eq:muF1sx} we can see that, similar to the expansion of $V_0(x)$, the $1$-norm of the expansion terms with different orders of $x$ in $V_1(x)$ follow the Possion distribution. This motivates us to eliminate the low-order terms such as $F_{1,1}(x)$. Based on Lemma~\ref{lem:TrotterOrder}, we have
\begin{equation}
F_{1,1}(x) = 0.
\end{equation}
As a result, we can directly remove the term $F_{1,1}(x)$ in \autoref{eq:V1xExpand}. The resulting formula is,
\begin{equation} \label{eq:V1xDS}
\begin{aligned}
V_1(x) &= I + \sum_{s=2}^{\infty} \eta_s V_{1,s}.
\end{aligned}
\end{equation}

After the elimination of the first-order term, we now introduce Euler's formula to suppress higher-order terms in $V_{1,2}$ and $V_{1,3}$. For the convenience of later discussion, we simplify the notation of $F_{1,s}(x)$ in \autoref{eq:V1xExpand2} as follows:
\begin{equation} \label{eq:F1sExpand3}
\begin{aligned}
F_{1,s}(x) = \eta_s \sum_{r,\gamma} \Pr(r,\gamma|s) P_1(r,\gamma), \\
\end{aligned}
\end{equation}
where $\gamma$ is used to denote all the expansion variables $\{\vec{r},l_{1:r}\}$ besides $r$. $\Pr(r,\gamma|s)$ and $P_1(r,\gamma)$ are then defined to be
\begin{equation} \label{eq:PrrgammasP1rgamma}
\begin{aligned}
\Pr(r,\gamma|s) &:= \Pr(r,\vec{r}|s) p^{(r)}_{l_{1:r}}, \\
P_1(r,\gamma) &:= (-i)^{2r-s} P^{(r)}_{l_{1:r}} \prod^{\rightarrow} \vec{P}^{\vec{r}}.
\end{aligned}
\end{equation}

To apply Euler's formula in \autoref{eq:IiyP} to \autoref{eq:V1xDS}, we need to make sure that the Pauli operator $P$ is Hermitian.
To this end, we classify the Pauli terms of $F_{1,s}$ in \autoref{eq:F1sExpand3} into Hermitian and anti-Hermitian types,
\begin{equation} \label{eq:F1sDecomp}
\begin{aligned}
F_{1,s}(x) &= \eta_s \sum_{ \{r,\gamma\} \in \text{Her} } \Pr(r,\gamma|s) P_{1,s}(r,\gamma) \\
& + i \eta_s \sum_{ \{r,\gamma\} \in \text{anti-Her} } \Pr(r,\gamma|s)  (-i) P_{1,s}(r,\gamma),
\end{aligned}
\end{equation}
where $\{r,\gamma\}\in \text{Her}$ and $\{r,\gamma\}\in \text{anti-Her}$, respectively, indicates the set of $\{r,\gamma\}$ such that $P_1(r,\gamma)$ is Hermitian and anti-Hermitian. 
When $\{r,\gamma\}\in\text{Her}$, the corresponding Pauli operator owns a real coefficient, which cannot be paired with $I$ based on Euler's formula. 

It seems that we cannot eliminate the Hermitian terms in \autoref{eq:F1sDecomp} by Euler's formula. However, by taking advantage of the properties of Trotter remainder $V_1(x)$, we can show that $V_{1,2}$ and $V_{1,3}$ are actually anti-Hermitian.
Recall that we can write the exponential form of $V_1(x)$ by applying the BCH formula on the definition of $V_1(x)$ in \autoref{eq:Vx},
\begin{equation} \label{eq:V1xExpandExp2}
V_1(x) = \exp\left(i ( E_{1,2} x^2 + E_{1,3} x^3 + E_{1,4} x^4 +...) \right),
\end{equation}
where $\{E_{1,s}\}$ are some Hermitian operators determined by the BCH formula. The first order term $E_{1,1}$ vanishes due to the order condition in \autoref{lem:order}. 

\begin{lemma}[Lemma~1 in \cite{childs2019nearly}] \label{lem:order}
Let $F(x)$ be an operator-valued function that is infinitely differentiable. Let $K\geq 1$ be a non-negative integer. The following two conditions are equivalent.
\begin{enumerate}
\item Asymptotic scaling: $F(x)=\mc{O}(x^{K+1})$.
\item Derivative condition: $F(0)=F'(0)=...=F^{(K)}(0)=0$.
\end{enumerate}
\end{lemma}

From \autoref{eq:V1xExpandExp2} we then have
\begin{equation} \label{eq:F12xF13x}
\begin{aligned}
F_{1,2}(x) &= i x^2 E_{1,2}, \\
F_{1,3}(x) &= i x^3 E_{1,3}
\end{aligned}
\end{equation}
by the Taylor-series expansion on \autoref{eq:V1xExpandExp2}. 

Comparing \autoref{eq:F1sDecomp} and \autoref{eq:F12xF13x}, we can see that
\begin{equation} \label{eq:F12xF13xHerZero}
\begin{aligned}
\sum_{ \{r,\gamma\} \in \text{Her} } \Pr(r,\gamma|s) P_{1,s}(r,\gamma) = 0, \quad s=2,3.
\end{aligned}
\end{equation}
This is because $F_{1,2}(x)$ and $F_{1,3}(x)$ are anti-Hermitian from \autoref{eq:F12xF13x}.

We can then modify the form of $F_{1,s}(x)$ ($s=2$ or $3$) in \autoref{eq:F1sDecomp} as follows:
\begin{equation} \label{eq:F1sxModify}
\begin{aligned}
F_{1,s}(x) 
&= i \eta_s \sum_{ \{r,\gamma \} \in \text{Her} } \Pr(r,\gamma|s) P_{1,s}(r,\gamma) \\
&\quad + i \sum_{ \{r,\gamma\} \in \text{anti-Her} } \Pr(r,\gamma|s) (-i) P_{1,s}(r,\gamma), \\
&= i \eta_s \sum_{r,\gamma} \Pr(r,\gamma|s) (-i)^{\mathbbm{1}[P_{1,s}(r,\gamma):\text{anti-Her}]} P_{1,s}(r,\gamma).
\end{aligned}
\end{equation}
In \autoref{eq:F1sxModify}, we intentionally add an extra phase $i$ on the Hermitian terms, which has no effect on $F_{1,s}(x)$ as they own zero summation value. In this way, all the Pauli expansion terms in $F_{1,2}(x)$ and $F_{1,3}(x)$ are with imaginary coefficient, which can be paired with $I$ using Euler's formula in \autoref{eq:IiyP}. We call the second- and third-order terms $V_{1,s} (s=2,3)$ the leading TS expansion orders of $V_1(x)$.
The major reason to the Hermitian terms with zero summation value is to simplify the sampling procedure of $V_{1,s}$, which will be clarified later. 

Now, we are going to eliminate the leading-order terms in $V_{1}(x)$,
\begin{equation}
\begin{aligned}
V_1(x) &= I + \sum_{s=2}^3 i \eta_s V_{1,s}' + \sum_{s>3} \eta_s V_{1,s}, \\
V_{1,s}' &= \sum_{r,\gamma} \Pr(r,\gamma|s) P_{1,s}'(r,\gamma), \quad s=2,3\\
V_{1,s} &= \sum_{r,\gamma} \Pr(r,\gamma|s) (-i)^{2r-s} P_{1,s}(r,\gamma), \quad s\geq 4.
\end{aligned}
\end{equation}
Here, the Pauli operator for the leading-order term $P_{1,s}'(r,\gamma) := (-i)^{\mathbbm{1}[P_{1,s}(r,\gamma):\text{anti-Her}]} P_{1,s}(r,\gamma)$ is always with a real coefficient. 
We can then pair $I$ with the Pauli operators in $F_{1,2}(x)$ and $F_{1,3}(x)$,
\begin{equation} \label{eq:V1xPTSexpand}
\begin{aligned}
& V_1^{(p)}(x) = I + i\eta_2 V_{1,2}' + i\eta_3 V_{1,3}' + \sum_{s=4}^{\infty} \eta_s V_{1,s} \\
&= \frac{\eta_2}{\eta_\Sigma} (I + \eta_{\Sigma} V_{1,2}'(x) ) + \frac{\eta_3}{\eta_\Sigma} (I + \eta_{\Sigma} V_{1,3}'(x) ) + \sum_{s=4}^{\infty} \eta_s V_{1,s} \\
&= \sqrt{1+\eta_\Sigma^2} \left( \frac{\eta_2}{\eta_\Sigma} R_{1,2}(\eta_\Sigma) + \frac{\eta_3}{\eta_\Sigma} R_{1,3}(\eta_\Sigma) \right) + \sum_{s=4}^{\infty} \eta_s V_{1,s}.
\end{aligned}
\end{equation}
Here, $\eta_\Sigma:= \eta_2 + \eta_3$. The third line of \autoref{eq:V1xPTSexpand} is the final LCU formula used for the first-order PTSC algorithm. In the third line, we apply the following pairing procedure,
\begin{equation} \label{eq:V1xPTSPara}
\begin{aligned}
& I + \eta_{\Sigma} V_{1,s}' = I + i \eta_{\Sigma} \sum_{r,\gamma} \Pr(r,\gamma|s) P_{1,s}'(r,\gamma) \\
&= \sum_{r,\gamma} \Pr(r,\gamma|s) (I + \eta_{\Sigma} P_{1,s}'(r,\gamma) ) \\
&= \sqrt{1+ \eta_\Sigma^2}  \sum_{r,\gamma} \Pr(r,\gamma|s) \exp\left( i\theta(\eta_\Sigma) P_{1,s}'(r,\gamma) \right) \\
&=: \sqrt{1+ \eta_\Sigma^2} R_{1,s}(\eta_\Sigma), \quad s=2,3.
\end{aligned}
\end{equation}

In practice, we introduce a truncation $s_c$ on the expansion formula in \autoref{eq:V1xPTSexpand}. For the convenience of analysis, we set $s_c>3$. 
The truncated LCU formula for $V_1^{(p)}(x)$ is then
\begin{equation} \label{eq:tildeV1xPTS}
\tilde{V}_1^{(p)}(x) = \mu_1^{(p)}(x) \Big( \sum_{s=2,3}\mr{Pr}_1^{(p)}(s) R_{1,s}(\eta_\Sigma) + \sum_{s=4}^{s_c} \mr{Pr}_1^{(p)}(s) V_{1,s} \Big),
\end{equation}
with the $1$-norm $\mu_1^{(p)}(x)$ and probabilities $\Pr_1^{(p)}(s)$ to sample the $s$-order term
\begin{equation} \label{eq:mu1_Pr1p}
\begin{aligned}
\mu_1^{(p)}(x) &= \sqrt{1+\eta_\Sigma^2} + \sum_{s=4}^{s_c} \eta_s, \\
\mr{Pr}_1^{(p)}(s) &= \frac{1}{\mu_1^{(p)}(x)}
\begin{cases}
\sqrt{1+\eta_\Sigma^2} \frac{\eta_s}{\eta_\Sigma}, \quad & s= 2,3, \\
\eta_s, \quad & s= 4,5,...,s_c.\\
\end{cases}
\end{aligned}
\end{equation}

Combined with the deterministic first-order Trotter formula, the overall LCU formula for $U(x)$ is
\begin{equation} \label{eq:tildeU1xPTS}
\tilde{U}_1^{(p)}(x) = \tilde{V}_1^{(p)}(x) S_1(x).
\end{equation}

The following proposition gives the performance characterization of $\tilde{U}^{(p)}_1(x)$ in \autoref{eq:tildeU1xPTS} to approximate $U(x)$.

\begin{proposition}[first-order Trotter-LCU formula by paired Taylor-series compensation] \label{prop:PTS1}
For $0<x<1/(2\lambda)$ and $s_c\geq 3$, $\tilde{V}_1(x)$ in \autoref{eq:tildeV1xPTS} is a $(\mu_1^{(p)}(x),\varepsilon_1^{(p)}(x))$-LCU formula of the first-order Trotter remainder $V_1(x)$ with
\begin{equation}
\mu_1^{(p)}(x) \leq e^{(e+\frac{2}{9})(2\lambda x)^{4}}, \quad 
\varepsilon_1^{(p)}(x) \leq \left(\frac{2e\lambda x}{s_c +1}\right)^{s_c+1}.
\end{equation}
As a result, $\tilde{U}_1(x)$ in \autoref{eq:tildeU1xPTS} is a $(\mu_1^{(p)}(x),\varepsilon_1^{(p)}(x))$-LCU formula of $U(x)$.
\end{proposition}

\begin{proof}
We first bound the normalization factor $\mu_1^{(p)}(x)$. When $2\lambda x<1$ we have
\begin{align*}
&\mu_1^{(p)}(x) \leq \sqrt{1+(\eta_\Sigma)^2} + \sum_{s=4}^{\infty} \eta_s \\
&\leq 1 + \frac{1}{2}\eta_\Sigma^2 + \left( e^{2\lambda x} - \sum_{s=0}^{3} \eta_s \right) \\
&= \frac{1}{2} (2\lambda x)^4 \left( \frac{1}{2!} + \frac{(2\lambda x)}{3!} \right)^2 + \left( e^{2\lambda x} - \sum_{s=1}^{3} \eta_s \right) \stepcounter{equation}\tag{\theequation}\label{eq:mu1PTSbound}\\ 
&\leq \frac{1}{2} (2\lambda x)^4 \left( \frac{1}{2} + \frac{1}{6} \right)^2 + \left( e^{2\lambda x} - \sum_{s=1}^{3} \eta_s \right)  \\
&\leq \frac{2}{9} (2\lambda x)^4 + e^{e(2\lambda x)^4} \leq e^{(e+\frac{2}{9})(2\lambda x)^4}.
\end{align*}

For the distance bound, from \autoref{eq:V1xFull} and \autoref{eq:tildeV1xPTS} we have
\begin{equation} \label{eq:Vtildebound1st}
\begin{aligned}
& \|V_1(x) - \tilde{V}_1^{(p)}(x)\| \leq \sum_{s>s_c} \eta_s \|V_{1,s} \| \\
&= \sum_{s>s_c} \frac{(2\lambda x)^s}{s!} \leq \left(\frac{2e\lambda x}{s_c +1}\right)^{s_c+1}
\end{aligned}
\end{equation} 
In the second line, we use the fact that $\|V_{1,s}\|\leq 1$. In the third line, we apply \autoref{eq:ExpTailPower}.
\end{proof} 

From \autoref{prop:PTS1} we have shown that, by introducing first-order Trotter formula, we can further suppress the 1-norm of the LCU formula $V_1^{(p)}(x)$ to $\exp(c(\lambda x)^4)\approx 1+ \mc{O}(\lambda x)^4$ where $c$ is a constant. 
In Appendix~\ref{sec:pairedTSgeq2}, we discuss the generalized LCU formula construction of higher-order Trotter remainder $V_K(x)$ with $K=2k, k\in\mbb{N}_+$. Under such constructions, we have,
\begin{equation}
\|V_K^{(p)}(x)\|_1 = \exp(c(\lambda x)^{2K+2}).
\end{equation}
In \autoref{ssc:pairedTSsampling} we will see how this can help us to improve the time scaling of the whole simulation algorithm.

\subsection{Random-sampling implementation and performance} \label{ssc:pairedTSsampling}

We have now derived the LCU formulas for Trotter remainders $V_K(x)$ and hence the small time evolution $U(x)=V_K(x) S_K(x)$ based on the idea to utilize the Trotter order condition and pair the leading-order terms in $V_K(x)$ to suppress the normalization factors. 
We now discuss the practical random-sampling implementation of them, taking the first-order case as an example.

Suppose we want to perform Hamiltonian evolution $e^{iHt}$ on an initial state $\rho$. As is illustrated in \autoref{fig:TrotterLCUidea}, in each segment, we first implement the Trotter circuit $S_1(x)$ and then compensate the remainder $V_1(x)$ by LCU. 
In the random-sampling implementation of LCU, we embed the LCU sampling into a modified Hadamard test.
In \autoref{fig:SamplingPTS}, we show the detailed sampling procedure of $V_i$ and $V_j$. In stage 1), we sample the Taylor-expansion order $s$ from a finite probability distribution $\mr{Pr}^{(p)}_1(s)$. Afterwards, in stage 2) and 3), we randomly sample Pauli string indices $\{r,\vec{r}\}$ and $l_1,...,l_r$ based on the LCU formula of $F_{1,s}$. The variables $\{r,\vec{r}\}$ obey a multinomial distribution $\text{Mul}(r,\vec{r};\{\frac{1}{2},\frac{\vec{p}}{2}\};s)$ defined in \autoref{eq:PrConMul} while $\{l_1,l_2,...,l_r\}$ are sampled identically and independently from normalized Hamiltonian coefficients $\vec{p}$ defined in \autoref{eq:H}. Finally, in stage 4), depending on whether $s$ is the leading-order ($s=2,3$ for $K=1$) or not, we determine the sampled unitary $V_i$ to be a Pauli-rotation unitary $e^{i\theta P}$ or just Pauli operators $P$.

For the gate complexity of the random-sampling implementation, we have the following theorem.

\begin{theorem}[Gate complexity of the $K$th-order random-sampling Trotter-LCU algorithm by paired Taylor-series compensation] \label{thm:PTS}
To realize a (probabilistic) Hamiltonian simulation of $e^{iHt}$ with accuracy $\varepsilon$, the gate complexity of random-sampling $K$th-order Trotter-LCU algorithm ($K=0,1$ or $2k$) based on paired Taylor-series compensation is
\begin{equation}
\mc{O}\left((\lambda t)^{1+ \frac{1}{2K+1}}(\kappa_K L + \frac{\log(1/\varepsilon)}{\log\log(1/\varepsilon)} ) \right).
\end{equation}
Here, $\kappa_K=K$ if $K=0$ or $1$, or $\kappa_K=2\times 5^{K/2-1}$ if $K=2k$.
\end{theorem}

\begin{proof}


Without loss of generality, we focus on the case when $K=1$. The case of $K=0$ and $K=2k, k\in\mbb{N}_+$ can be analyzed similarly following \autoref{prop:PTS0} and \autoref{prop:PTS2k}, respectively. 

For the random-sampling implementation, the overall LCU formula for $U(t)$ is to repeat the sampling of $\tilde{U}_{1}^{(p)}(x)$ for $\nu$ times, $\tilde{U}_{1,\mc{tot}}^{(p)}(t) = \tilde{U}_{1}^{(p)}(x)^\nu$. 
Using \autoref{prop:ProductRLCU} and \autoref{prop:PTS1}, when $0<x<\frac{1}{2\lambda}$ and $s_c\geq 3$, we conclude that $\tilde{U}_{1,\mc{tot}}^{(p)}(t)$ is a $(\mu_{1,\mc{tot}}^{(p)}(t), \varepsilon_{1,\mc{tot}}^{(p)}(t) )$-LCU formula of $U(t)$ with
\begin{equation} \label{eq:muEps1stLOP}
\begin{aligned}
\mu_{1,\mc{tot}}^{(p)}(t) &= \mu_{1}^{(p)}(x)^\nu \leq e^{(e+c_1)\frac{(2\lambda t)^{4}}{\nu^{3}}}, \\
\varepsilon_{1,\mc{tot}}^{(p)}(t) &\leq \nu \mu_{1,\mc{tot}}^{(p)}(t) \varepsilon_{1}^{(p)}(x) \leq \nu e^{(e+c_1)\frac{(2\lambda t)^{4}}{\nu^{3}}} \left( \frac{2e\lambda x}{s_c+1} \right)^{s_c+1}.
\end{aligned}
\end{equation}
Here, $c_1 = 2/9$.

To realize a $(\mu,\epsilon)$-LCU formula for $U(t)$, we only need to set the segment number $\nu$ and the truncation order $s_c$ to satisfy
\begin{equation} \label{eq:nuScBoundKthLOP}
\begin{aligned}
\nu &\geq \max\left\{ \nu_{1}^{(p)}(t),  2\lambda t \right\}, \\
s_c &\geq \max\left\{ \left\lceil \frac{\ln\left( \frac{\mu}{\varepsilon} \nu_{1}^{(p)}(t) \right)}{W_0\left( \frac{1}{2e\lambda t}\nu_{1}^{(p)}(t) \ln\left( \frac{\mu}{\varepsilon} \nu_{1}^{(p)}(t) \right) \right)} - 1 \right\rceil, 3 \right\},
\end{aligned}
\end{equation}
we can then realize a $(\mu,\varepsilon)$-LCU formula for $U(t)$ based on $\nu$ segments of $\tilde{U}_{1}^{(p)}$ in \autoref{eq:tildeU1xPTS}. Here, $\nu_{1}^{(p)}(t):=\left(\frac{2(e+c_1)\lambda t}{\ln\mu}\right)^{\frac{1}{3}} 2\lambda t$. $W_0(y)$ is the principle branch of the Lambert $W$ function whose scaling is approximately $\ln(y)$ according to the tight bound in \autoref{lem:WfuncBound} in Appendix~\ref{sec:formulas}. 
To derive the bound for $s_c$ in the second line of \autoref{eq:nuScBoundKthLOP}, we use \autoref{lem:ExpTailSolu} in Appendix~\ref{sec:formulas}. 

The gate complexity of the random-sampling implementation of the Trotter-LCU algorithm is determined by the segment number $\nu$, Trotter order $K$ and the gate complexity of each elementary gate in the LCU formula. As shown in Fig.~\ref{fig:SamplingPTS}, to construct controlled-$U(t)$, we split it to $\nu$ segments. In each segment, we need to implement $K$th-order Trotter circuits and LCU circuits. The gate complexity of each elementary gate in the LCU circuit is determined by the truncation order $s_c$ of the Taylor-series compensation. Specifically, we consider the compilation of the controlled-Pauli gate and controlled-Pauli rotation gate, as shown in \autoref{fig:gateSyn}. The number of gates is determined by the weight of the sampled Pauli matrices $\mr{wt}(P_l)$, which is upper bounded by $\mc{O}(s_c)$.

\begin{figure}[htbp]
\centering
\includegraphics[width=0.8\linewidth]{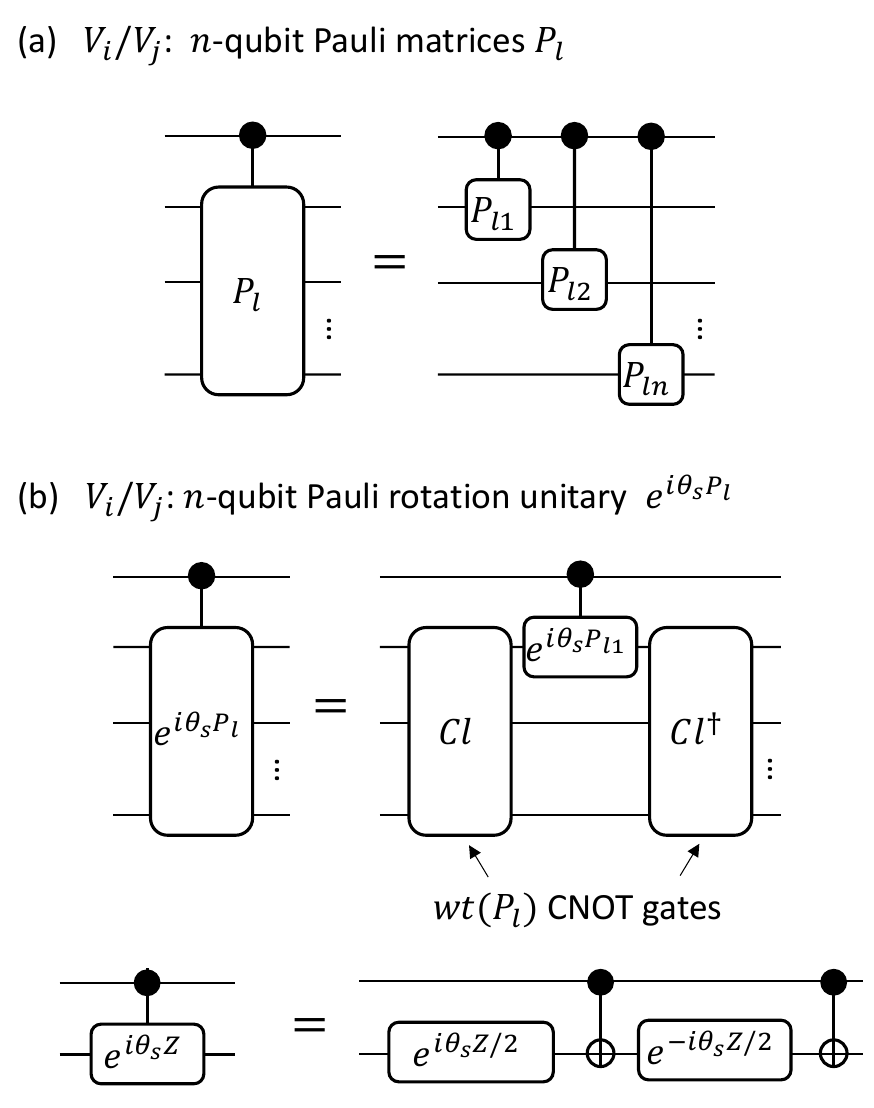}
\caption{ 
(a) Compilation of the LCU circuit when $V_i (V_j)$ is a $n$-qubit Pauli gate. (b) Compilation of the LCU circuit when $V_i (V_j)$ is a Pauli rotation unitary. Here, $Cl$ is a Clifford gate with $\mr{wt}(P_l)$ $\mc{CNOT}$ gates.
}
\label{fig:gateSyn}
\end{figure}

Therefore, the gate complexity of the overall algorithm using $K$th Trotter formula ($K=0,1,2k$) is given by 
\begin{equation} \label{eq:gateNumRandom}
\begin{aligned}
N_{K}^{(RS)} = \mc{O}(\nu(\kappa_K L + s_c))
\end{aligned}
\end{equation}
where
\begin{equation}
\kappa_K = 
\begin{cases}
0, & K=0, \\
1, & K=1, \\
2\times 5^{K/2-1}, & K=2k, k\in \mathbb{N}_+.
\end{cases}
\end{equation}

Based on \autoref{eq:gateNumRandom}, the gate complexity of the $2k$th-order Trotter-LCU algorithm is then 
\begin{equation}
\mc{O}(\nu(\kappa_K L+ s_c)) = \mc{O}\left((\lambda t)^{1+ \frac{1}{4k+1}}(\kappa_K L + \frac{\log(1/\varepsilon)}{\log\log(1/\varepsilon)} ) \right).
\end{equation}
Here, $\kappa_K$ is the stage number of the Trotter formula. When $K=0$, $\kappa_K=0$. When $K=2k$, $\kappa_K = 2\times 5^{K/2-1}$.
\end{proof}

From \autoref{thm:PTS} we can see that, by introducing the LCU compensation, the time scaling of the $K$th-order Trotter-LCU algorithm improve the bare $K$th-order Trotter time scaling from $1+1/K$ to $1+1/(2K+1)$. Moreover, the accuracy scaling is exponentially improved. This allowss us to achieve optimal gate complexity with lower-order Trotter implementation. For example, the time dependence of using first-order (respectively, second-order) Trotter formula can be improved from $t^2$ (respectively, $t^{1+1/2}$) to $t^{1+1/3}$ (respectively, $t^{1+1/5}$) by adding LCU compensation.

\section{TROTTER-LCU FORMULA WITH NESTED-COMMUTATOR COMPENSATION} \label{sec:NCC}

In this section, we provide detailed construction of the nested-commutator compensation Trotter-LCU algorithms and the gate complexity analysis.
We first sketch the procedure to derive the nested-commutator form LCU formula in \autoref{ssc:NCDerive}. Then we construct the LCU formula of  of the first-order Trotter remainder of the lattice Hamiltonians in \autoref{ssc:NestedCommLattice1st} as an example. In \autoref{ssc:NestedCommSampling}, we describe the random-sampling implementation of the nested-commutator compensation algorithm and analyze its gate complexity.

\subsection{Derivation of the nested-commutator formula} \label{ssc:NCDerive}

Our aim is to expand the LCU formula for the $K$th-order Trotter remainder $V_K(x):= U(x) S_K(x)^\dag$ ($K=1$ or $2k$, $k\in\mbb{N}_+$) in the following form:
\begin{equation} \label{eq:VKxDef}
V_K(x) = I + \sum_{s=K+1}^{2K+1} F_{K,s}^{(nc)}(x) + F_{K,\mc{res}}^{(nc)}(x),
\end{equation}
where the leading-order terms $F_{K,s}^{(nc)}(x)$ are written as a summation of nested commutators. We will use the following lemma of the operator-valued differential equation.

\begin{lemma}[(Lemma A.1 in \cite{childs2021theory}] \label{lem:variation_of_para}
Let $H(x)$, $R(x)$ be continuous operator-valued functions defined for $x\in\mbb{R}$. Then, the first-order differential equation
\begin{equation} \label{eq:ddxWx}
\frac{d}{dx} W(x) = H(x) W(x) + R(x), \quad W(0) \text{ known},
\end{equation}
has a unique solution given by
\begin{equation}
\begin{aligned}
&W(x) = \mc{T} \exp\left( \int_0^x d\tau H(\tau) \right) W(0) \\
& \quad + \int_0^x d\tau_1 \mc{T} \exp\left( \int_{\tau_1}^x d\tau_2 H(\tau_2) \right) R(\tau_1).
\end{aligned}
\end{equation}
Here, $\mc{T}$ is the time-ordering operator.
\end{lemma}
In \autoref{eq:ddxWx}, if we set $W(x)$ to be the real-time evolution $U(x)=e^{-iHx}$, we can find that $H(x)=-iH$ and $R(x)=0$. Therefore, $R(x)$ reflects the derivation of $W(x)$ from the exponential function.
We are going to apply Lemma~\ref{lem:variation_of_para} to $S_K(x)^\dag$, i.e., we set $W(x):=S_K(x)^\dag$ and $H(x)=iH$. The deviation of $S_K(x)^\dag$ from $U(x)^\dag=e^{ixH}$ is characterized by the following function,
\begin{equation} \label{eq:RKxDef}
R_K(x) = \frac{d}{dx} S_K(x)^\dag - iH S_K(x)^\dag.
\end{equation}

Applying Lemma~\ref{lem:variation_of_para}, we have
\begin{equation}
S_K(x)^\dag = e^{ixH} + \int_0^x d\tau e^{i(x-\tau)H} R_K(\tau).
\end{equation}
The Trotter remainder $V_K(x)$ can then be expressed as
\begin{equation} \label{eq:VKxRecurrence}
\begin{aligned}
V_K(x) &= U(x) S_K(x)^\dag = I + \int_0^x d\tau e^{-i\tau H} R_K(\tau) \\
&= I + \int_0^x d\tau U(\tau) S_K(\tau)^\dag S_K(\tau) R_K(\tau) \\
&= I + \int_0^x d\tau V_K(\tau) J_K(\tau),
\end{aligned}
\end{equation}
where
\begin{equation} \label{eq:JKxDef}
J_K(x) := S_K(x) R_K(x).
\end{equation}
\autoref{eq:VKxRecurrence} provides a recurrence formula to solve the expansion terms in $V_K(x)$. To be more explicit, if we expand $V_K(x)$ and $J_K(x)$ based on the operator-valued Taylor series,
\begin{equation}
V_K(x) = \sum_{s=0}^\infty G_s \frac{x^s}{s!},\quad 
J_K(x) = \sum_{s=0}^\infty C_s \frac{x^s}{s!},
\end{equation}
where $G_s$ and $C_s$ denotes the respective $s$-order term, then we have
\begin{equation} \label{eq:FsKthRecurrence}
\begin{aligned}
G_{s+1} &= \frac{(s+1)!}{x^{s+1}} \int_0^{x}d\tau \sum_{m=0}^{s} G_m C_{s-m} \frac{\tau^s}{m!(s-m)!} \\
&= \sum_{m=0}^{s} \binom{s}{m} G_m C_{s-m}
\end{aligned}
\end{equation}
since \autoref{eq:VKxRecurrence} holds for all $x\in\mbb{R}$. This is a recurrence formula which can be used to solve all the expansion terms $G_s$ of the remainder $V_K(x)$ from the expansion terms $C_r$ of $J_K(x)$.

To solve the explicit form of $G_s$, we need to study the form of expansion terms $\{C_r\}_{r=0}^{s-1}$ for the function $J_K(x)$. 
We will use the following proposition derived from \autoref{lem:order}.

\begin{proposition}[Order condition] \label{prop:order}
For the $K$th-order Trotter formula $S_K(x)$, the multiplicative remainder $V_K(x)$ and derivative remainder $J_K(x)$ defined in \autoref{eq:VKxDef}, \autoref{eq:RKxDef} and \autoref{eq:JKxDef}, respectively, the following statements are equivalent.
\begin{enumerate}[(1)]
\item $S_K(x) = U(x) + \mc{O}(x^{K+1})$,
\item $(S_K(0))^{(j)}=(-iH)^j$, for $0\leq j\leq K$,
\item $J_K(x) = \mc{O}(x^K)$,
\item $(J_K(0))^{(j)}=0$, for $0\leq j\leq K-1$,
\item $V_K(x) = I + \mc{O}(x^{K+1})$,
\item $V_K(0)=I$ and $(V_K(0))^{(j)}=0$, for $1\leq j\leq K$.
\end{enumerate}
Here, $f(x)^{(j)}$ is the $j$th-derivative of the function $f(x)$.
\end{proposition}

\begin{proof}
First, we note that $(1)\Leftrightarrow (2)$, $(3)\Leftrightarrow (4)$, and $(5)\Leftrightarrow (6)$ by applying \autoref{lem:order} and setting $F(x)$ to be $S_K(x)-U(x)$, $J_K(x)$, and $V_K(x)-I$, respectively. So we only need to prove $(2)\Leftrightarrow (4)$ and $(2)\Leftrightarrow (6)$.

We first prove $(2)\Leftrightarrow (4)$. From $(2)$ we also have $(S_K(x)^\dag)^{(j)}=(iH)^j$ for $0\leq j\leq K$. Based on \autoref{eq:RKxDef} and \autoref{eq:JKxDef}, the derivatives of $R_K(x)$ and $J_K(x)$ are
\begin{equation} \label{eq:RtJt}
\begin{aligned}
(R_K(x))^{(j)} &= (S_K(x)^\dag)^{(j+1)} - iH (S_K(x)^\dag)^{(j)}, \\
(J_K(x))^{(j)} &= \sum_{l=0}^j \binom{j}{l} (S_K(x))^{(l)} (R_K(x))^{(j-l)}.
\end{aligned}
\end{equation}
Based on $(2)$ we have $(R_K(0))^{(j)}=0$ for $0\leq j\leq K-1$. Hence $(J_K(0))^{(j)}=0$ for $0\leq j\leq K-1$. 

For the reverse direction, we notice that $R_K(x)=J_K(x)S_K(x)^\dag$. This implies
\begin{equation}
\begin{aligned}
(R_K(x))^{(j)} &= \sum_{l=0}^j \binom{j}{l} (J_K(x))^{(l)} (S_K(x)^\dag)^{(j-l)}.
\end{aligned}
\end{equation}
If we have $(J_K(0))^{(j)}=0$ for $0\leq j\leq K-1$, then $(R_K(0))^{(j)}=0$ for $0\leq j\leq K-1$. Then, from \autoref{eq:RtJt} we have $(S_K(x)^\dag)^{(j)}=(iH)^j$ for $0\leq j\leq K$, which is equivalent to $(2)$.

Now, we prove $(2)\Leftrightarrow (6)$. Based on \autoref{eq:VKxDef} we have
\begin{equation}
(V_K(x))^{(j)} = \sum_{l=0}^j \binom{j}{l} (U(x))^{(l)} (S_K(x)^\dag)^{(j-l)}.
\end{equation}
From $(2)$ we have
\begin{equation}
(V_K(0))^{(j)} = \sum_{l=0}^j \binom{j}{l} (-iH)^l (iH)^{j-l} = 0,
\end{equation}
for $1<j<K$. The reverse direction can be proven similarly based on the derivative of the formula $S_K(x) = V_K(x)^\dag U(x)$.
\end{proof}

From Proposition~\ref{prop:order}, we have the following order condition for the $K$th-order Trotter formula and remainders,
\begin{equation} \label{eq:SKJKVKorder}
\begin{aligned}
S_K(x) &= U(x) + \mc{O}(x^{K+1}), \\
J_K(x) &= \mc{O}(x^K), \\
V_K(x) &= I + \mc{O}(x^{K+1}).
\end{aligned}
\end{equation}

Compare the recurrence formula in \autoref{eq:VKxRecurrence} and the order condition in \autoref{eq:SKJKVKorder}, we will obtain the following relationship for the Taylor-series expansion terms $G_s$ and $C_s$ for $V_K(x)$ and $J_K(x)$, respectively,
\begin{equation} \label{eq:GsCsOrder}
\begin{aligned}
G_s &= 0, \quad s = 1, 2, ..., K, \\
C_s &= 0, \quad s = 0, 1, ..., K-1, \\
G_s &= C_{s-1}, \quad s = K+1, K+2, ..., 2K+1.
\end{aligned}
\end{equation}

Based on \autoref{eq:GsCsOrder}, we are going to expand $J_K(x)$ based on the operator-valued Taylor-series expansion with integral remainders,
\begin{equation} \label{eq:JKxExpand}
J_K(x) = J_{K,L}(x) + J_{K,res,2K}(x) = \sum_{s=K}^{2K} C_s \frac{x^s}{s!} + J_{K,res,2K}(x),
\end{equation}
where
\begin{equation} \label{eq:CsJresidueKth}
\begin{aligned}
&C_s = J_K^{(s)}(0), \quad s=K,K+1,...,2K, \\
&J_{K,res,s}(x) = \int_0^x d\tau \frac{(x-\tau)^{s}}{s!} J_K^{(s+1)}(\tau).
\end{aligned}
\end{equation}
Here, $J_{K,L}(x):= \sum_{s=K}^{2K} C_s \frac{x^s}{s!}$ denotes the leading-order terms in $J_K(x)$. Then, from \autoref{eq:VKxRecurrence} and \autoref{eq:GsCsOrder}, the $K$th-order Trotter remainder can be expressed as
\begin{equation} \label{eq:VKxExpandNC}
\begin{aligned}
V_K(x) &= I + M_K(x) = I + \sum_{s=K+1}^{2K+1} F_{K,s}^{(nc)}(x) + F_{K,\mc{res}}^{(nc)}(x), \\
F_{K,s}^{(nc)}(x) &= C_{s-1} \frac{x^s}{s!}, \quad s= K+1, K+2, ..., 2K+1, \\
F_{K,\mc{res}}^{(nc)}(x) &= \int_0^x d\tau \left( M_K(\tau) J_{K,L}(\tau) + V_K(\tau) J_{K,res,2K}(\tau) \right).
\end{aligned}
\end{equation}
Here, $M_K(x):=V_K(x) - I$. 

From \autoref{eq:VKxExpandNC} we can see that, the leading-order expansion terms $F_{K,s}^{(nc)}(x)$ of $V_K(x)$ with $s=K+1, K+2, ..., 2K+1$ owns a simple expression related to $C_{s-1}:=J_K^{(s-1)}(0)$. Later we will show that, $C_s$ can be simply written as a summation of nested commutators with the form, 
\begin{equation}
\ad_{H_{l_j}}^{m_j}... \ad_{H_{l_2}}^{m_2} \ad_{H_{l_1}}^{m_1} H_l,
\end{equation}
where $\sum_j m_j = s$, $H_l, H_{l_1}, ..., H_{l_j}$ are different summands in $H$.
Furthermore, we will show that $\{C_s\}$ are all anti-Hermitian. As a result, the expansion-order-pairing based on Euler's formula introduced in \autoref{sec:PTSC} can also be applied here. 
Since $\{C_s\}$ are anti-Hermitian, we can expand them by the Pauli operators as follows:
\begin{equation}
C_s = i \|C_s\|_1 \sum_{\gamma} \Pr(\gamma|s) P^{(nc)}_s(\gamma),
\end{equation}
where $\{P^{(nc)}_s(\gamma)\}$ are Hermitian Pauli operators with coefficients $1$ or $-1$. $\|C_s\|_1$ denotes the $1$-norm of this expansion. The leading-order expansion term $F_{K,s}^{(nc)}(x)$ can then be expressed as
\begin{equation} \label{eq:FKsNCexpand}
\begin{aligned}
F_{K,s}^{(nc)}(x) &= i \|C_{s-1}\|_1 \frac{x^s}{s!} \sum_{\gamma} \Pr(\gamma|s-1) P^{(nc)}_{s-1}(\gamma) \\
&=: \eta^{(nc)}_s V_{K,s}^{(nc)}.
\end{aligned}
\end{equation}
Here, $\eta^{(nc)}_s:= \|C_{s-1}\|_1 \frac{x^s}{s!}$ is the $1$-norm of the LCU formula of $F_{K,s}^{(nc)}(x)$ in \autoref{eq:FKsNCexpand}. $V_{K,s}^{(nc)}$ is the normalized LCU formula.

The residue term $F_{K,\mc{res}}^{(nc)}$ in \autoref{eq:VKxExpandNC}, however, is complicated and hard to be expressed simply using nested commutators. Due to the hardness to compensate $F_{K,\mc{res}}^{(nc)}$, we will remove it in the truncated Trotter remainder formula,
\begin{equation} \label{eq:tildeVKxNCpair}
\begin{aligned}
& \tilde{V}_K^{(nc)}(x) = I + \sum_{s=K+1}^{2K+1} F_{K,s}^{(nc)}(x) \\
&= \sum_{s=K}^{2K} \frac{\eta^{(nc)}_s}{\eta_\Sigma^{(nc)}} \left( I + \eta^{(nc)}_\Sigma V_{K,s}^{(nc)} \right)  \\
&= \sqrt{1+(\eta_\Sigma^{(nc)})^2} \sum_{s=K}^{2K} \frac{\eta^{(nc)}_s}{\eta_\Sigma^{(nc)}} R_{K,s}^{(nc)}(\eta_\Sigma).
\end{aligned}
\end{equation}
Here, $\eta_\Sigma^{(nc)}:= \sum_{s=K+1}^{2K+1}\|C_{s-1}\|_1 \frac{x^s}{s!}$. 
In the fourth line, we apply the following pairing procedure based on Euler's formula, similar to \autoref{eq:V1xPTSPara},
\begin{equation} \label{VKxNCpaired}
\begin{aligned}
&I + \eta_{\Sigma}^{(nc)} V_{K,s}^{(nc)} = I + i \eta_{\Sigma}^{(nc)} \sum_{\gamma} \Pr(\gamma|s-1) P^{(nc)}_{s-1}(\gamma) \\
&= \sqrt{1+ (\eta_\Sigma^{(nc)})^2}  \sum_{\gamma} \Pr(\gamma|s-1)  \exp\left( i\theta(\eta_\Sigma^{(nc)}) P^{(nc)}_{s-1}(\gamma) \right) \\
&=: \sqrt{1+ (\eta_\Sigma^{(nc)})^2} R_{2,s}^{(nc)}(\eta_\Sigma^{(nc)}),
\end{aligned}
\end{equation}
for $s=K+1,K+2,...,2K+1$. Recall that $\theta(x):=\tan^{-1}(x)$.

$\tilde{V}_K^{(nc)}(x)$ in \autoref{eq:tildeVKxNCpair} is the final LCU formula for the nested-commutator compensation of $V_K(x)$. In what follows, we estimate the $1$-norm $\mu_K^{(nc)}(x)$ and distance $\varepsilon_K^{(nc)}(x)$ of this LCU formula,
\begin{equation}
\begin{aligned}
\mu_K^{(nc)}(x) &:= \sqrt{1+ (\eta_\Sigma^{(nc)})^2}, \\
\varepsilon_K^{(nc)}(x) &:= \|\tilde{V}_K^{(nc)}(x) - V_K(x)\| = \|F_{K,\mc{res}}^{(nc)}(x)\|.
\end{aligned}
\end{equation}

\begin{proposition}[Bound the 1-norm and error of nested-commutator expansion formula] \label{prop:NCnormBound}
$\tilde{V}^{(nc)}_K(x)$ in \autoref{eq:tildeVKxNCpair} is a $(\mu_K^{(nc)}(x), \varepsilon_K^{(nc)}(x))$ formula of the $K$th-order Trotter remainder $V_K(x)$ with
\begin{equation}
\begin{aligned}
\mu_K^{(nc)}(x) &= \sqrt{1+ (\eta_\Sigma^{(nc)})^2} = \sqrt{1 + \left( \sum_{s=K}^{2K} \|C_s\|_1 \frac{x^{s+1}}{(s+1)!} \right)^2} \\
\varepsilon_K^{(nc)}(x) &\leq \int_0^x d\tau \left( \|M_K(\tau)\| \|J_{K,L}(\tau)\| + \|J_{K,res,2K}(\tau)\| \right),
\end{aligned}
\end{equation}
where
\begin{equation} \label{eq:JKLJKresMKbound}
\begin{aligned}
\|J_{K,L}(\tau)\| &\leq \sum_{s=K}^{2K} \|C_s\| \frac{x^{s+1}}{(s+1)!},  \\
\|J_{K,res,2K}(\tau)\| &\leq \int_0^x d\tau \frac{(x-\tau)^s}{s!} \|J_K^{(s+1)}(\tau)\|, \\
\|M_K(\tau)\| &\leq \int_0^\tau d\tau_1 \int_0^{\tau_1} d\tau_2 \frac{(\tau_1-\tau_2)^{(K-1)}}{(K-1)!} \|J_K^{(K)}(\tau_2)\|.
\end{aligned}
\end{equation}

\end{proposition}

\begin{proof}

The value of $\mu_K^{(nc)}(x)$ is derived based on \autoref{eq:FKsNCexpand} and \autoref{eq:tildeVKxNCpair}. 
To calculate $\varepsilon_K^{(nc)}(x)$, we will use the following bound,
\begin{equation} \label{eq:EpsKNCbound}
\begin{aligned}
&\varepsilon_K^{(nc)}(x) = \|F_{K,\mc{res}}^{(nc)}(x)\| \\
&\leq \int_0^x d\tau \left( \|M_K(\tau)\| \|J_{K,L}(\tau)\| + \|V_K(\tau)\| \|J_{K,res,2K}(\tau)\| \right) \\
&= \int_0^x d\tau \left( \|M_K(\tau)\| \|J_{K,L}(\tau)\| + \|J_{K,res,2K}(\tau)\| \right).
\end{aligned}
\end{equation}
In the second equality, we use the fact that $V_K(\tau)$ is a unitary. To bound $\|J_{K,L}(\tau)\|$, we have
\begin{equation}
\|J_{K,L}(\tau)\| \leq \sum_{s=K}^{2K} \|C_s\| \frac{x^{s+1}}{(s+1)!}.
\end{equation}

To bound $J_{K,res,2K}(\tau)$ in \autoref{eq:EpsKNCbound}, from \autoref{eq:CsJresidueKth} we have
\begin{equation}
\|J_{K,res,2K}(\tau)\| \leq \int_0^x d\tau \frac{(x-\tau)^s}{s!} \|J_K^{(s+1)(\tau)}\|.
\end{equation}

Finally, to bound $\|M_K(\tau)\|$ in \autoref{eq:EpsKNCbound}, we have
\begin{equation} \label{eq:MKtauBound}
\begin{aligned}
\|M_K(\tau)\| &= \left\|\int_0^\tau d\tau_1 V_K(\tau_1) J_K(\tau_1) \right\| \\
&\leq \int_0^\tau d\tau_1 \|J_K(\tau_1)\| = \int_0^\tau d\tau_1 \|J_{K,res,K-1}(\tau_1)\| \\
&\leq \int_0^\tau d\tau_1 \int_0^{\tau_1} d\tau_2 \frac{(\tau_1-\tau_2)^{(K-1)}}{(K-1)!} \|J_K^{(K)}(\tau_2)\|.
\end{aligned}
\end{equation}
In the first line, we use the definition of $M_K(\tau)=V_K(\tau)-I$ and the recurrence formula in \autoref{eq:VKxRecurrence}. In the second line, we use the property that $V_K(\tau)$ is a unitary. In the third line, we use the order condition in Proposition~\ref{prop:order}. In the final line, we apply the operator-valued Taylor-series expansion on $J_K(\tau)$ and set the truncation order $s_c=K-1$.

\end{proof}

From \autoref{prop:NCnormBound} we can see that, to study the performance of the truncated LCU formula $\tilde{V}^{(nc)}_K(x)$ in \autoref{eq:tildeVKxNCpair}, we only need to study the property of the derivatives of $J_K(x)$, including $\{C_s\}$. 
In the following section, we are going to derive the explicit formula of $\tilde{V}^{(nc)}_K(x)$, taking the lattice Hamiltonian with first-order Trotter formulas as an example.

\subsection{Example: first-order lattice Hamiltonian} \label{ssc:NestedCommLattice1st}

Now, we focus on the lattice Hamiltonians with the form $H=A+B$ in \autoref{eq:HLattice}. For the lattice Hamiltonian, the first-order Trotter formula is
$S_1(x) = e^{-ixB} e^{-ixA}$.
Here, we assume that the time-evolution of each two-qubit component $e^{-ix H_{j,j+1}}$ is easy to be implement on the quantum computer.

To derive the explicit form of the LCU formula $\tilde{V}_1^{(nc)}(x)$ in \autoref{eq:tildeVKxNCpair}, we first derive $J_1(x)$ defined in \autoref{eq:JKxDef} and its derivatives $C_s:= J_1^{(s)}(x)$. From \autoref{eq:RKxDef} and \autoref{eq:JKxDef} we have
\begin{equation}
\begin{aligned}
R_1(x) &= \frac{d}{dx} S_1(x)^\dag - iH S_1(x)^\dag = i [e^{ixA},B] e^{ixB}, \\
J_1(x) &= S_1(x) R_1(x) = i e^{-ixB}e^{-ixA} [e^{ixA},B] e^{ixB} \\
&= i\left( B - e^{-ix \ad_B}e^{-ix \ad_A}B \right).
\end{aligned}
\end{equation}

Applying the general Libniz formula to $J_1(x)$ we have,
\begin{equation} \label{eq:J1sderi}
J_1^{(s)}(x) = (-i)^{s+1} \sum_{\substack{m_1,n_1;\\m_1+n_1=s} } \binom{s}{m_1} e^{-ix \ad_B} \ad_B^{m_1} e^{-ix \ad_A} \ad_A^{n_1} B.
\end{equation}

If we set the truncation of $J_1^{(s)}(x)$ to be $s_c$ and apply the following operator Taylor-series expansion formula,
\begin{equation} \label{eq:operatorTaylor}
Q(x) = \sum_{s=0}^{s_c} \frac{x^s}{s!} Q^{(s)}(0) + \int_0^x d\tau \frac{(x-\tau)^{s_c}}{s_c!} Q^{(s_c+1)}(\tau),
\end{equation}
we can expand $J_1(x)$ as
\begin{equation}
J_1(x) = \sum_{s=0}^{s_c} C_s \frac{x^s}{s!} + J_{1,\mc{res}}(x),
\end{equation}
where
\begin{equation} \label{eq:CsJtresidue1st}
\begin{aligned}
C_s = J_1^{(s)}(0) &= (-i)^{s+1} \sum_{\substack{m_1,n_1;\\m_1+n_1=s} } \binom{s}{m_1} \ad_B^{m_1} \ad_A^{n_1} B, \\
J_{1,res,s_c}(x) &= \int_0^x d\tau \frac{(x-\tau)^{s_c}}{s_c!} J_1^{(s_c+1)}(\tau).
\end{aligned}
\end{equation}
We can see that, $C_s$ can be written as the summation of nested commutators with concise form $\ad_B^{m_1} \ad_A^{n_1} B$. Note that, $(-i)^{s+1} \ad_B^{m_1} \ad_A^{n_1} B$ with $m_1+n_1=s$ is always anti-Hermitian when $A$ and $B$ are Hermitian. 
On the other hand, these nested commutators are all with the nice property that their spectral norm and $1$-norm is linear to the system size $n$. To be more specific, we have the following norm bound.

\begin{proposition}
\label{prop:adnorm1st}
Consider a lattice Hamiltonian $H=A+B$ with the form in \autoref{eq:HLattice}. Suppose the spectral norm and $1$-norm of its components $H_{j,j+1}$ are upper bounded by $\Lambda$ and $\Lambda_1$. Then for the nested commutators appearing in \autoref{eq:CsJtresidue1st}, we have the following bound
\begin{equation} \label{eq:adnormAB1st}
\begin{aligned}
\left\| e^{-i \tau \ad_B} \ad_B^{m_1} e^{-i\tau \ad_A} \ad_A^{n_1} B \right\| \leq \frac{n}{2} 3^{m_1} 2^{n_1} 2^s \Lambda^{s+1},
\end{aligned}
\end{equation}
where $m_1,n_1$ are non-negative integers satisfying $m_1+n_1 = s$. As a result, we can bound the spectral norm of $J_1^{(s)}$ as
$\|J_1^{(s)}(x)\| \leq \frac{n}{2} 10^s \Lambda^{s+1}$.
The $1$-norm upper bound is to simply replace $\Lambda$ by $\Lambda_1$.
\end{proposition}

\begin{proof}
We first focus on one Hamiltonian term $H_{j,j+1}$ contained in $B$ and bound the norm,
\begin{equation} \label{eq:adnormH1st}
\left\| e^{-i\tau \ad_B} \ad_B^{m_1} e^{-i\tau \ad_A} \ad_A^{n_1} (H_{j,j+1}) \right\| \leq 3^{m_1}2^{n_1} 2^s \Lambda^{s+1}.
\end{equation} 
To do this, we are going to decompose commutator to the elementary nested commutators in the following form:
\begin{equation} \label{eq:elementNC1st}
\begin{aligned}
& e^{-i\tau \ad_B} \ad_{H_{j_s,j_s+1}} ... \ad_{H_{j_{n_1+1},j_{n_1+1}+1}} \cdot \\
& \; e^{-i\tau \ad_A} \ad_{H_{j_{n_1},j_{n_1}+1}} ... \ad_{H_{j_1,j_1+1}} H_{j,j+1},
\end{aligned}
\end{equation}
where $j_1,j_2,...,j_s$ are the possible vertice's indices. For each elementary nested commutator, the spectral norm can be easily bounded by 
$(2\Lambda)^s \Lambda$
by simply expanding all the commutators and applying triangle inequality. Here, we use the property that the spectral norm of all the exponential operators with anti-Hermitian exponent is $1$.

Now, we count the number of the possible elementary commutators with the form in \autoref{eq:elementNC1st}. We will check the action of each adjoint operator $\ad_A$ or $\ad_B$ from the right to the left. For the first location, if we expand $\ad_A$, there will be only two possible elementary nonzero components, $\ad_{H_{j-1,j}}$ and $\ad_{H_{j+1,j+2}}$. If the next $\ad$ is still $\ad_A$, the support will still be on the four qubits: $j-1,j,j+1$, and $j+2$. As a result, there will still be only two possible components, $\ad_{H_{j-1,j}}$ and $\ad_{H_{j+1,j+2}}$. Similarly, the exponential operator $e^{-i\tau \ad_A}$ will not enlarge the support since one can expand it to the power of $\ad_A$.
The support will be enlarged when $\ad_B$ comes. In this layer, the support of the operator will be expanded to six qubits, and there will be three elementary components. The number of possible elementary commutators is then
\begin{equation}
3^{m_1} 2^{n_1}.
\end{equation}
Combining the number of elementary nested commutators and the norm bound for each commutator and applying triangle inequality, we will obtain \autoref{eq:adnormH1st}.
Finally, in the operator $B$, there are $\frac{n}{2}$ possible summands. This finishes the proof of \autoref{eq:adnormAB1st}.

Now, we apply \autoref{eq:adnormAB1st} to bound the norm of $J_1^{(s)}(x)$. From \autoref{eq:J1sderi} we have
\begin{equation}
\begin{aligned}
\|J_1^{(s)}(x)\| &\leq \sum_{\substack{m_1,n_1;\\m_1+n_1=s} } \binom{s}{m_1} \| e^{-ix \ad_B} \ad_B^{m_1} e^{-ix \ad_A} \ad_A^{n_1} B\| \\
&\leq \frac{n}{2} 2^s \Lambda^{s+1} \sum_{\substack{m_1,n_1;\\m_1+n_1=s} } \binom{s}{m_1}  3^{m_1} 2^{n_1} \\
&= \frac{n}{2} 10^s \Lambda^{s+1}.
\end{aligned}
\end{equation}
In the third line, we apply the binomial theorem.

Since $1$-norm can be estimated based on the same logic by counting the number of elementary nested commutators and the $1$-norm of each nest commutator, the derivation for the $1$-norm is similar by replacing $\Lambda$ to $\Lambda_1$.
\end{proof}

As introduced in \autoref{ssc:NCDerive}, when $K=1$, we set $F_{1,s}^{(nc)}(x)$ with $s=2,3$ to be the leading-order terms and set the truncation order $s_c=3$. That is, we only compensate the second- and third-order error using LCU methods. While we are not able to achieve the logarithmic accuracy similar to PTSC algorithms, we can achieve a high accuracy of $\mc{O}(\varepsilon^{-1/3})$, which is cubicly improved comparing to the bare first-order Trotter result $\mc{O}(\varepsilon^{-1})$.

Based on \autoref{eq:tildeVKxNCpair}, the truncated nested-commutator LCU formula for $V_1(x)$ can be written as
\begin{equation} \label{eq:V1xNCexpand}
\begin{aligned}
\tilde{V}_1(x) &= I + \sum_{s=2}^3 F_{1,s}^{(nc)}(x), \\
&=\sqrt{1+ (\eta_\Sigma^{(nc)})^2} \sum_{s=2}^3 \frac{\eta_s^{(nc)}}{\eta_\Sigma^{(nc)}} R_{1,s}^{(nc)}(\eta_\Sigma).
\end{aligned}
\end{equation}
Here, $\eta_\Sigma^{(nc)}:= \sum_{s=2}^{3}\|C_{s-1}\|_1 \frac{x^s}{s!}$. The explicit form of $R_{1,s}^{(nc)}(\eta_\Sigma)$ can be obtained by the definitions in \autoref{eq:FKsNCexpand}, \autoref{VKxNCpaired} and the Pauli operator decomposition based on the nested-commutator form in \autoref{eq:CsJtresidue1st}.

Combined with the deterministic first-order Trotter formula, the overall LCU formula for $U(x)$ is
\begin{equation} \label{eq:tildeU1xNC}
\begin{aligned}
\tilde{U}_1^{(nc)}(x) = \tilde{V}_1^{(nc)}(x) S_1(x).
\end{aligned}
\end{equation}

Following \autoref{prop:NCnormBound} and \autoref{prop:adnorm1st}, we can bound the $1$-norm and error of the LCU formula in \autoref{eq:V1xNCexpand} and \autoref{eq:tildeU1xNC} as follows.

\begin{proposition}[first-order Trotter-LCU formula by nested-commutator compensation for lattice Hamiltonians] \label{prop:NC1}

Consider a lattice Hamiltonian with the form in \autoref{eq:HLattice}. For $\min\{\frac{1}{4\Lambda}, \frac{3}{20\Lambda_1} \}>x>0$, \autoref{eq:V1xNCexpand} is a $(\mu_1^{(nc)}(x),\varepsilon_1^{(nc)}(x))$-LCU formula of $V_1(x)$ with
\begin{equation}
\begin{aligned}
\mu_1^{(nc)}(x) &\leq \exp\left(10 n^2 (\Lambda_1 x)^4\right), \\
\varepsilon_1^{(nc)}(x) &\leq 15 n^2 (\Lambda x)^4.
\end{aligned}
\end{equation}
As a result, $\tilde{U}_1^{(nc)}(x)$ in \autoref{eq:tildeU1xNC} is a $(\mu_1^{(nc)}(x),\varepsilon_1^{(nc)}(x))$-LCU formula of $U(x)$. Here, $\Lambda_1$ and $\Lambda$ are defined in \autoref{eq:LatticeHnormbound}. 
\end{proposition}

\begin{proof}

We start from bounding the $1$-norm $\mu_1^{(nc)}(x)$. From \autoref{prop:NCnormBound} we have
\begin{equation}
\begin{aligned}
\mu_1^{(nc)}(x) &\leq \sqrt{1 + \left(\sum_{s=1}^2 \|C_s\|_1 \frac{x^{s+1}}{(s+1)!} \right)^2}  \\
&\leq 1 + \frac{1}{2} \left(\sum_{s=1}^2 \|C_s\|_1 \frac{x^{s+1}}{(s+1)!} \right)^2 \\
&= 1 + n^2\left(\frac{25}{8}(\Lambda_1 x)^4 + \frac{250}{12}(\Lambda_1 x)^5 + \frac{1250}{3!3!}(\Lambda_1 x)^6 \right) \\
&\leq 1 + 3 n^2 \frac{25}{8}(\Lambda_1 x)^4 \leq \exp\left( 10 n^2 (\Lambda_1 x)^4 \right). 
\end{aligned}
\end{equation}
In the third line, we use the fact that $C_s = J_1^{(s)}(0)$ and \autoref{prop:adnorm1st}. In the fourth line, we use the assumption that $\Lambda_1 x \leq \frac{3}{20}$. 

Now, we bound the spectral norm distance $\varepsilon_1^{(nc)}(x)=\|F_{1,\mc{res}}^{(nc)}\|$. From \autoref{prop:NCnormBound} we know that we only need to bound $\|J_{1,L}(\tau)\|$, $\|J_{1,res,2}(\tau)\|$, and $M_1(\tau)$ based on \autoref{eq:JKLJKresMKbound}.
For $\|J_{1,L}(\tau)\|$, from \autoref{prop:adnorm1st} we have
\begin{equation} \label{eq:J1Lbound}
\begin{aligned}
\|J_{1,L}(\tau)\| &\leq \sum_{s=1}^{2} \|C_s\| \frac{\tau^s}{s!} \leq \frac{n}{2} \sum_{s=1}^{2} 10^s \Lambda^{s+1} \frac{\tau^s}{s!}.
\end{aligned}
\end{equation}
For $\|J_{1,res,2}(\tau)\|$ and $\|M_1(\tau)\|$, from \autoref{eq:JKLJKresMKbound} and \autoref{prop:adnorm1st} we have
\begin{equation} \label{eq:J1res2M1bound}
\begin{aligned}
\|J_{1,res,2}(\tau)\| &\leq \int_0^\tau d\tau_1 \frac{(\tau-\tau_1)^2}{2!} \|J_1^{(3)}(\tau_1)\| \\
&\leq \frac{n}{2}10^3 \Lambda^4 \frac{\tau^3}{3!}, \\
\|M_1(\tau)\| &\leq \int_0^\tau d\tau_1 \int_0^{\tau_1} d\tau_2 \|J_1^{(1)}(\tau_2)\| \\
&\leq \frac{n}{2}10 \Lambda^2 \frac{\tau^2}{2!}.
\end{aligned}
\end{equation}
Based on \autoref{eq:J1Lbound}, \autoref{eq:J1res2M1bound} and \autoref{prop:NCnormBound},
\begin{equation}
\begin{aligned}
\|F_{1,\mc{res}}^{(nc)}\| &\leq \int_0^x d\tau \left( \|M_1(\tau)\|\|J_{1,L}(\tau)\| + \|J_{1,res,2}(\tau)\| \right) \\
&\leq \int_0^x d\tau \left( \frac{n^2}{4}\sum_{s=1}^2 10^{s+1}\Lambda^{s+3} \frac{\tau^{s+2}}{s!2!} + \frac{n}{2}10^3 \Lambda^4 \frac{\tau^3}{3!} \right) \\
&= \frac{n^2}{4}\sum_{s=1}^2 10^{s+1} \frac{(\Lambda x)^{s+2}}{s!2!(s+3)} + \frac{n}{2}10^3 \frac{(\Lambda x)^4}{4!} \\
&\leq 2\cdot \frac{n^2}{4} 10^2\frac{(\Lambda x)^4}{8} + \frac{n}{2}10^3 \frac{(\Lambda x)^4}{4!} \\
&\leq 15n^2 (\Lambda x)^4.
\end{aligned}
\end{equation}
In the fourth line, we use the assumption that $x\leq \frac{1}{4\Lambda}$.
\end{proof}

\subsection{General construction and performance}
\label{ssc:NestedCommSampling}

We can easily extend the first-order analysis above to the higher-order case. 
In Appendix~\ref{sec:AppNested2nd}, we provide an explicit construction for second-order nested-commutator expansion, which will be used for the later numerical results.
For the general LCU formula for the $K$th-order Trotter remainder, we have the following proposition to characterize the LCU formulas $\tilde{V}_K(x)$ in \autoref{eq:tildeVKxNCpair}. 

\begin{proposition}[Trotter-LCU formula by nested-commutator compensation for lattice Hamiltonians] \label{prop:NCK}
Consider a lattice Hamiltonian with the form in \autoref{eq:HLattice}. We set $\beta:=2(4\kappa+5)$ where $\kappa$ is the stage number of the $K$th-order Trotter formula ($K=1$ or $2k$). For $\min\{\frac{K+1}{\beta\Lambda},\frac{K+2}{\beta\Lambda_1}\}>x>0$, $\tilde{V}_K^{(nc)}(x)$ in \autoref{eq:tildeVKxNCpair} is a $(\mu_K^{(nc)}(x),\varepsilon_K^{(nc)}(x))$-LCU formula of $V_K(x)$ with
\begin{equation}
\begin{aligned}
\mu_K^{(nc)} &\leq \exp\left( \frac{n^2 \kappa^2 \beta^{2(K+1)}}{2 (K!)^2}  (\Lambda_1 x)^{2K+2} \right), \\
\varepsilon_K^{(nc)} &\leq (n\kappa)^2 \frac{\beta^{2K+1}}{K!(K+1)!} (\Lambda x)^{2K+2}.
\end{aligned}
\end{equation}
As a result, $\tilde{U}_K^{(nc)}(x)$ in \autoref{eq:tildeVKxNCpair} is a $(\mu_K^{(nc)}(x),\varepsilon_K^{(nc)}(x))$-LCU formula of $U(x)$. Here, $\Lambda$ and $\Lambda_1$ are, respectively, the largest spectral norm and  $1$-norm of the lattice Hamiltonian components defined in \autoref{eq:LatticeHnormbound}.
\end{proposition}
The proof of \autoref{prop:NCK} is in Appendix~\ref{sec:AppNChigherTrotter}.

The circuit of random-sampling implementation of the NCC algorithm is similar to the PTSC algorithm, which is illustrated in \autoref{fig:SamplingNC}. The only difference is that we sample the Pauli operators based on the nested-commutator expansion formula. In Appendix~\ref{sec:AppNCSampling}, we generalize our random-sampling algorithm to the $K$th-order situation to demonstrate its scalability. Specifically, the space and time cost of the sampling algorithm are $\mc{O}(K\kappa_K)$ and $\mc{O}(K(\log{\kappa_K}+\log{n}))$, respectively.
In practice, we need to expand the leading-order terms $\{F_{K,s}\}$ to a summation of different adjoint operators based on the methods in \autoref{ssc:NCDerive} first, and then calculate the corresponding sampling probability $\Pr(\gamma|s)$.
We have the following theorem to characterize the gate complexity of the $K$th-order NCC algorithm with random-sampling implementation.

\begin{theorem}[Gate complexity of the $K$th-order random-sampling Trotter-LCU algorithm by nested-commutator compensation for lattice Hamiltonians] \label{thm:NC}
In a $K$th-order Trotter-LCU algorithm ($K=1$ or $2k$) based on nested-commutator compensation, if the segment number $\nu$ satisfy all the requirements below,
\begin{equation} \label{eq:nuBoundKthNC}
\begin{aligned}
\nu &\geq \max\left\{\frac{\beta\Lambda}{K+1}t, \frac{\beta\Lambda_1}{K+2}t \right\}, \\
\nu &\geq \left( \frac{\kappa^2}{2\ln\mu (K!)^2} \right)^{\frac{1}{2K+1}} \beta^{1+\frac{1}{2K+1}} n^{\frac{2}{2K+1}} (\Lambda_1 t)^{1+\frac{1}{2K+1}}, \\
\nu &\geq \left( \frac{\mu\kappa^2}{K!(K+1)!} \right)^{\frac{1}{2K+1}} \beta \varepsilon^{-\frac{1}{2K+1}} n^{\frac{2}{2K+1}} (\Lambda t)^{1+\frac{1}{2K+1}},
\end{aligned}
\end{equation}
we can then realize a $(\mu,\varepsilon)$-LCU formula for $U(t)$ based on $\nu$ segments of $\tilde{U}^{(nc)}_K(x)$ in \autoref{eq:tildeVKxNCpair}. As a result, the gate complexity of random-sampling $K$th-order Trotter-LCU algorithm based on nested-commutator compensation for the lattice Hamiltonian is 
\begin{equation}
\mc{O}\left( n^{1+\frac{2}{2K+1}} t^{1+\frac{1}{2K+1}} \varepsilon^{-\frac{1}{2K+1}} \right).
\end{equation}
Here, $\beta:=2(4\kappa+5)$ where $\kappa$ is the stage number of the Trotter formula. $\Lambda$ and $\Lambda_1$ are, respectively, the largest spectral norm and $1$-norm of the lattice Hamiltonian components defined in \autoref{eq:LatticeHnormbound}.
\end{theorem}

\begin{proof}

For the random-sampling implementation, the overall LCU formula for $U(t)$ is to repeat the sampling of $\tilde{U}_K^{(nc)}(x)$ for $\nu$ times, $\tilde{U}_{K,\mc{tot}}^{(nc)}(t) = \tilde{U}_K^{(nc)}(x)^\nu$. 
Using Proposition~\ref{prop:ProductRLCU} and \ref{prop:NCK}, when $0<x<\min\{\frac{K+1}{\beta\Lambda},\frac{K+2}{\beta\Lambda_1}\}$, we conclude that $\tilde{U}_{K,\mc{tot}}^{(nc)}(t)$ is a $(\mu_{K,\mc{tot}}^{(nc)}(t), \epsilon_{K,\mc{tot}}^{(nc)}(t) )$-LCU formula of $U(t)$ with
\begin{equation} \label{eq:muEpsKthNC}
\begin{aligned}
&\mu_{K,\mc{tot}}^{(nc)}(t) = \mu_K(x)^\nu \leq \exp\left( \frac{n^2 \kappa^2 \beta^{2(K+1)}}{2 (K!)^2} \frac{(\Lambda_1 t)^{2K+2}}{\nu^{2K+1}} \right), \\
&\varepsilon_{K,\mc{tot}}^{(nc)}(t) \leq \nu \mu_{K,\mc{tot}}^{(nc)}(t) \varepsilon_K^{(nc)}(x) \\
&\quad\quad \leq \nu \mu_{K,\mc{tot}}^{(nc)}(t)\left( (n\kappa)^2 \frac{\beta^{2K+1}}{K!(K+1)!}\frac{(\Lambda t)^{2K+2}}{\nu^{2K+2}} \right).
\end{aligned}
\end{equation} 

To realize a $(\mu,\varepsilon)$-LCU formula for $U(t)$, we only need the segment number $\nu$ satisfy all the requirements in \autoref{eq:nuBoundKthNC}.
It suffices to choice 
\begin{equation}
\nu = \mc{O}\left( n^{\frac{2}{2K+1}} t^{1+\frac{1}{2K+1}} \varepsilon^{-\frac{1}{2K+1}} \right).
\end{equation}

Based on \autoref{eq:gateNumRandom}, the gate complexity of the $K$th-order Trotter-LCU algorithm is then 
\begin{equation}
\mc{O}(\nu(\kappa_K L+ s_c)) = \mc{O}\left( n^{1+\frac{2}{2K+1}} t^{1+\frac{1}{2K+1}} \varepsilon^{-\frac{1}{2K+1}} \right).
\end{equation}
Here we use the fact that $s_c=2K+1=\mc{O}(1)$ and $L=\mc{O}(n)$ for lattice Hamiltonians.
\end{proof}

So far, we have restricted our construction of the nested-commutator expansion to lattice Hamiltonians. In practice, various physical Hamiltonians, including those for the electronic structure of quantum materials~\cite{babbush2018lowdepth}, quantum chemistry Hamiltonians~\cite{lee2021even,Berry2019qubitizationof,vonBurg2021quantum} and power-law interaction Hamiltonians~\cite{childs2021theory}, also possess sparse properties. As a result, the methods for the nested-commutator expansion of the Trotter remainder 
$V_K(x)$ introduced in \autoref{ssc:NCDerive} can also be applied to a general Hamiltonian $H$. 

In Appendix~\ref{sec:AppNCgeneral}, we discuss how to perform the nested-commutator expansion of $V_K(x)$ for general Hamiltonians, and discuss the performance of the resulting LCU formula in \autoref{prop:NCGenH}.
We find that the $1$-norm of the LCU formula based on nested-commutator expansion is closely related to the following nested commutator norm of the Hamiltonian $H=\sum_{l=1}^L H_l$,
\begin{equation}
\tilde{\alpha}_{\mr{com}}(H) := \sum_{l_{s+1}=1}^L ...\sum_{l_{2}=1}^L \| [H_{l_{s+1}},...[H_{l_2},H_l]]...] \|.
\end{equation}
$\tilde{\alpha}_{\mr{com}}(H)$ is originally defined in Ref.~\cite{childs2021theory} to analyze the performance of Trotter methods. In Ref.~\cite{childs2021theory}, the authors estimate the values of $\alpha^{(s)}_H$ for typical Hamiltonian models like plane-wave-basis quantum chemistry models, $k$-local Hamiltonian, and Hamiltonians with power-law interactions. Following similar estimation methods, we can also calculate $\tilde{\alpha}_{\mr{com}}(H)$ for different models and consider their explicit nested-commutator expansions. We will leave the explicit evaluation of other typical Hamiltonians for a future work.

\section{CONCLUSION AND OUTLOOK} \label{sec:conclusion}

We study the Hamiltonian simulation algorithms based on the composition of Trotter and LCU algorithms. 
In both theoretical and numerical studies, we show that the $0$th-order paired Taylor-series compensation (PTSC) algorithm, $2k$th-order PTSC algorithm and the $2k$th-order nested-commutator compensation (NCC) algorithm enjoy different advantages and will be useful in different scenarios. 
Taking the $n$-qubit lattice Hamiltonian as an example: the $0$th-order PTSC algorithm performs the best when $t$ is small compared with $n$ and $1/\varepsilon$; the $2k$th-order PTSC algorithm performs the best when $n$ is small compared with $t$ and $1/\varepsilon$; while the $2k$th-order NC algorithm performs the best when $1/\varepsilon$ is small compared with $n$ and $t$. 
In practice, with finite system size $n$, simulation time $t$ and inverse accuracy $1/\varepsilon$, we can think about a hybrid implementation of different algorithms. For example, when the sparsity $L$ of a given Hamiltonian is large, we can first split the Hamiltonian to two parts,
\begin{equation}
H = H_1 + H_2 = \sum_{j=1}^{L_1} H_{j} + \sum_{k=L_1+1}^{L} H_{k},
\end{equation}
where the summands in $H_1$ are the few dominant terms with large coefficients. We can then perform second-order Trotter only for $H_1$, and use PTS to expand the remainder
\begin{equation}
V_2(x;H_1) = \prod_{j_1=L_1}^1 e^{i H_{j_1} x} e^{-iHx} \prod_{j_2=1}^{L_1} e^{i H_{j_2} x}.
\end{equation}
If the number of dominant terms $L_1$ is small, we can then reduce the $L$ dependence of the algorithm similar to $0$th-order PTSC algorithm while keep the good $t$-dependence of second-order PTSC algorithm.
As another example, we can hybridize second-order PTSC and NCC algorithms: we apply the nested-commutator compensation for the leading-order terms (i.e., the terms with $s=3,4$, and $5$), and normal Taylor-series compensation for higher-order terms. In this case, we can find an optimal truncation location $s_c$ which fulfills the high simulation accuracy requirement and keeps the nested-commutator scaling for the leading-order compensation terms.

The design of Trotter-LCU algorithms is based on a series connection of Trotter and LCU algorithms.
A similar composition method can also be exploited for other Hamiltonian simulation algorithms~\cite{hagan2022composite}.
For example, we may replace the deterministic Trotter with the ones with random permutation~\cite{childs2019fasterquantum}. 
Recently, Cho et al.~\cite{cho2022doubling} consider similar idea to compensate the Trotter error using randomized unitary operators.
Using anticommutative cancellation~\cite{Zhao2021exploiting}, we can further reduce the compensation terms.

\begin{acknowledgments}
We thank Xiaoming Zhang, Xiao Yuan, Min-Hsiu Hsieh, Yuan Su, Kaiwen Gui, Ming Yuan, Senrui Chen, and Ying Li for helpful discussion and suggestions. 
We would like to especially thank Wenjun Yu for highlighting the validity of the variant where the ancillary qubit is measured and reset for each segment.
P.~Z.~ and L.~J.~acknowledge support from the ARO MURI (W911NF-21-1-0325), AFOSR MURI (FA9550-19-1-0399, FA9550-21-1-0209), AFRL (FA8649-21-P-0781), NSF (OMA-1936118, ERC-1941583, OMA-2137642), NTT Research, and the Packard Foundation (2020-71479).
J.S. would like to thank support from the Innovate UK (Project No.10075020) and support through Schmidt Sciences, LLC.
Q.~Z.~acknowledges HKU Seed Fund for Basic
Research for New Staff via Project No.~2201100596, Guangdong Natural Science Fund via Project No.~2023A1515012185, National Natural Science Foundation of China (NSFC) via Projects No.~12305030 and No.~12347104, Hong Kong Research Grant Council (RGC) via No.~27300823, No.~N\_HKU718/23, and No.~R6010-23, Guangdong Provincial Quantum Science Strategic Initiative GDZX2200001.
\end{acknowledgments}

\begin{appendix}

\section{PAIRED TAYLOR-SERIES COMPENSATION WITH HIGHER-ORDER TROTTER FORMULAS} \label{sec:pairedTSgeq2}

Following the same idea in \autoref{ssc:pairedTS1}, we now generalize it to the case with $2k$th-order Trotter formula. Expanding the $2k$th-order Trotter remainder, we have
\begin{equation}
V_{2k}(x) = \sum_{s=0}^\infty F_{2k,s}(x) = \sum_{s=0}^\infty \eta_s V_{2k,s},
\end{equation}
where
\begin{equation} \label{eq:F2ksSimp}
\begin{aligned}
F_{2k,s}(x) &= \eta_s \sum_{r,\gamma} \Pr(r,\gamma|s) P_{2k}(r,\gamma), \\
\Pr(r,\gamma|s) &:= \Pr(r,\vec{r}_{1:\kappa},\vec{r}'_{1:\kappa}|s) \sum_{l_{1:r}} p_{l_{1:r}}^{(r)}, \\
P_{2k}(r,\gamma) &:= (-i)^{2r-s} \prod^\leftarrow\vec{P}^{\vec{r}'_\kappa}... \prod^\leftarrow\vec{P}^{\vec{r}'_1} \left( P^{(r)}_{l_{1:r}} \right) \cdot \\
&\quad \prod^\rightarrow\vec{P}^{\vec{r}_1}...\prod^\rightarrow\vec{P}^{\vec{r}_\kappa}.
\end{aligned}
\end{equation}
Here, we use $\gamma$ to denote all the expansion variables $\{\vec{r}_{1:\kappa},\vec{r}'_{1:\kappa},l_{1:r}\}$ besides $r$.

We ignore the derivation and provide the general form of the LCU formula for $\tilde{V}_{2k}^{(p)}(x)$,
\begin{equation} \label{eq:tildeV2kxPTS}
\begin{aligned}
\tilde{V}_{2k}^{(p)}(x) &= \mu_{2k}^{(p)}(x) \Big( \sum_{s=2k+1}^{4k+1} \mr{Pr}_{2k}^{(p)}(s) R_{2k,s}^{(p)}(\eta_\Sigma) \\
&\quad + \sum_{s=4k+2}^{s_c} \mr{Pr}_{2k}^{(p)}(s) V_{2k,s}^{(p)} \Big).
\end{aligned}
\end{equation}
where
\begin{align*}
\mu_{2k}^{(p)}(x) &= \sqrt{1+(\eta_\Sigma)^2} + \sum_{s=4k+2}^{s_c} \eta_s, \quad 
\eta_\Sigma = \sum_{s=2k+1}^{4k+1} \eta_s, \stepcounter{equation}\tag{\theequation}\label{eq:V2kxPTSPara}\\ 
\mr{Pr}^{(p)}(s) &= \frac{1}{\mu_{2k}^{(p)}(x)}
\begin{cases}
& \sqrt{1+(\eta_\Sigma)^2} \frac{\eta_s}{\eta_\Sigma}, \\ &\quad  s=2k+1,2k+2,...,4k+1, \\
& \eta_s, \\ &\quad s= 4k+2,4k+3,...,s_c. \\
\end{cases}
\end{align*}
The Pauli rotation unitary $R_{2k,2q+1}^{(p)}(y):=\exp(i\theta(y) P_{2k}'(r,\gamma)')$ where $\theta(y):=\tan^{-1}(1+y^2)$ and $P_{2k}'(r,\gamma) := (-i)^{\mathbbm{1}[P_{2k,s}(r,\gamma):\text{anti-Her}]} P_{2k}(r,\gamma)$. 

Combined with the deterministic Trotter formula, the overall LCU formula for $U(x)$ is
\begin{equation} \label{eq:tildeU2kxPTS}
\begin{aligned}
\tilde{U}_{2k}^{(p)}(x)  = \tilde{V}_{2k}^{(p)}(x) S_{2k}(x).
\end{aligned}
\end{equation}

The following proposition gives the performance characterization of $\tilde{U}_{2k}^{(p)}(x)$ to approximate $U(x)$.

\begin{proposition}[$2k$th-order Trotter-LCU formula by paired Taylor-series compensation] \label{prop:PTS2k}

For $0<x<1/(2\lambda)$ and $s_c\geq 4k+1$, $\tilde{V}_{2k}^{(p)}(x)$ in \autoref{eq:tildeV2kxPTS} is a $(\mu_{2k}^{(p)}(x),\varepsilon_{2k}^{(p)}(x))$-LCU formula of $V_{2k}(x)$ with
\begin{equation}
\begin{aligned}
\mu_{2k}^{(p)}(x) &\leq e^{(e+c_k)(2\lambda x)^{4k+2}}, \\
\varepsilon_{2k}^{(p)}(x) &\leq \left(\frac{2e\lambda x}{s_c +1}\right)^{s_c+1}.
\end{aligned}
\end{equation}
Here, 
\begin{equation} \label{eq:c_k}
c_k := \frac{1}{2}\left( \frac{e}{2k+1} \right)^{4k+2},
\end{equation}
so that $0.3>c_k>0$. As a result, $\tilde{U}_{2k}^{(p)}(x)$ in \autoref{eq:tildeU2kxPTS} is a $(\mu_{2k}^{(p)}(x),\varepsilon_{2k}^{(p)}(x))$-LCU formula of $U(x)$.
\end{proposition}

\begin{proof}

We will focus on the case with $k=1$. The proof for a general $k$ is simiar.
We first bound the normalization factor $\mu_{2}^{(p)}(x)$. When $2\lambda x<1$ we have
\begin{align*}
&\mu_{2}^{(p)}(x) \leq \sqrt{1+(\eta_\Sigma)^2} + \sum_{s=6}^{\infty} \eta_s \\
&\leq 1 + \frac{1}{2}\eta_\Sigma^2 + \left( e^{2\lambda x} - \sum_{s=0}^{5} \eta_s \right) \\
&= \frac{1}{2} (2\lambda x)^{6} \left( \frac{1}{3!} + \frac{2\lambda x}{4!} +\frac{(2\lambda x)^{2}}{5!} \right)^2 + \left( e^{2\lambda x} - \sum_{s=1}^{5} \eta_s \right) \stepcounter{equation}\tag{\theequation}\label{eq:mu2kLOPbound}\\ 
&\leq \frac{1}{2} (2\lambda x)^{6} \left( \sum_{s=3}^\infty \frac{1}{3!} \right)^2 + \left( e^{2\lambda x} - \sum_{s=1}^{5} \eta_s \right)  \\
&\leq \frac{1}{2}\left( \frac{e}{3} \right)^6 (2\lambda x)^{6} + e^{e(2\lambda x)^{6}} \leq e^{(e+c_1)(2\lambda x)^{6}}.
\end{align*}
In the fifth line, we use Corollary~\ref{coro:ExpTailPower}.

The distance bound can be derived in the same manner as the one in Proposition~\ref{prop:PTS1}.
\end{proof} 

\section{TIGHT NESTED-COMMUTATOR ANALYSIS FOR SECOND-ORDER TROTTER-LCU ALGORITHM} \label{sec:AppNested2nd}

We can easily extend the methods for the first-order analysis in \autoref{ssc:NestedCommLattice1st} to the higher-order case. Taking the second-order case as an example, the second-order Trotter formula of the lattice Hamiltonian is
$S_2(x) = e^{-i\frac{x}{2}A} e^{-ixB} e^{-i\frac{x}{2}A}$.

To derive the explicit LCU formula $\tilde{V}_2^{(nc)}(x)$ in \autoref{eq:tildeVKxNCpair}, we first derive $J_2(x)$ defined in \autoref{eq:JKxDef} and its derivatives $C_s:=J_2^{(s)}(x)$. From \autoref{eq:RKxDef} and \autoref{eq:JKxDef} we have
\begin{equation} \label{eq:R2xJ2x}
\begin{aligned}
S_2(x)^\dag &= e^{ixH} + \int_0^x d\tau e^{i(x-\tau)H} R_2(\tau), \\
R_2(\tau) &= i[e^{i\frac{x}{2}A}, B]e^{ixB}e^{i\frac{x}{2}A} + \frac{i}{2}[e^{i\frac{x}{2}A}e^{ixB}, A]e^{i\frac{x}{2}A}, \\
J_2(x) &= S_2(x)R_2(x) = \frac{i}{2}\left( A - e^{-i\frac{x}{2}\ad_A}e^{-ix\ad_B}A \right) \\ 
& + i\left( e^{-i\frac{x}{2}\ad_A}B - e^{-i\frac{x}{2}\ad_A}e^{-ix\ad_B} e^{-i\frac{x}{2}\ad_A}B \right).
\end{aligned}
\end{equation}

Following the approach in \autoref{ssc:NCDerive}, we expand $V_2(x)$ and $J_2(x)$ by
\begin{equation}
\begin{aligned}
V_2(x) = \sum_{s=0}^\infty G_s \frac{x^s}{s!},
\quad J_2(x) = \sum_{s=0}^\infty C_s \frac{x^s}{s!},
\end{aligned}
\end{equation}
then based on \autoref{prop:order} and the recurrence formula \autoref{eq:VKxRecurrence} we have,
\begin{equation} \label{eq:2ndNCorder}
\begin{aligned}
G_1 = G_2 = 0, \quad C_0 = C_1 = 0, \\
G_s = C_{s-1}, \quad s=3,4,5.
\end{aligned}
\end{equation}
Combining \autoref{eq:R2xJ2x} and \autoref{eq:2ndNCorder}, we can show that
\begin{equation}
G_s = \mc{O}(n^{\lfloor \frac{s}{3} \rfloor}).
\end{equation}
Therefore, the first three nontrivial terms $G_3,G_4,G_5=\mc{O}(n)$. These terms will be set as the leading-order terms. 

Based on \autoref{eq:2ndNCorder}, we are going to expand $J_K(x)$ based on the operator-valued Taylor-series expansion with integral remainders,
\begin{equation}
J_2(x) = J_{2,L}(x) + J_{2,res,4}(x) = \sum_{s=2}^{4} C_s \frac{x^s}{s!} + J_{2,res,4}(x),
\end{equation}
where
\begin{equation} \label{eq:CsJresidue2nd}
\begin{aligned}
C_s &= J_2^{(s)}(0), \quad s=2,3,4, \\
J_{2,res,s}(x) &= \int_0^x d\tau \frac{(x-\tau)^{s}}{s!} J_2^{(s+1)}(\tau).
\end{aligned}
\end{equation}
$J_{2,L}(x):= \sum_{s=2}^{4} C_s \frac{x^s}{s!}$ denotes the leading-order terms in $J_K(x)$. Then, from \autoref{eq:VKxRecurrence} and \autoref{eq:2ndNCorder}, the second-order Trotter remainder can be expressed as
\begin{equation} \label{eq:V2xExpandNC}
\begin{aligned}
V_2(x) &= I + \sum_{s=3}^5 F_{2,s}^{(nc)}(x) + F_{2,\mc{res}}^{(nc)}(x), \\
F_{2,s}^{(nc)}(x) &= C_{s-1} \frac{x^s}{s!}, \quad 3,4,5, \\
F_{2,\mc{res}}^{(nc)}(x) &= \int_0^x d\tau \left( M_2(\tau) J_{2,L}(\tau) + V_2(\tau) J_{2,res,4}(\tau) \right).
\end{aligned}
\end{equation}
We put the explicit nested-commutator expressions of the leading-order terms $F_{2,s}^{(nc)}(x)$ ($s=3,4$, or $5$) in Sec.~\ref{sec:AppNested2nd}.

In practice, we truncate the formula with the order $s_c=5$. Based on \autoref{eq:tildeVKxNCpair}, the truncated nested-commutator LCU formula for $V_2(x)$ can be written as
\begin{equation} \label{eq:V2xNCexpandpair}
\begin{aligned}
\tilde{V}_2(x) &= I + \sum_{s=3}^5 F_{2,s}^{(nc)}(x), \\
&=\sqrt{1+ (\eta_\Sigma^{(nc)})^2} \sum_{s=3}^5 \frac{\eta_s^{(nc)}}{\eta_\Sigma^{(nc)}} R_{2,s}^{(nc)}(\eta_\Sigma).
\end{aligned}
\end{equation}
Here, $\eta_\Sigma^{(nc)}:= \sum_{s=3}^{5}\|C_{s-1}\|_1 \frac{x^s}{s!}$. The explicit form of $R_{2,s}^{(nc)}(\eta_\Sigma)$ can be obtained by the definitions in \autoref{eq:FKsNCexpand}, \autoref{VKxNCpaired} and the Pauli operator decomposition based on the nested-commutator form in \autoref{eq:NumC2C3C4}. 

For a tight numerical estimation, we seek for a tighter bound of the $1$-norm $\mu_2^{(nc)}(x)$ and spectral norm accuracy $\varepsilon_2^{(nc)}(x)$ for the LCU formula in \autoref{eq:V2xNCexpandpair}.

\subsection{Bound the 1-norm of recurrence function}

Hereafter, we define $A':=-\frac{i}{2}A$ and $B' := -i B$ to simplify the notation.
We have
\begin{equation} \label{eq:JxJaxJbx}
\begin{aligned}
J_2(x) &= J_a(x) + J_b(x), \\
J_a(x) &= \left(e^{x\ad_{A'}}e^{x\ad_{B'}}A' -A'\right), \\
J_b(x) &= \left(e^{x\ad_{A'}}e^{x\ad_{B'}}e^{x\ad_{A'}}B' -e^{x\ad_{A'}}B'\right).
\end{aligned}
\end{equation}

Apply the Libniz rule to $J_a(x)$ and $J_b(x)$, we have
\begin{equation} \label{eq:JaJbDeri}
\begin{aligned}
J_a^{(s)}(x) &= \sum_{l=0}^s \binom{s}{l} e^{x\ad_{A'}} \ad_{A'}^l e^{x\ad_{B'}} \ad_{B'}^{s-l} A', \\
J_b^{(s)}(x) &= \sum_{\substack{l_1, l_2, l_3 \\ l_1+l_2+l_3=s,l_2+l_3\neq 0}} \binom{s}{l_1,l_2,l_3}\cdot \\
&\quad\quad e^{x\ad_{A'}} \ad_{A'}^{l_1} e^{x\ad_{B'}} \ad_{B'}^{l_2}  e^{x\ad_{A'}} \ad_{A'}^{l_3} B'.
\end{aligned}
\end{equation}
Then,
\begin{equation}
\begin{aligned}
C_s &= \sum_{l=0}^{s-1} \ad_{A'}^{l}\ad_{B'}^{s-l}A \\
&\quad + \sum_{\substack{l_1, l_2, l_3 \\ l_1+l_2+l_3=s,l_2+l_3\neq 0}} \binom{s}{l_1,l_2,l_3} \ad_{A'}^{l_1}\ad_{B'}^{l_2} \ad_{A'}^{l_3} B.
\end{aligned}
\end{equation}
We can solve the explicit form of the leading-order terms with $s=2,3$, and $4$,
\begin{equation} \label{eq:NumC2C3C4}
\begin{aligned}
C_2 &= \left(\ad_{B'}^2 A' + 2\ad_{A'}\ad_{B'} A' \right)+ \left(3\ad_{A'}^2 B' + 2\ad_{B'}\ad_{A'} B'\right), \\
C_3 &= \left(\ad_{B'}^3 A' + 3\ad_{A'}\ad_{B'}^2 A' + 3\ad_{A'}^2\ad_{B'} A' \right) + ( 7\ad_{A'}^3 B'  \\
&\quad  + 3\ad_{B'}\ad_{A'}^2 B' + 6\ad_{A'}\ad_{B'}\ad_{A'} B' + 3\ad_{B'}^2\ad_{A'} B' ), \\
C_4 &= \left(\ad_{B'}^4 A' + 4\ad_{A'}\ad_{B'}^3 A' + 6\ad_{A'}^2\ad_{B'}^2 A' + 4\ad_{A'}^3\ad_{B'} A'  \right) \\ 
&\quad + ( 15\ad_{A'}^4 B' + 4\ad_{B'}\ad_{A'}^3 B' + 12 \ad_{A'}^2\ad_{B'}\ad_{A'} B' \\ 
&\quad + 12 \ad_{A'}\ad_{B'}\ad_{A'}^2 B'+ 6\ad_{B'}^2\ad_{A'}^2 B'  \\
&\quad + 12\ad_{A'}\ad_{B'}^2\ad_{A'} B' + 4\ad_{B'}^3\ad_{A'} B' ).
\end{aligned}
\end{equation}
We can then expand the operators $A'$ and $B'$ to Pauli operators and solve the $1$-norm of $C_2, C_3$, and $C_4$ under the Pauli decomposition based on \autoref{eq:NumC2C3C4}. The $1$-norm of $\tilde{V}_2^{(nc)}(x)$ is given by
\begin{equation} \label{eq:tildeMuNC}
\tilde{\mu}_2^{(nc)}(x) = \sqrt{1 + \left(\sum_{s=2}^4 \|C_s\|_1 \frac{x^{s+1}}{(s+1)!} \right)^2}.
\end{equation}

We summarize the procedure to estimate the segment number $\nu$ of second-order Trotter-LCU under nested-commutator compensation as follows.
\begin{enumerate}
\item For a specific lattice Hamiltonian, get the explicit Pauli expansion form of $C_2$, $C_3$, $C_4$ by \autoref{eq:NumC2C3C4}. 
\item Get the expression of the $1$-norm $\tilde{\mu}^{(nc)}(x)$ by \autoref{eq:tildeMuNC}. Calculate $x_{U}$ based on $\tilde{\mu}^{(nc)}(x)=2$.
\item Solve the following residue operator,
\begin{equation}
V_{2,\mc{res}}(x) = V_2(x) - \tilde{V}_2(x) = U(x)S_2(x)^\dag,
\end{equation}
calculate its spectral norm $\|V_{2,\mc{res}}(x)\|$ numerically. Check either the following requirement are satisfied:
\begin{equation} \label{eq:VresxRequire}
\|V_{2,\mc{res}}(x)\| \frac{t}{x} \leq \varepsilon,
\end{equation}
where $\varepsilon$ is a preset accuracy requirement. If not, we search the largest $x$ value in the region $0<x<x_U$ by dichotomy which makes \autoref{eq:VresxRequire} satisfied.
\item The number of segments is given by $\nu:=t/x$. We then calculate the number of gates based on methods in Appendix~\ref{ssc:NC2ndGateNum}.
\end{enumerate}

We have a tighter count bound for the norm of $J_a^{(s)}(x)$ and $J_b^{(s)}(x)$,
\begin{proposition} \label{prop:adnorm2nd}
Consider a lattice Hamiltonian $H=A+B$ with the form in \autoref{eq:HLattice}. Suppose the spectral norm and $1$-norm of its components $H_{j,j+1}$ are bounded by $\Lambda$ and $\Lambda_1$. Then for the recurrence function $J_2(x)$, we have the following norm bound for its derivatives,
\begin{equation}
\begin{aligned}
\|J^{(s)}_2(x)\| &\leq \frac{n}{2} (7^s + 12^s) \Lambda^{s+1}, \\
\|J^{(s)}_2(x)\|_1 &\leq \frac{n}{2} (7^s + 12^s) \Lambda_1^{s+1}.
\end{aligned}
\end{equation}
\end{proposition}

\begin{proof}
From \autoref{eq:JaJbDeri} we have,
\begin{equation} \label{eq:JasJbsNorm}
\begin{aligned}
\|J_a^{(s)}(x)\| &\leq \sum_{l=0}^s \binom{s}{l} \|e^{x\ad_{A'}} \ad_{A'}^l e^{x\ad_{B'}} \ad_{B'}^{s-l} A'\|, \\
\|J_b^{(s)}(x)\| &\leq \sum_{\substack{l_1, l_2, l_3 \\ l_1+l_2+l_3=s,l_2+l_3\neq 0}} \binom{s}{l_1,l_2,l_3} \cdot \\
&\quad\|e^{x\ad_{A'}} \ad_{A'}^{l_1} e^{x\ad_{B'}} \ad_{B'}^{l_2}  e^{x\ad_{A'}} \ad_{A'}^{l_3} B'\|.
\end{aligned}
\end{equation}

To bound the norm of nested commutators, we use the same methods in \autoref{prop:adnorm1st}. We have
\begin{equation} \label{eq:adnorm2ndtight}
\begin{aligned}
& \|e^{x\ad_{A'}} \ad_{A'}^l e^{x\ad_{B'}} \ad_{B'}^{s-l} A\| \\
&\quad\quad \leq (3^l 2^{s-l})(\Lambda^l (2\Lambda)^{s-l} ) \cdot (\frac{n}{2}\Lambda), \\
& \|e^{x\ad_{A'}} \ad_{A'}^{l_1} e^{x\ad_{B'}} \ad_{B'}^{l_2}  e^{x\ad_{A'}} \ad_{A'}^{l_3} B\| \\
&\quad\quad \leq (4^{l_1} 3^{l_2} 2^{l_3})(\Lambda^{l_1} (2\Lambda)^{l_2} \Lambda^{l_3} ) \cdot (\frac{n}{2}\Lambda).
\end{aligned}
\end{equation}
Here, the first bracket of each bound corresponds to possible nest commutators, while the second bracket of each bound indicates the norm enlargement of each nested commutator.

Combining \autoref{eq:JasJbsNorm} and \autoref{eq:adnorm2ndtight}, we have
\begin{equation} \label{eq:JasJbsNorm2}
\begin{aligned}
\|J_a^{(s)}(x)\| &\leq \sum_{l=0}^s \binom{s}{l} (3\Lambda)^l (2\cdot 2\Lambda)^{s-l}) \cdot (\frac{n}{2}\Lambda) \\
&\leq (7\Lambda)^s \left( \frac{n}{2}\Lambda \right), \\
\|J_b^{(s)}(x)\| &\leq \sum_{\substack{l_1, l_2, l_3 \\ l_1+l_2+l_3=s,l_2+l_3\neq 0}} \binom{s}{l_1,l_2,l_3}\cdot \\
&\quad(4\Lambda)^{l_1} (3\cdot 2\Lambda)^{l_2} (2\Lambda)^{l_3}) \cdot (\frac{n}{2}\Lambda) \\
&\leq (12\Lambda)^s \left( \frac{n}{2}\Lambda \right).
\end{aligned}
\end{equation}

The $1$-norm bound can be derived in the same way.
\end{proof}

We can derive the bounds for $\tilde{\mu}_2^{(nc)}(x)$ and $\varepsilon_2^{(nc)}(x)$ as follows:
\begin{align}
&\tilde{\mu}_2^{(nc)}(x) = \sqrt{ 1 + \left(\sum_{s=2}^4 \|C_s\|_1 \frac{x^{s+1}}{(s+1)!}\right)^2} \label{eq:tildemuNC2ndBound1} \\
&\leq \sqrt{ 1 + \frac{n^2}{4} \left( \sum_{s=2}^4 (7^s + 12^s) \frac{(\Lambda_1 x)^{s+1}}{(s+1)!}\right)^2 } \label{eq:tildemuNC2ndBound2} \\
&\leq \exp\left( (35n)^2 (\Lambda_1 x)^6 \right) \label{eq:tildemuNC2ndBound3}, \\
&\varepsilon_2^{(nc)}(x) \leq \frac{n^2}{4}(7^2+12^2)\sum_{s=2}^4 (7^s+12^s)\frac{(\Lambda x)^{s+4}}{3!s!(s+4)} \nonumber \\
&\quad\quad\quad + \frac{n}{2}(7^5+12^5)\frac{(\Lambda x)^6}{6!} \label{eq:epsilonNC2ndBound1}\\
&\leq (3n^2 + 185n)(\Lambda x)^6 \leq 370 n^2 (\Lambda x)^6, \label{eq:epsilonNC2ndBound2} 
\end{align}
Here, \autoref{eq:tildemuNC2ndBound3} holds when $\Lambda_1 x\leq 1/3$, \autoref{eq:epsilonNC2ndBound2} holds when $\Lambda x\leq 21/72$.

\subsection{Estimating the gate counts} \label{ssc:NC2ndGateNum}

To estimate the performance, i.e., gate complexity of the Trotter-LCU algorithm based on nested-commutator compensation, we need to first estimate the number of segments $\nu=t/x$ in the algorithm. This is determined by the following three constraints,
\begin{equation} \label{eq:nuConstr}
\begin{aligned}
\Lambda_1 x, \Lambda x\leq 1 \quad \Leftarrow\quad &\nu \geq \max\{\Lambda_1 t, \Lambda t\}, \\
(\tilde{\mu}^{(nc)})^\nu \leq 2 \quad \Leftrightarrow\quad &\nu\ln\tilde{\mu}^{(nc)} \leq \ln2, \\
\nu\tilde{\mu}^\nu (\varepsilon^{(nc)}) \leq \varepsilon \quad \Leftarrow \quad &\nu\varepsilon^{(nc)} \leq \frac{1}{2}\varepsilon.
\end{aligned}
\end{equation}

After we solve the required segment number $\nu$ by \autoref{eq:nuConstr}, we can estimate the gate number accordingly. Here, we introduce the method to estimate the gate number of the LCU part. The gate number in the implementation of second-order Trotter formula can be evaluated following the methods in Ref.~\cite{childs2018toward}. 
In the worst-case scenario, the Pauli weight of the gate is determined by the weight of Pauli operators contained in $C_4$ in \autoref{eq:NumC2C3C4}. The largest Pauli weight is $6$. As a result, a controlled-Pauli gate will cost at most six $\mc{CNOT}$ gates and no non-Clifford gate.

If we consider the paired algorithm, when we sample the third-order term, it will be a Pauli rotation unitary on four qubits. In this case, it will cost eight $\mc{CNOT}$ gates and two single-qubit Pauli rotation gates $R_z(\theta)$.

We summarize the whole procedure to estimate the gate complexity.
\begin{enumerate}
\item Input: Hamiltonian parameters: $H=A+B$, $n$, $\Lambda$, $\Lambda_1$. Normalization requirements $\mu=2$, accuracy requirements $\varepsilon$, time requirements $t$.

\item Normalization factor estimation.
\begin{enumerate}
\item Analytical way (scalable). Using \autoref{eq:tildemuNC2ndBound2} to get the function $\tilde{\mu}^{(nc)}(x)$.
\item Numerical way (scalable). Using \autoref{eq:NumC2C3C4} to get the form of $C_2$, $C_3$, and $C_4$. Then get the function $\tilde{\mu}^{(nc)}(x)$ by \autoref{eq:tildemuNC2ndBound1}.
\end{enumerate}

\item Accuracy estimation.
\begin{enumerate}
\item Analytical way (scalable). Using \autoref{eq:epsilonNC2ndBound1} to get the function $\varepsilon^{(nc)}(x)$.
\item Numerical way (unscalable). Calculate $V_{2,\mc{res}}(x):= U(x)S_2(x)^\dag - I - \sum_{s=2}^4 C_s \frac{x^{s+1}}{(s+1)!}$ numerically. Solve its largest singular value, which is an upper bound of $\varepsilon^{(nc)}(x)$.
\end{enumerate}

\item Based on the constraints in \autoref{eq:nuConstr}, calculate the segment number $\nu$.

\item Analyze the $\mc{CNOT}$ and $R_z(\theta)$ gate number.
\end{enumerate}

\section{EXPLICIT NESTED-COMMUTATOR COMPENSATION FOR HIGHER-ORDER TROTTER REMAINDERS} \label{sec:AppNChigherTrotter}

In this section, we provide detailed results for the nested-commutator compensation for higher-order Trotter remainders introduced in Sec.~\ref{sec:NCC}. Here, we will focus on the lattice model Hamiltonians in \autoref{eq:HLattice}. The results for general Hamiltonians will be presented in Appendix~\ref{sec:AppNCgeneral}.

As introduced in Sec.~\ref{ssc:NCDerive}, for the (asymmetric) multiplicative remainder $V_{K}(x)$, our aim is to construct a LCU formula with the following form
\begin{equation} \label{eq:AppV2kxNC}
V_{K}(x) = I + \sum_{s=K+1}^{2K+1} F_{K,s}^{(nc)}(x) + F_{K,\mc{res}}^{(nc)}(x),
\end{equation}
where $F_{K,\mc{res}}^{(nc)}(x)$ denotes the term with $x$-order $\mc{O}(x^{2K+2})$. In practice, we will use the truncated LCU formula with paired leading-order terms, derived in \autoref{eq:tildeVKxNCpair},
\begin{equation} \label{eq:ApptildeVKxNCpair}
\begin{aligned}
&\tilde{V}_K^{(nc)}(x) = I + \sum_{s=K+1}^{2K+1} F_{K,s}^{(nc)}(x) \\
&= I + i \sum_{s=K}^{2K} \eta^{(nc)}_s V_{K,s}^{(nc)} \\
&= \sum_{s=K}^{2K} \frac{\eta^{(nc)}_s}{\eta_\Sigma^{(nc)}} \left( I + \eta^{(nc)}_\Sigma V_{K,s}^{(nc)} \right)  \\
&= \sqrt{1+(\eta_\Sigma^{(nc)})^2} \sum_{s=K}^{2K} \frac{\eta^{(nc)}_s}{\eta_\Sigma^{(nc)}} R_{K,s}^{(nc)}(\eta_\Sigma),
\end{aligned}
\end{equation}
where $\eta^{(nc)}_\Sigma = \sum_{s=K+1}^{2K+1} \eta^{(nc)}_s = \sum_{s=K}^{2K} \|C_{s-1}\|_1 \frac{x^{s+1}}{(s+1)!}$.

We are going to finish the following tasks,
\begin{enumerate}
\item (Sec.~\ref{ssc:AppNCDerive}) Derive the explicit formulas for the leading-order expansion terms $F_{K,s}^{(nc)}(x)$ with $s=K+1,...,2K+1$. 
\item (Sec.~\ref{ssc:AppNormBound}) Prove the $1$-norm $\mu_K^{(nc)}(x)$ and error bound $\varepsilon_K^{(nc)}(x)$ in Proposition~\ref{prop:NCK}.
\end{enumerate}

\subsection{Derivation of LCU formula with nested-commutator form} \label{ssc:AppNCDerive}

We first introduce the canonical expression of the $K$th-order Trotter formula,
\begin{equation} \label{eq:AppSKtWt}
S_K(x) := W(x) = \prod_{j=1}^{\kappa} W_j(x) = W_\kappa(x)... W_2(x)W_1(x),
\end{equation}
where
\begin{equation} \label{eq:AppWjt}
W_j(x) := e^{-ix b_j B} e^{-ix a_j A},
\end{equation}
is the $j$th stage of the Trotter formula. $\kappa$ is the stage number. We have $\kappa=1$ for $K=1$ and $\kappa\leq 2\times 5^{k-1}$ when $K=2k$. The stage lengths $a_j$ and $b_j$ are determined based on \autoref{eq:S2x} and \autoref{eq:S2kx}. We have
\begin{equation}
\begin{aligned}
0\leq & a_j, b_j \leq 1, \quad \forall j=1,2,...,\kappa, \\
\sum_{j=1}^\kappa & a_j = \sum_{j=1}^\kappa b_j =1.
\end{aligned}
\end{equation}
For example, for the second-order Trotter formula $S_2(x)=e^{-ix \frac{A}{2}} e^{-ix B} e^{-ix\frac{A}{2}}$, we set the stage number $\kappa=2$ with $a_1=\frac{1}{2}, b_1=1; a_2 = \frac{1}{2}, b_2=0$.

The Hermitian conjugate of $W(x)$ is
\begin{equation}
W(x)^\dag = \prod_{j=\kappa}^{1} W_j(x)^\dag = W_1(x)^\dag W_2(x)^\dag ... W_\kappa(x)^\dag.
\end{equation}

As is discussed in Sec.~\ref{ssc:NCDerive}, to derive the nested-commutator form of $V_K(x)$, we first solve $J_K(x)$ defined in \autoref{eq:JKxDef}. We have
\begin{equation} \label{eq:Rt}
\begin{aligned}
R_K(x) &= \frac{d}{dx} W(x)^\dag - iH W(x)^\dag, \\
J_K(x) &:= W(x) R_K(x). 
\end{aligned}
\end{equation}

From Proposition~\ref{prop:order}, we can write the $K$th-Trotter remainder $V_K(x)$ and $J_K(x)$ as the following form:
\begin{equation} \label{eq:VtJt}
\begin{aligned}
V_K(x) &:= I + M_K(x) = I + \sum_{s=K+1}^{2K+1} C_{s-1} \frac{x^s}{s!} + F_{K,\mc{res}}(x), \\
J_K(x) &= \sum_{s=K}^{2K} C_s \frac{x^s}{s!} + J_{K,res,2K}(x),
\end{aligned}
\end{equation}
where $F_{K,\mc{res}}(x)=\mc{O}(x^{2K+2})$ and $J_{K,res,2K}(x)=\mc{O}(x^{2K+1})$ are the higher-order remaining terms to be analyzed later. We also denote
\begin{equation} \label{eq:FLJL}
\begin{aligned}
F_L(x) &:= \sum_{s=K+1}^{2K+1} C_{s-1} \frac{x^s}{s!}, \\
J_L(x) &:= \sum_{s=K}^{2K} C_s \frac{x^s}{s!},
\end{aligned}
\end{equation}
as the leading-orders whose explicit forms will be calculated in this section. We will show that $F_L(x)$ contains $\mc{O}(n)$ term where $n$ is the lattice size. 

Now, we are going to solve the exact form of $C_s$ for the leading-orders. We first try to solve the succinct form of $R_K(x)$ based on its definition in \autoref{eq:Rt}. Taking the derivative for each Trotter stage, we have
\begin{equation} \label{eq:Rtderive1}
\begin{aligned}
& R_K(x) = \frac{d}{dx}\left[ \prod_{j=\kappa}^1 W_j(x)^\dag \right] - iH \prod_{j=\kappa}^1 W_j(x)^\dag \\
&= \sum_{j=\kappa}^{1} \prod_{l=j-1}^1 W_l(x)^\dag \frac{d}{dx} W_j(x)^\dag \prod_{l=\kappa}^{j+1} W_l(x)^\dag \\
&\quad\quad - \sum_{j=\kappa}^1 (ia_j A + ib_j B) \prod_{l=\kappa}^{1} W_l(x)^\dag.
\end{aligned}
\end{equation}
Here we assume $1\leq j-1<j+1\leq \kappa$. When $j=1$ (or $\kappa$), the value of the product $\prod_{l=j-1}^1 W_l(x)^\dag$ (or $\prod_{l=\kappa}^{j+1} W_l(x)^\dag$) will be regarded as $I$. Recall that $W_j(x)^\dag = e^{ixa_j A} e^{ix b_j B}$. We further expand $\frac{d}{dx} W_j(x)^\dag$ and merge the two terms together,
\begin{equation} \label{eq:Rtderive2}
\begin{aligned}
& R_K(x) = \sum_{j=\kappa}^{1} \prod_{l=j-1}^1 W_l(x)^\dag \left( ia_j A W_j(x)^\dag + W_j(x)^\dag ib_j B \right) \cdot \\
&\quad\quad \prod_{l=\kappa}^{j+1} W_l(x)^\dag - \sum_{j=\kappa}^1 (ia_j A + ib_j B) \prod_{l=\kappa}^{1} W_l(x)^\dag \\
&= \sum_{j=\kappa}^{1} \left[ \prod_{l=j}^1 W_l(x)^\dag, i a_{j+1} A + i b_j B \right] \prod_{l=\kappa}^{j+1} W_l(x)^\dag.
\end{aligned}
\end{equation}
Here, we assume $a_{\kappa+1}=0$. Now, we simplify the commutator by splitting the product and then change the summation order,
\begin{equation} \label{eq:Rtderive3}
\begin{aligned}
& R_K(x) = \sum_{j=\kappa}^{1} \sum_{s=j}^{1} \Big( \prod_{m=s-1}^1 W_m(x)^\dag [W_s(x)^\dag, i a_{j+1} A + i b_j B] \cdot\\
&\qquad\qquad \prod_{m=j}^{s+1} W_m(x)^\dag \Big)  \prod_{l=\kappa}^{j+1} W_l(x)^\dag \\
&= \sum_{s=\kappa}^1 \sum_{j=\kappa}^s \Big( \prod_{m=s-1}^1 W_m(x)^\dag [W_s(x)^\dag, i a_{j+1} A + i b_j B] \cdot \\
&\qquad\qquad \prod_{m=j}^{s+1} W_m(x)^\dag \Big)  \prod_{l=\kappa}^{j+1} W_l(x)^\dag \\
&= \sum_{s=\kappa}^1 \prod_{m=s-1}^1 W_m(x)^\dag [W_s(x)^\dag, i c_{j+1} A + i d_s B] \prod_{l=\kappa}^{s+1} W_l(x)^\dag,
\end{aligned}
\end{equation}
where 
\begin{equation} \label{eq:csds}
c_s:= \sum_{j=\kappa}^s a_j, \quad d_s=\sum_{j=\kappa}^s b_j, \quad s\leq \kappa.
\end{equation}
We also set $c_{\kappa+1}=0$.

\begin{widetext}
Based on \autoref{eq:Rtderive3}, we now derive the succinct form of $J_K(x)$,
\begin{equation} \label{eq:Jtderive1}
\begin{aligned}
J_K(x) = W(x)R_K(x) = \prod_{l=1}^{\kappa} W_l(x) R_K(x) = \sum_{s=\kappa}^1 \prod_{l=s}^\kappa W_l(x) [W_s(x)^\dag, ic_{s+1}A + id_s B] \prod_{m=\kappa}^{s+1} W_m(x)^\dag.
\end{aligned}
\end{equation}

We expand each stage of the Trotter formula $W_l(x)$ in the formula,
\begin{equation} \label{eq:Jtderive2}
\begin{aligned}
J_K(x) &= i\sum_{j=\kappa}^1\left( \prod_{l=j+1}^\kappa W_l(x)^\dag \left( c_{j+1}A + d_j B \right) \prod_{l=\kappa}^{j+1} W_l(x)^\dag - \prod_{l=j}^\kappa W_l(x)^\dag \left( c_{j+1} A + d_j B \right) \prod_{l=\kappa}^{j} W_l(x)^\dag \right) \\
&= -i\sum_{j=\kappa}^1\left( \prod_{l=j}^\kappa e^{-i\tau b_l \ad_B} e^{-i\tau a_l \ad_A} \left( c_{j+1}A + d_j B \right) - \prod_{l=j+1}^\kappa e^{-i\tau b_l \ad_B} e^{-i\tau a_l \ad_A} \left( c_{j+1} A + d_j B \right) \right).
\end{aligned}
\end{equation}

Finally, we apply the following operator-valued Taylor expansion formula with integral form of the remainder:
\begin{equation} \label{eq:AppOperatorTaylor}
\begin{aligned}
Q(x) = \sum_{s=0}^{k} \frac{x^s}{s!} Q^{(s)}(0) + \int_0^x d\tau \frac{(x-\tau)^k}{k!} Q^{(k+1)}(\tau).
\end{aligned}
\end{equation}

By the general Libniz formula, we obtain the derivatives of $J_K(x)$ as follows:
\begin{equation} \label{eq:JxDerivatives}
\begin{aligned}
J^{(s)}_K(x) &= (-i)^{s+1} \sum_{j=\kappa}^1 \sum_{\substack{m_j,n_j;...;m_{\kappa},n_{\kappa} \\ m_j,n_j\neq0; \sum_{l=j}^{\kappa} m_l + n_l =s}} \binom{s}{m_j,n_j,...,m_\kappa, n_\kappa} \cdot \\
&\quad  \prod_{l=j}^\kappa (b_l^{m_l} a_l^{n_l}) \prod_{l=j}^\kappa \left( e^{-ix b_l \ad_B} \ad_B^{m_l} e^{-ix a_l \ad_A} \ad_A^{n_l}\right) (c_{j+1}A + d_j B). 
\end{aligned}
\end{equation}
\end{widetext}

We can then expand $J_K(x)$ around $t=0$ as follows:
\begin{equation} \label{eq:JtExpand}
\begin{aligned}
J_K(x) &= \sum_{s=K}^{2K} C_s \frac{x^s}{s!} + J_{K,res,2K}(x).
\end{aligned}
\end{equation}
Here, we use the order condition in Proposition~\ref{prop:order} so that the terms with expansion order from $0$ to $K-1$ are all zeros. The $s$th-order term $C_s$ and the $2K$th-order residue $J_{K,res,2K}(x)$ can then be expressed as
\begin{equation} \label{eq:CsJxRes}
\begin{aligned}
C_s & = J^{(s)}_K(0), \\
J_{K,res,2K}(x) &= \int_0^x d\tau \frac{(x-\tau)^{2K}}{(2K)!} J^{(2K+1)}_K(\tau).
\end{aligned}
\end{equation}

\subsection{Norm bounds for LCU formula} \label{ssc:AppNormBound}

In Sec.~\ref{ssc:AppNCDerive} we have derived the explicit form of the LCU formula for the $K$th-order Trotter remainder $V_K(x)$. Based on \autoref{eq:Jtderive2}, \autoref{eq:JxDerivatives}, \autoref{eq:CsJxRes} and Proposition~\ref{prop:NCnormBound}, we are going to prove the $1$-norm bound $\mu^{(nc)}_K(x)$ and error bound $\varepsilon^{(nc)}_K(x)$ in Proposition~\ref{prop:NCK}.

First of all, we need to estimate the norms of $J_K^{(s)}(x)$. We have the following results.

\begin{proposition}[Upper bound of the norm of nested commutators] \label{prop:adnormKth}
Consider a lattice Hamiltonian $H=A+B$ with the form in \autoref{eq:HLattice}. Suppose the spectral norm and $1$-norm of its components $H_{j,j+1}$ are bounded by $\Lambda$ and $\Lambda_1$. Then for the nested commutators appearing in \autoref{eq:CsJxRes}, we have the following bound
\begin{equation} \label{eq:adnormAB}
\begin{aligned}
& \left\| \prod_{l=j}^\kappa e^{-i\tau b_l \ad_B} \ad_B^{m_l} e^{-i\tau a_l \ad_A} \ad_A^{n_l} (c_{j+1}A + d_j B) \right\|  \\
&\quad \leq \frac{n}{2}2^s\Lambda^{s+1} \Big( c_{j+1} \prod_{l=j}^\kappa (2(l-j)+2)^{m_l}(2(l-j)+1)^{n_l} \\
&\quad\quad  + d_j \prod_{l=j}^\kappa (2(l-j)+3)^{m_l}(2(l-j)+2)^{n_l} \Big)
\end{aligned}
\end{equation}
where $\{m_l,n_l\}_{l=1}^\kappa$ are non-negative integers satisfying $\sum_{l=j}^\kappa (m_l + n_l) = s$. $\kappa$ is the stage number determined by the Trotter formula in \autoref{eq:AppSKtWt}. $\{a_l,b_l\}_{l=1}^\kappa$ are defined by the Trotter formula in \autoref{eq:AppWjt}. $\{c_l,d_l\}_{l=1}^\kappa$ are defined in \autoref{eq:csds}. As a result, we can bound the norm of $J^{(s)}_K(x)$ defined in \autoref{eq:JxDerivatives} as follows:
\begin{equation}
\|J^{(s)}_K(x)\| \leq n\kappa(4\kappa+5)^s 2^s \Lambda^{s+1}.
\end{equation}
The $1$-norm upper bound for $J^{(s)}_K(x)$ is to simply replace $\Lambda$ by $\Lambda_1$.
\end{proposition}

\begin{proof}
To begin with, we now focus on one Hamiltonian term $H_{j,j+1}$ in $A$ and bound the norm 
\begin{equation} \label{eq:adnormHA}
\begin{aligned}
&\left\| \prod_{l=j}^\kappa e^{-i\tau b_l \ad_B} \ad_B^{m_l} e^{-i\tau a_l \ad_A} \ad_A^{n_l} (H_{j,j+1}) \right\| \\
&\quad \leq 2^s\Lambda^{s+1} \prod_{l=j}^\kappa (2(l-j)+2)^{m_l}(2(l-j)+1)^{n_l}.
\end{aligned}
\end{equation} 
To this end, we decompose operator to the elementary nested commutators with the form,
\begin{equation} \label{eq:elementNCKth}
\begin{aligned}
& e^{-i\tau \ad_B} \ad_{H_{j_{s},j_{s}+1}} ... \ad_{H_{j_{s-m_\kappa+1},j_{s-m_\kappa+1}+1}} \cdot \\
& ... \\
& e^{-i\tau \ad_B} \ad_{H_{j_{m_j+n_j},j_{m_j+n_j}+1}} ... \ad_{H_{j_{n_j+1},j_{n_j+1}+1}} \cdot \\
& \; e^{-i\tau \ad_A} \ad_{H_{j_{n_j},j_{n_j}+1}} ... \ad_{H_{j_1,j_1+1}} H_{j,j+1},
\end{aligned}
\end{equation}
where $j_1, j_2, ..., j_s$ are the possible vertex indices. For each elementary nested commutator, the spectral norm can be bounded by
\begin{equation}
(2\Lambda)^s \Lambda.
\end{equation}
This can he done by expand all the commutators and apply triangle inequality. Here, we use the property that the spectral norm of all the exponential operator with anti-Hermitian exponent is 1.

Now, we count the number of the possible elementary commutators with the form in \autoref{eq:elementNCKth}. We will check the action of the adjoint operators from right to left. For the first location, we know that $\ad_A H_{j,j+1} = 0$. For simplicity, we keep the term $\ad_{H_{j,j+1}} H_{j,j+1}$ in the counting of elementary commutators. If the next $\ad$ is still $\ad_A$, the support will still be on the two qubits $j$ and $j+1$. As a result, there will be one possible elementary term $\ad_{H_{j,j+1}}$ left. Similarly, the exponential operator $e^{-i\tau \ad_A}$ will not enlarge the support since one can expand it to the power of $\ad_A$. The support will be enlarged when $\ad_B$ comes into play. In this layer, the support of the operator will be expanded to four qubits: $j-1, j, j+1$, and $j+2$. We can decompose $\ad_B$ to $2$ nonzero elementary elements, $\ad_{H_{j-1,j}}$ and $\ad_{H_{j+1,j+2}}$. If the next operator is still $\ad_B$ or $e^{-i\ad_B}$, it will not enlarge the support. Following this logic, we can see that the number of possible elementary commutators is bounded by
\begin{equation}
\begin{aligned}
&(2(\kappa-j)+2)^{m_\kappa} (2(\kappa-j)+1)^{n_\kappa} ... 4^{m_{j+1}}3^{n_{j+1}} 2^{m_j} 1^{n_j} \\
&\quad = \prod_{l=j}^\kappa (2(l-j)+2)^{m_l}(2(l-j)+1)^{n_l}.
\end{aligned}
\end{equation}
We remark that, the elementary nested commutators with $n_j\neq 0$ is actually $0$. Here, we keep these commutators for the simplicity of counting.

Combining the number of elementary nested commutators and the norm bound for each commutator and applying triangle inequality, we will obtain \autoref{eq:adnormHA}. 

Similarly, we can check one Hamiltonian term $H_{j,j+1}$ in $B$ and bound the norm with
\begin{equation} \label{eq:adnormHB}
\begin{aligned}
&\left\| \prod_{l=j}^\kappa e^{-i\tau b_l \ad_B} \ad_B^{m_l} e^{-i\tau a_l \ad_A} \ad_A^{n_l} (H_{j,j+1}) \right\| \\
&\quad\leq 2^s\Lambda^{s+1} \prod_{l=j}^\kappa (2(l-j)+3)^{m_l}(2(l-j)+2)^{n_l}.
\end{aligned}
\end{equation} 
The counting logic is similar to the case for $H_{j,j+1}$ in $A$. The only difference is that when counting the number of elementary nested commutators, the action of the first $\ad_A$ will enlarge the operator space to four qubits.

Applying \autoref{eq:adnormHA} and \autoref{eq:adnormHB} for all the components $H_{j,j+1}$ in $H$, we will obtain \autoref{eq:adnormAB}.

Now, we apply \autoref{eq:adnormAB} to bound the norm of $J^{(s)}_K(x)$ in \autoref{eq:JxDerivatives}. We have
\begin{widetext}
\begin{align*}
&\|J^{(s)}_K(x)\| \leq \sum_{j=\kappa}^1 \sum_{\substack{m_j,n_j;...;m_{\kappa},n_{\kappa} \\ \sum_{l=j}^{\kappa} m_l + n_l =s}} \binom{s}{m_j,n_j,...,m_\kappa, n_\kappa} \prod_{l=j}^\kappa (b_l^{m_l} a_l^{n_l}) \left\|\prod_{l=j}^\kappa \left( e^{-it b_l \ad_B} \ad_B^{m_l} e^{-it a_l \ad_A} \ad_A^{n_l}\right) (c_{j+1}A + d_j B)\right\| \\
&\leq 2^s\Lambda^{s+1}\frac{n}{2} \sum_{j=\kappa}^1 \sum_{\substack{m_j,n_j;...;m_{\kappa},n_{\kappa} \\ \sum_{l=j}^{\kappa} m_l + n_l =s}} \binom{s}{m_j,n_j,...,m_\kappa, n_\kappa} \prod_{l=j}^\kappa (b_l^{m_l} a_l^{n_l}) \\
&\quad \left( c_{j+1} \prod_{l=j}^\kappa (2(l-j)+2)^{m_l}(2(l-j)+1)^{n_l}  + d_j \prod_{l=j}^\kappa (2(l-j)+3)^{m_l}(2(l-j)+2)^{n_l} \right) \\
&\leq 2^s\Lambda^{s+1}\frac{n}{2} \sum_{j=\kappa}^1 c_{j+1}\left( \sum_{l=j}^\kappa b_l(2(l-j)+2) + a_l(2(l-j)+1) \right)^s + d_j \left( \sum_{l=j}^\kappa b_l(2(l-j)+3) + a_l(2(l-j)+2) \right)^s \\
&\leq 2^s\Lambda^{s+1}\frac{n}{2} \sum_{j=\kappa}^1 c_{j+1}\left( 4(\kappa-j)+3 \right)^s + d_j \left( 4(\kappa-j)+5 \right)^s \stepcounter{equation}\tag{\theequation} \label{eq:JsxBoundDerive} \\
&\leq 2^s\Lambda^{s+1}\frac{n}{2} \sum_{j=\kappa}^1 \left( 4(\kappa-j)+3 \right)^s + \left( 4(\kappa-j)+5 \right)^s \\
&\leq 2^s\Lambda^{s+1}\frac{n}{2}\cdot 2\kappa(4\kappa+5)^s = n\kappa (4\kappa+5)^s 2^s \Lambda^{s+1}.
\end{align*}
\end{widetext}
In the third line, we use multinomial theorem. In the fourth line, we use the fact that $a_l,b_l>0$ and $\sum_{l}a_l,\sum_l b_l \leq 1$. The fifth line is due to $0\leq c_j,d_j<1$ for all $j$. In the sixth line, we use the following bound:
\begin{equation}
\begin{aligned}
&\;\; \sum_{j=1}^{\kappa} (4(\kappa-j)+a)^s = \sum_{l=0}^{\kappa-1} (4l+a)^s \\
&\leq \int_0^{\kappa} (4x+a)^s dx = \frac{1}{4} \frac{(4\kappa+a)^{s+1}-a^{s+1}}{s+1} \\
&= \kappa \frac{1}{s+1} \sum_{m=0}^s (4\kappa+a)^m a^{s-m} \leq \kappa (4\kappa+a)^s.
\end{aligned}
\end{equation}
Here, $a$ is any real number.

Since $1$-norm can be estimated based on the same logic by counting the number of nested commutators and the $1$-norm of each nest commutator, the derivation for the $1$-norm is similar by replacing $\Lambda$ to $\Lambda_1$.

\end{proof}

Now, we prove the $1$-norm bound $\mu^{(nc)}_K(x)$ of $\tilde{V}_K(x)$. From Proposition~\ref{prop:NCnormBound} we have
\begin{align*}
&\|\tilde{V}_K(x)\|_1 = \sqrt{ 1 + \left( \sum_{s=K}^{2K} \|C_s\|_1 \frac{x^{s+1}}{(s+1)!} \right)^2 } \\
&\leq  1 + \frac{1}{2} \left( \sum_{s=K}^{2K} \|C_s\|_1  \frac{x^{s+1}}{(s+1)!} \right)^2 \\
&\leq 1 + \frac{1}{2} n^2 \kappa^2 \left( \sum_{s=K}^{2K} \beta^s \frac{(\Lambda_1 x)^{s+1}}{(s+1)!} \right)^2 \stepcounter{equation}\tag{\theequation}\label{eq:VKx1normBound} \\
&\leq 1 + \frac{1}{2} n^2 \kappa^2 \left( (K+1) \beta^{K+1} \frac{(\Lambda_1 x)^{K+1}}{(K+1)!} \right)^2  \\
&\leq 1 + \frac{n^2 \kappa^2 \beta^{2(K+1)}}{2 (K!)^2}  (\Lambda_1 x)^{2K+2} \\
&\leq \exp\left( \frac{n^2 \kappa^2 \beta^{2(K+1)}}{2 (K!)^2}  (\Lambda_1 x)^{2K+2} \right).
\end{align*}
Here, we set $\beta:=2(4\kappa+5)$. In the third line, we use Proposition~\ref{prop:adnormKth}. In the fourth line, we use the assumption that $\beta \Lambda_1 x\leq (K+2)$.

Then, we prove the distance bound $\varepsilon_K^{(nc)}(x) = \|F_{K,\mc{res}}\|$. From Proposition~\ref{prop:NCnormBound} we know that we only need to bound $\|J_{K,L}(\tau)\|$, $\|J_{K,res,2K}(\tau)\|$, and $M_K(\tau)$ based on \autoref{eq:JKLJKresMKbound}.

We start from $\|J_{K,L}(\tau)\|$. From Proposition~\ref{prop:adnorm1st} we have
\begin{equation} \label{eq:JKLbound}
\begin{aligned}
\|J_{K,L}(\tau)\| &\leq \sum_{s=K}^{2K} \|C_s\| \frac{\tau^s}{s!} \\
&\leq n\kappa \sum_{s=K}^{2K} \beta^s \Lambda^{s+1} \frac{\tau^s}{s!}.
\end{aligned}
\end{equation}

Then, we bound $\|J_{K,res,2K}(\tau)\|$ and $\|M_K(\tau)\|$. From \autoref{eq:JKLJKresMKbound} and Proposition~\ref{prop:adnormKth} we have
\begin{equation} \label{eq:JKres2KMKbound}
\begin{aligned}
\|J_{K,res,2K}(\tau)\| &\leq \int_0^\tau d\tau_1 \frac{(\tau-\tau_1)^{2K}}{(2K)!} \|J_K^{(2K+1)}(\tau_1)\| \\
&\leq n\kappa \frac{(\beta \tau)^{2K+1}}{(2K+1)!} \Lambda^{2K+2}, \\
\|M_K(\tau)\| &\leq \int_0^\tau d\tau_1 \int_0^{\tau_1} d\tau_2 \frac{(\tau_1-\tau_2)^{(K-1)}}{(K-1)!} \|J_K^{(K)}(\tau_2)\| \\
&\leq n\kappa \frac{\beta^K \tau^{K+1}}{(K+1)!} \Lambda^{K+1}.
\end{aligned}
\end{equation}

Based on \autoref{eq:JKLbound}, \autoref{eq:JKres2KMKbound} and Proposition~\ref{prop:NCnormBound} we have
\begin{align*}
&\|F_{K,\mc{res}}(x)\| \leq \int_0^\tau \left(\|M_K(\tau)\| \|J_L(\tau)\| + \|J_{K,res,2K}(\tau)\| \right) \\
&\leq \int_0^x d\tau (n\kappa)^2 \sum_{s=K}^{2K} \frac{\beta^{s+K}\Lambda^{s+K+2}\tau^{s+K+1}}{(K+1)!s!} \\
&\quad + (n\kappa)\frac{\beta^{2K+1}\Lambda^{2K+2}\tau^{2K+1}}{(2K+1)!}  \stepcounter{equation}\tag{\theequation}\label{eq:FKresBound} \\
&= (n\kappa)^2 \sum_{s=K}^{2K} \frac{\beta^{s+K}(\Lambda x)^{s+K+2}}{(K+1)!s!(s+K+2)} + (n\kappa)\frac{\beta^{2K+1}(\Lambda x)^{2K+2}}{(2K+2)!} \\
&\leq (n\kappa)^2 (K+1)\frac{\beta^{2K} (\Lambda x)^{2K+2}}{(2K+2)K!(K+1)!} + (n\kappa) \frac{\beta^{2K+1}x^{2K+2}}{(2K+2)!} \\
&\leq (n\kappa)^2 \frac{\beta^{2K+1}}{K!(K+1)!} (\Lambda x)^{2K+2}.
\end{align*}
In the fourth line, we assume $\beta \Lambda x\leq K+1$. In the fifth line, we use the fact that $n\kappa>1$.

To summarize, from \autoref{eq:VKx1normBound} and \autoref{eq:FKresBound}, we have proven that,
\begin{equation}
\begin{aligned}
\mu_{K}^{(nc)} &:= \|V_K(x)\|_1 \leq \exp\left( \frac{n^2 \kappa^2 \beta^{2(K+1)}}{2 (K!)^2}  (\Lambda_1 x)^{2K+2} \right), \\
\varepsilon_K^{(nc)} &:= \|\tilde{V}_K(x) - V_K(x)\| = \|F_{K,\mc{res}}(x)\| \\
&\leq (n\kappa)^2 \frac{\beta^{2K+1}}{K!(K+1)!} (\Lambda x)^{2K+2}.
\end{aligned}
\end{equation}
Here, $\beta:=2(4\kappa+5)$. We assume that $\beta x\leq \max\{\frac{K+1}{\Lambda}, \frac{K+2}{\Lambda_1}\}$. This finish the proof of Proposition~\ref{prop:NCK}.

\section{EFFICIENT SAMPLING FOR THE HIGHER-ORDER NESTED-COMMUTATOR COMPENSATION ALGORITHM} \label{sec:AppNCSampling}

Following the analysis in Appendix~\ref{sec:AppNChigherTrotter}, we now design a general sampling method for the general $K$th-order Trotter remainder $\tilde{V}_K^{(nc)}(x)$. We consider the general $K$th-order Trotter formula in the canonical form in \autoref{eq:AppSKtWt}.

The truncated Trotter remainder can be expanded as
\begin{equation} \label{eq:VKxExpandTruc}
\begin{aligned}
V_K(x) &:= I + \sum_{s=K+1}^{2K+1} F_{K,s}, \\
F_{K,s} &= C_{s-1} \frac{x^s}{s!}.
\end{aligned}
\end{equation}
Based on \autoref{eq:JxDerivatives}, we have shown that $C_s$ can be written as
\begin{equation} \label{eq:CsNCForm}
\begin{aligned}
C_s 
= \sum_{j=\kappa}^1 C_{s,j}^{(A')} + C_{s,j}^{(B')},
\end{aligned}
\end{equation}
where $A':= -iA, B':= -iB$, 
\begin{equation} \label{eq:CsjApBp}
\begin{aligned}
C_{s,j}^{(A')} &:= c_{j+1} \sum_{\substack{m_j,n_j;...;m_{\kappa},n_{\kappa} \\ m_j,n_j\neq0; \sum_{l=j}^{\kappa} m_l + n_l =s}} \binom{s}{m_j,n_j,...,m_\kappa, n_\kappa} \cdot \\ 
& \quad\quad \prod_{l=j}^\kappa (b_l^{m_l} a_l^{n_l}) \left( \ad_{B'}^{m_l} \ad_{A'}^{n_l}\right) A', \\
C_{s,j}^{(B')} &:= d_j \sum_{\substack{m_j,n_j;...;m_{\kappa},n_{\kappa} \\ m_j,n_j\neq0; \sum_{l=j}^{\kappa} m_l + n_l =s}} \binom{s}{m_j,n_j,...,m_\kappa, n_\kappa} \cdot \\
&\quad\quad \prod_{l=j}^\kappa (b_l^{m_l} a_l^{n_l}) \left( \ad_{B'}^{m_l} \ad_{A'}^{n_l}\right) B'. \\
\end{aligned}
\end{equation}
The values $c_s$ and $d_s$ are defined in \autoref{eq:csds}.

We now construct efficient random sampling of $C_{s,j}^{(A')}$ and $C_{s,j}^{(B')}$ \autoref{eq:CsjApBp} based on LCU formulas with the $1$-norm of 
\begin{equation} \label{eq:musj}
\begin{aligned}
\mu_{s,j}^{(A')} &= 2^s \Lambda_1^{s+1} \frac{n}{2} c_{j+1} \chi_{A',j}^s, \\
\mu_{s,j}^{(B')} &= 2^s \Lambda_1^{s+1} \frac{n}{2} d_j \chi_{B',j}^s,
\end{aligned}
\end{equation}
respectively. Here,
\begin{equation}
\begin{aligned}
\chi_{A',j} &:= \sum_{l=j}^\kappa b_l(2(l-j)+2) + a_l(2(l-j)+1), \\
\chi_{B',j} &:= \sum_{l=j}^\kappa b_l(2(l-j)+3) + a_l(2(l-j)+2).
\end{aligned}
\end{equation}

There are $\frac{n}{2}$ summands $\{H_{q,q+1}\}$ in $A'$. We now focus on a generic summand $H_{q,q+1}$ in $A'$ and check the action of the adjoint operators on it,
\begin{equation} \label{eq:adBp_adApHj}
\begin{aligned}
&\quad \ad_{s,j;\vec{m},\vec{n}}^{(A')} H_{q,q+1} :=\\
&(-i) \ad_{B'}^{m_\kappa} \ad_{A'}^{n_\kappa} ... \ad_{B'}^{m_{j+1}} \ad_{A'}^{n_{j+1}} \ad_{B'}^{m_j} \ad_{A'}^{n_j} H_{q,q+1}.
\end{aligned}
\end{equation}
Recall that $m_\kappa + n_\kappa + ... + m_{j+1} + n_{j+1} + m_j + n_j = s$.

We would like to follow that construction in Proposition~\ref{prop:adnormKth}, where we first decompose the adjoint operator in \autoref{eq:adBp_adApHj} into the elementary operators with the form,
\begin{equation} \label{eq:elementAd}
\begin{aligned}
\ad_{q_{1:s}}^{(s)} H_{q,q+1} := (-i)^{s+1} \ad_{H_{q_{s},q_{s}+1}} ... \ad_{H_{q_1,q_1+1}} H_{q,q+1}.
\end{aligned}
\end{equation}
For each elementary nested commutator, the $1$-norm is
\begin{equation}
\Lambda_1^{(s)} = (2\Lambda_1)^s \Lambda_1.
\end{equation}
Here, we have assumed that all the summands $\{H_{q,q+1}\}$ in $A$ or $B$ have been padded similar to \autoref{eq:PadHj} so that their $1$-norms are all $\Lambda_1$.

Now, we count the number of elementary nested commutators with the form of \autoref{eq:elementAd} in the adjoint operator $\ad_{s,j}^{(A')}$ in \autoref{eq:adBp_adApHj}. In the proof of Proposition~\ref{prop:adnormKth}, we show the number of possible elementary nested commutator is upper bounded by
\begin{equation}
\begin{aligned}
&\quad N(\ad_{s,j;\vec{m},\vec{n}}^{(A')} H_{q,q+1} ) \\
&= (2(\kappa-j)+2)^{m_\kappa} (2(\kappa-j)+1)^{n_\kappa} ... 4^{m_{j+1}}3^{n_{j+1}} 2^{m_j} 1^{n_j} \\
&= \prod_{l=j}^\kappa (2(l-j)+2)^{m_l}(2(l-j)+1)^{n_l}.
\end{aligned}
\end{equation}
We now discuss how to achieve this bound by ``padding'' zero-valued elementary nested commutators into the decomposition of $\ad_{s,j}^{(A')}$. We notice that, after the sequential action of $\ad_{A'}^{n_j}, \ad_{B'}^{m_j}, \ad_{A'}^{n_{j+1}}, \ad_{B'}^{m_{j+1}},...$, the support, i.e., the index of the qubit where the operators act nontrivially, of the resulting operator is given by the ``light cone'' of $q:(q+1), (q-1):(q+2), (q-2):(q+3), (q-3):(q+4),...$, as illustrated in\autoref{fig:AdjExpandSamp}. We keep track of the largest possible support of the resulting adjoint operators and pad the LCU formula in the following situations:
\begin{enumerate}
\item (Padding small summands) When the $1$-norm of the summand $H_{q_l,q_l+1}$ is smaller than $\Lambda_1$, we add extra $\pm I$ terms in the LCU formula of $H_{q_l,q_l+1}$ to make its $1$-norm to $\Lambda_1$, similar to \autoref{eq:PadHj}.
\item (Padding $0$-valued commutators) Many nested commutators with the form of \autoref{eq:elementAd} may be 0 due to the commutation relationship, for example $\ad_{H_{q+2,q+3}} \ad_{H_{q-1,q}} H_{q,q+1} = 0$, since the support of $\ad_{H_{q-1,q}} H_{q,q+1}$ is on qubits $q-1,q$ and $q+1$ which commutes with $H_{q+2,q+3}$. However, we keep all these $0$-valued terms as long as the support of the nested commutator is in the ``light-cone'' range of the operator. 
\item (Padding boundary terms) For the summands $H_{q,q+1}$ which are close to the boundary, after a few action of the adjoint operators, the ``light cone'' will touch the boundary of the system. In this case, we introduce virtual padding ancillary qubits and extra padding nested commutators on it:
    \begin{enumerate}
    \item For the summand $H_0$ or $H_n$ which own only one-qubit support on the boundary, we redefine the Pauli terms in the LCU of them. For example, if $H_0 = \Lambda_1 \sum_{\omega} p_\omega P_0^{(\omega)}$, we then redefined it to $H_{-1,0} = \sum_{\omega} p_\omega \left(I_{-1}\otimes P_{0}\right)^{(\omega)}$. 
    \item We define the ``virtual'' summand $H_{q_l,q_l+1}$ where qubits $q_l, q_l+1$ are all virtual qubits as
    \begin{equation}
        H_{q_l,q_l+1} := \frac{\Lambda_1}{2} \left( I_{q_l,q_l+1} + (- I_{q_l,q_l+1}) \right).
    \end{equation}
    \end{enumerate}
Since all the operations on the virtual qubits are $\pm I$, we do not need to introduce these qubits in the real implementation. 
\end{enumerate}

After these padding, we have now construct the LCU formula of $\ad_{s,j}^{(A')}$ with the following form:
\begin{equation} \label{eq:adsjApPadded}
\begin{aligned}
&\quad \ad_{s,j;\vec{m},\vec{n}}^{(A')} H_{q,q+1} \\ 
&= N(\ad_{s,j}^{(A')}H_{q,q+1})\cdot \\
&\quad\quad \sum_{q_1, ..., q_s} (-i)^{s+1} \ad_{H_{q_{s},q_{s}+1}} ... \ad_{H_{q_1,q_1+1}} H_{q,q+1} \\
&=: N(\ad_{s,j}^{(A')}H_{q,q+1}) \sum_{q_1, ..., q_s} \ad_{q_{1:s}}^{(s)} H_{q,q+1}
\end{aligned}
\end{equation}
where $q_1,...,q_s$ are indices in the light-cone region.
Recall that all the elementary operators with the form \autoref{eq:elementAd} are with the same $1$-norm of $(2\Lambda_1)^s \Lambda_1$, irrelevant of the qubit index $q$ and the rank number $\{m_l,n_l\}_{l=j}^\kappa$. Therefore, the $1$-norm of $\ad_{s,j}^{(A')} H_{q,q+1}$ is 
\begin{equation} \label{eq:adsjApH_1norm}
\mu\{\ad_{s,j;\vec{m},\vec{n}}^{(A')} H_{q,q+1}\} = N\left(\ad_{s,j;\vec{m},\vec{n}}^{(A')} \right)\, (2\Lambda_1)^s \Lambda_1,
\end{equation}
which is irrelevant of the qubit index $q$.

Now, based on \autoref{eq:adsjApPadded}, we can write $C_{s,j}^{(A')}$ in \autoref{eq:CsjApBp} as
\begin{equation} \label{eq:CsjAp_1norm}
\begin{aligned}
&\frac{C_{s,j}^{(A')} }{c_{j+1}} = \sum_{\substack{ 1\leq q\leq n \\ q:\text{odd} } } \sum_{\substack{m_j,n_j;...;m_{\kappa},n_{\kappa} \\ \sum_{l=j}^{\kappa} m_l + n_l =s}} \binom{s}{m_j,n_j,...,m_\kappa, n_\kappa}\cdot \\
&\quad\quad \prod_{l=j}^\kappa (b_l^{m_l} a_l^{n_l})  \ad_{s,j;\vec{m},\vec{n}}^{(A')} H_{q,q+1} \\
&= \sum_{\substack{ 1\leq q\leq n \\ q:\text{odd} } } \sum_{\substack{m_j,n_j;...;m_{\kappa},n_{\kappa} \\ \sum_{l=j}^{\kappa} m_l + n_l =s}} \binom{s}{m_j,n_j,...,m_\kappa, n_\kappa} \cdot \\
&\quad \prod_{l=j}^\kappa (b_l^{m_l} a_l^{n_l}) N(\ad_{s,j}^{(A')}H_{q,q+1}) \sum_{q_1,...,q_s} \ad_{q_{1:s}}^{(s)} H_{q,q+1} \\
&= \chi_{A',j}^s \mr{Mul}\left(\{m_j,n_j,...,m_\kappa,n_\kappa\}; \{\vec{p}_{A',b}, \vec{p}_{A',a}\}; s \right) \cdot \\
&\quad \sum_{\substack{ 1\leq q\leq n \\ q:\text{odd} } } \sum_{q_1,...,q_s} \ad_{q_{1:s}}^{(s)} H_{q,q+1},
\end{aligned}
\end{equation}
where
\begin{equation}
\begin{aligned}
p_{A',b;l} = \frac{b_l(2(l-j)+2)}{\chi_{A',j}}, \quad p_{A',a;l} = \frac{a_l(2(l-j)+2)}{\chi_{A',j}},
\end{aligned}
\end{equation}
for $l=j,j+1,...,\kappa$ and $\mr{Mul}(v;\vec{p};s)$ denotes a multinomial distribution where we sample the variable value $v$ based on the probability distribution $\vec{p}$ for $s$ times. Based on \autoref{eq:adsjApH_1norm} and \autoref{eq:CsjAp_1norm}, we can easily check that the $1$-norm of $C_{s,j}^{(A')}$ is given by $\mu_{s,j}^{(A')}$ in \autoref{eq:musj}. More importantly, since the $1$-norm of $\ad_{q_{1:s}}^{(s)} H_{q,q+1}$ is independent of $q$ and $q_1,...,q_s$, the sampling of $q$ from all odd qubits and $q_1,...,q_s$ from the light-cone region follows uniform distribution.

Similarly, we can decompose and pad extra terms in $C_{s,j}^{(B')}$ to construct LCU with the following form:
\begin{equation} \label{eq:CsjBpExpand}
\begin{aligned}
\frac{C_{s,j}^{(B')} }{d_j} &= \chi_{B',j}^s \mr{Mul}\left(\{m_j,n_j,...,m_\kappa,n_\kappa\}; \{\vec{p}_{B',b}, \vec{p}_{B',a}\}; s \right) \cdot \\
&\quad \sum_{\substack{ 1\leq q\leq n \\ q:\text{even} } } \sum_{q_1,...,q_s} \ad_{q_{1:s}}^{(s)} H_{q,q+1}, \\
p_{B',b;l} &= \frac{b_l(2(l-j)+3)}{\chi_{B',j}}, \quad p_{B',a;l} = \frac{a_l(2(l-j)+2)}{\chi_{A',j}},
\end{aligned}
\end{equation}
for $l=j,j+1,...,\kappa$, so that the $1$-norm of $C_{s,j}^{(B')}$ is given by $\mu_{s,j}^{(B')}$ in \autoref{eq:musj}. Again, the $1$-norm of the operator $\sum_{q_1,...,q_s} \ad_{q_{1:s}}^{(s)} H_{q,q+1}$ is $(2\Lambda_1)^s \Lambda_1$, independent of $q$ and $q_1,...,q_s$. Therefore, the $1$-norm of $C_s$ in \autoref{eq:CsNCForm} is
\begin{equation} \label{eq:muCs}
\mu\{C_s\} = \sum_{j=\kappa}^1 \mu_{s,j}^{(A')} + \mu_{s,j}^{(B')},
\end{equation}
and the $1$-norm of $F_{K,s}$ is
\begin{equation} \label{eq:muFKs}
\begin{aligned}
\mu\{F_{K,s}\} = \mu\{C_{s-1}\} \frac{x^s}{s!}, \\
\end{aligned}
\end{equation}

Based on the above analysis, we summarize the overall sampling procedure in \autoref{fig:NCSamplingSketch} and Algorithm~\ref{Alg:NCCSampling}. We consider a multistage sampling: first, we sample the expansion order $s$ based on $\mu\{F_{K,s}\}$; second, for a given expansion order $s$, we sample the terms $C_{s-1,j}^{(A')}$ or $C_{s-1,j}^{(B')}$ in the nested commutator $C_{s-1}$ based on \autoref{eq:CsNCForm} and \autoref{eq:musj}; third, we sample the power of the adjoint operators $\vec{m}$ and $\vec{n}$ based on the multinomial distribution in \autoref{eq:CsjAp_1norm} and \autoref{eq:CsjBpExpand}; fourth, we uniformly sample the specific Hamiltonian summands $H_{q,q+1}$ and uniformly sample the adjoint Hamiltonian summands $\{H_{q_1,q_1+1},..., H_{q_{s-1}, q_{s-1}+1}\}$, each from the light-cone region; finally, if there are multiple terms in the Hamiltonian summands, we then uniformly sample the specific Pauli terms in the summands.

\begin{figure}[htbp]
\centering
\includegraphics[width=\columnwidth]{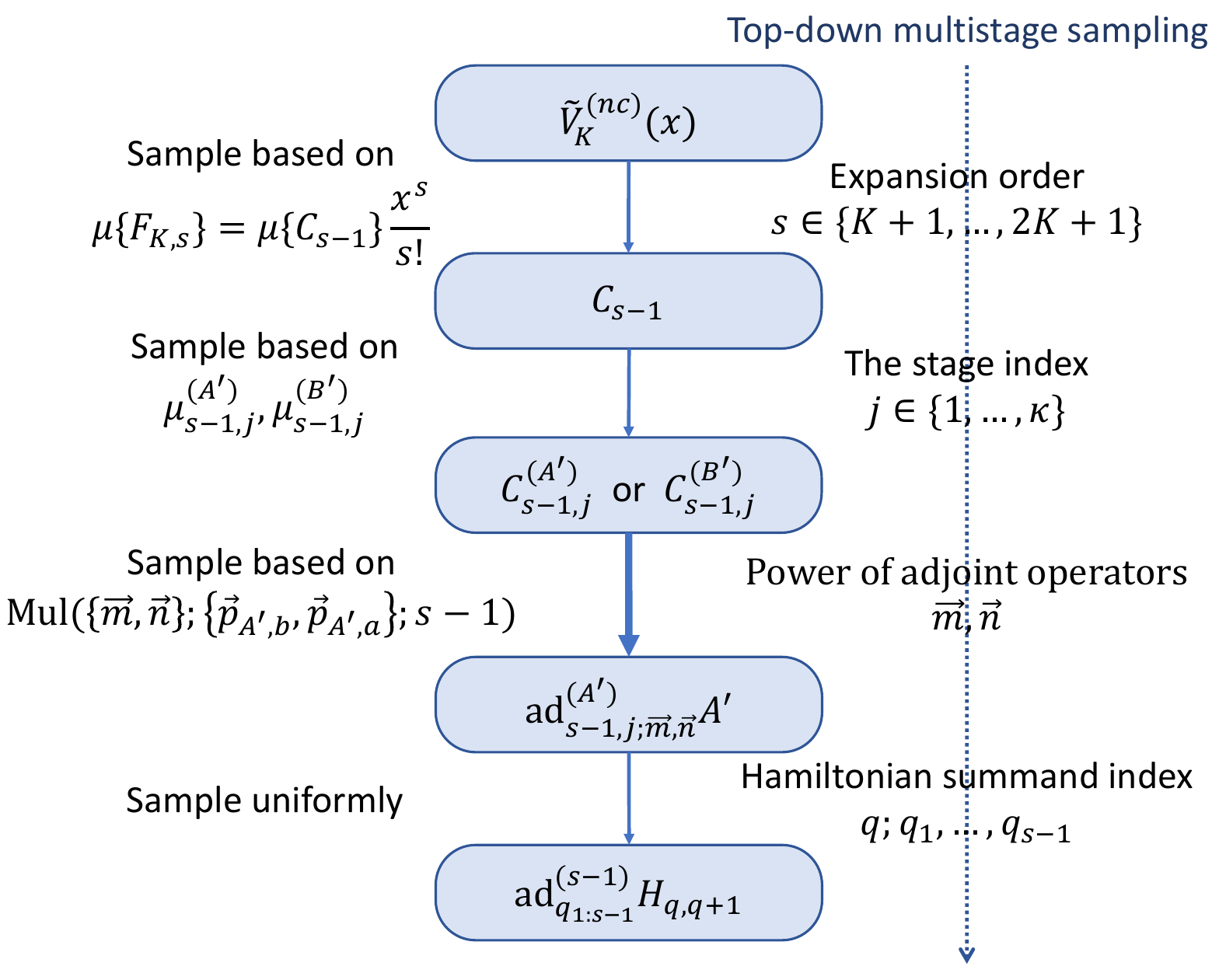} 
\caption{ 
A general procedure to obtain the sampling weights and perform the sampling in the nested-commutator compensation algorithm.}
\label{fig:NCSamplingSketch}
\end{figure}

\begin{figure}
\begin{algorithm}[H]
\caption{Nested-commutator compensation algorithm: sampling of Pauli operators}
\label{Alg:NCCSampling}
\begin{algorithmic}[1]
    \Require
    An $n$-qubit Hamiltonian $H$; unit evolution time $0<x<1$ for each $K$th-order Trotter segment; the canonical form of the $K$th-order Trotter method with coefficients $\{a_j,b_j\}$ for $j=1,...,\kappa$; 
    \Ensure
    Sampling of a Pauli operator $P^{(\omega_j;...;\omega_{j_{s}},b_{s})}$ from the Trotter remainder $\tilde{V}_K^{(nc)}(x)$.
    \State Calculate the $1$-norm of the nested commutator $C_s$ for $s=K,...,2K$ and its $j$th-stage components $C_{s,j}^{(A')}$ and $C_{s,j}^{(B')}$ for $j=1,...,\kappa$ based on \autoref{eq:musj} and \autoref{eq:muCs}.
    \State Sample the expansion order $s\in\{K,...,2K\}$ based on the $1$-norm $\mu\{F_{K,s+1}\}= \mu\{C_{s} \} \frac{x^{s+1}}{(s+1)!}$. This determines the sampled nested commutator $C_s$.
    \State From $C_s$, sample the $j$th-stage components $C_{s,j}^{(A')}$ or $C_{s,j}^{(B')}$ based on the $1$-norm $\mu_{s,j}^{(A')}$ and $\mu_{s,j}^{(B')}$ in \autoref{eq:musj}. 
    \State From $C_{s,j}^{(A')}$ (or $C_{s,j}^{(B')}$) with a given $j$, sample the power of the sequential adjoint operators $m_j,n_j;...;m_\kappa, n_\kappa$ based on a multinomial distribution $\mr{Mul}(\{\vec{m},\vec{n}\};\{\vec{p}_{A' (B'),b}, \vec{p}_{A' (B'),a}\};s)$ defined in \autoref{eq:CsjAp_1norm} and \autoref{eq:CsjBpExpand}.
    \State Sample the index of the starting Hamiltonian summand $q$ uniformly from $A'$ ($B'$). Sample the index $q_1,...,q_s$ of the subsequent adjoint Hamiltonian summands, each from the ``light-cone'' region of that location.
    \State For Hamiltonian summands indexed by $q$ and $\{q_1,..., q_s\}$, sample the Pauli operators $P^{(\omega_j)}$ and $\{P^{(\omega_{j_1})},...,P^{(\omega_{j_s})}\}$ independently based on the padded LCU formula for each Hamiltonian summand in \autoref{eq:PadHj}. For each adjoint location $q_1,...,q_s$, uniformly and independently sample the multiplication order $b_1,...,b_s\in\{0,1\}$, which indicates the multiplication order of the Pauli operators.
    \State Set $P := P^{(\omega_j)}$
    \For{$l=1~\textbf{to}~s$} \Comment{Calculate the output Pauli operator}
        \If{$b_l=0$}
            \State Set $P:= P^{(\omega_{j_1})}\cdot P$,
        \Else
            \State Set $P:= - P \cdot P^{(\omega_{j_1})}$,
        \EndIf
    \EndFor
    \State Output $P$ as the sampled Pauli operator.
\end{algorithmic}
\end{algorithm}
\end{figure}

Now, we analyze the space and time cost of the whole sampling algorithm in Algorithm~\ref{Alg:NCCSampling}. The calculation of the $1$-norms of $C_s$ and its components $C_{s,j}^{(A')}$ and $C_{s,j}^{(B')}$ requires $\mc{O}(K\kappa)$ spacetime resources. Consider a parallel calculation, we need $\mc{O}(K\kappa)$ spatial resources and $\mc{O}(1)$ time resources. We store all the above $1$-norm coefficients in the memory with the size $\mc{O}(K\kappa)$. The sampling of $F_{K,s}$ from $K$ discrete values requires $\mc{O}(\log{K})$ steps~\cite{bringmann2012efficient}. For a given nested commutator $C_{s-1}$, the sampling of $C_{s,j}^{(A')}$ and $C_{s,j}^{(B')}$ from $\kappa$ discrete values requires $\mc{O}(\log{\kappa})$ steps. The multinomial sampling of the power of adjoint operators $\vec{m}$ and $\vec{n}$ requires $\mc{O}(s\log{\kappa})$ steps. Finally, the uniform sampling of the Hamiltonian summands from the light-cone region requires $\mc{O}(s\log{n})$ steps. To summarize, the space and time cost of Algorithm~\ref{Alg:NCCSampling} are $\mc{O}(K\kappa)$ and $\mc{O}(K(\log{\kappa}+\log{n}))$, respectively.

\section{NESTED COMMUTATOR COMPENSATION FOR TROTTER FORMULAS OF GENERAL HAMILTONIANS} \label{sec:AppNCgeneral}

We now extend the methods to analyze the lattice model Hamiltonian to a general Hamiltonian. Consider a $L$-sparse Hamiltonian with the form $H = \sum_{l=1}^L H_l$.
The $K$th-order Trotter formula ($K=1$ or $2k$, $k\in \mbb{N}_+$) for $U(x)=e^{-ixH}$ can be written as
\begin{equation} \label{eq:SxGeneric}
S_K(x) = \prod_{j=1}^\kappa \prod_{l=1}^L e^{ -ix a_{(j,l)} H_{\pi_j(l)} }.
\end{equation}
Here, $\kappa$ is the number of stages in the Trotter formula, that is, how many times each Hamiltonian component $H_l$ is repeated in the implementation.
We have $\kappa=1$ for $K=1$ and $\kappa= 2\times 5^{k-1}$ when $K=2k$. The stage length coefficients $a_{(j,l)}$ are determined based on \autoref{eq:S2x} and \autoref{eq:S2kx}. The permutation $\pi_j$ indicates the ordering of the summands $\{H_l\}$ within the $j$th stage in the Trotter formula. In Suzuki's constructions~\cite{suzuki1990fractal} of Trotter formulas considered in this work, we alternately reverse the ordering of summands between neighboring stages. 

In what follows, we omit the subscript $K$ in $S_K(x)$ for simplicity. To further simplify the notation, we introduce the lexicographical order~\cite{childs2021theory} for the pair of tuples $(j,l)$ in \autoref{eq:SxGeneric}. For two pairs of tuples $(j,l)$ and $(j',l')$, we have
\begin{enumerate}
\item $(j,l)\succeq (j',l')$ if $j>j'$, or if $j=j'$ and $l>l'$.
\item $(j,l)\succ (j',l')$ if $(j,l)\succeq (j',l')$ and $(j,l)\neq (j',l')$.
\end{enumerate}
We can also define $(j,l)\preceq (j',l')$ and $(j,l)\prec (j',l')$ in the same way. We denote the number of different tuples $(j,l)$ in $S_K(x)$ as $\Upsilon$, which is usually equal to $\kappa L$.
We can then express $S_K(x)$ and $S_K(x)^\dag$ in the following way:
\begin{equation} \label{eq:SxGenOrder}
\begin{aligned}
S_K(x) = \prod_{(j,l)}^{\leftarrow} e^{-ix a_{(j,l)} H_{\pi_j(l)} }, 
\quad S_K(x)^\dag = \prod_{(j,l)}^{\rightarrow} e^{ix a_{(j,l)} H_{\pi_j(l)} }.
\end{aligned}
\end{equation}

Similar to the case of lattice Hamiltonians, based on \autoref{eq:VKxRecurrence} and \autoref{VKxNCpaired}, we can expand the Trotter remainder $V_K(x)$ to the following form:
\begin{equation}
V_K(x) = I + F_L(x) + F_{K,\mc{res}}(x),
\end{equation}
where
\begin{equation} \label{eq:FLxFKres}
\begin{aligned}
F_L(x) &:= \sum_{s=K}^{2K} C_s \frac{x^{s+1}}{(s+1)!}, \\
F_{K,\mc{res}}(x) &:= \int_0^x d\tau \left( M_K(\tau) J_L(\tau) + V_K(\tau)J_{K,res,2K}(\tau)  \right).
\end{aligned}
\end{equation}
Here, $F_L(x)$ indicates the leading-order terms to be compensated by LCU formula, $F_{K,\mc{res}}(x)=\mc{O}(x^{2K+2})$ is the high-order part.
In practice, we remove the high-order part with order $s>2K+1$ and implement only the leading-order terms
\begin{equation} \label{eq:tildeVGen}
\begin{aligned}
\tilde{V}_K^{(nc)}(x) &= I + F_L(x) = I + \sum_{s=K+1}^{2K+1} C_{s-1} \frac{x^{s}}{s!} \\
&=\sqrt{1+ (\eta_\Sigma^{(nc)})^2} \sum_{s=K+1}^{2K+1} \frac{\eta_s^{(nc)}}{\eta_\Sigma^{(nc)}} R_{2,s}^{(nc)}(\eta_\Sigma),
\end{aligned}
\end{equation}
where $\eta_\Sigma^{(nc)}:= \sum_{s=K+1}^{2K+1}\|C_{s-1}\|_1 \frac{x^s}{s!}$. The explicit form of $R_{2,s}^{(nc)}(\eta_\Sigma)$ can be obtained by the definitions in \autoref{eq:FKsNCexpand}, \autoref{VKxNCpaired} and the Pauli operator decomposition based on the nested-commutator form in \autoref{eq:NumC2C3C4}. 


We are going to finish the following tasks:
\begin{enumerate}
\item (App.~\ref{ssc:AppGenNCDerive}) Derive the explicit formulas for the leading-order expansion terms $C_s$ with $s=K,...,2K$. 
\item (App.~\ref{ssc:AppGenNCBound}) Prove the $1$-norm $\mu_K^{(nc)}(x)$ and error bound $\varepsilon_K^{(nc)}(x)$ of the LCU formula $\tilde{V}^{(nc)}_K(x)$ in \autoref{eq:tildeVGen}.
\end{enumerate}

\subsection{Derivation of nested-commutator form for general Hamiltonians} \label{ssc:AppGenNCDerive}

Based on the Trotter formulas in \autoref{eq:SxGenOrder}, we expand $R_K(x)$ in \autoref{eq:RKxDef} as follows:
\begin{equation}
\begin{aligned}
R_K(x) &= i \sum_{\gamma=1}^\Upsilon \prod_{\gamma'=\gamma-1}^{1} e^{ix a_{\gamma'} H_{\gamma'} } \left( a_\gamma H_\gamma \right)\cdot \\
&\quad \prod_{\gamma'=\Upsilon}^{\gamma+1} e^{ix a_{\gamma'} H_{\gamma'} } - iH \prod_{\gamma'=\Upsilon}^1 e^{ix a_{\gamma'} H_{\gamma'} },
\end{aligned}
\end{equation}
$J_K(x)$ in \autoref{eq:JKxDef} can then be written as,
\begin{equation} \label{eq:GenJx}
\begin{aligned}
J_K(x) &= i \sum_{\gamma=1}^\Upsilon \prod_{\gamma'=\gamma+1}^{\Upsilon} e^{-ix a_{\gamma'} \ad_{H_{\gamma'}} } (a_\gamma H_\gamma) \\
&\quad - i \prod_{\gamma'=1}^\Upsilon e^{-ix a_{\gamma'}\ad_{H_{\gamma'}}} H.
\end{aligned}
\end{equation}

The derivative of $J_K(x)$ is
\begin{widetext}
\begin{equation}
\begin{aligned}
J^{(s)}_K(x) &= i \sum_{\gamma=1}^\Upsilon \sum_{\substack{m_{\gamma+1}...m_{\Upsilon}\\ \sum_{\upsilon=\gamma+1}^\Upsilon m_{\upsilon} =s}} \binom{s}{m_{\gamma+1},...,m_{\Upsilon}} \left( \prod_{\gamma'=\gamma+1}^{\Upsilon} (-i a_{\gamma'} \ad_{H_{\gamma'}} )^{m_{\gamma'}} e^{-ix a_{\gamma'} \ad_{H_{\gamma'}} }  \right) (a_\gamma H_\gamma) \\
&\quad - i \sum_{\substack{m_1...m_{\Upsilon}\\ \sum_{\upsilon=1}^\Upsilon m_{\upsilon} =s}} \binom{s}{m_1,...,m_{\Upsilon}} \left( \prod_{\gamma'=1}^\Upsilon (-i a_{\gamma'} \ad_{H_{\gamma'}} )^{m_{\gamma'}} e^{-ix a_{\gamma'} \ad_{H_{\gamma'}} }  \right) H.
\end{aligned}
\end{equation}

Using the operator-valued Taylor-series expansion in \autoref{eq:operatorTaylor}, we have
\begin{equation} \label{eq:AppGenCsNC}
\begin{aligned}
C_s 
&= i \sum_{\gamma=1}^\Upsilon \sum_{\substack{ \sum_{(j',l')} m_{(j',l')} = s \\ (j',l')\succ (j,l) } } \binom{s}{ \{m_{(j',l')}\} } \left( \prod_{(j',l')\succ (j,l)}^{\leftarrow} (-i a_{(j',l')} \ad_{H_{\pi_{j'}(l')}} )^{m_{(j',l')}} \right) (a_{(j,l)} H_{\pi_j(l)}) \\
&\quad -i \sum_{ \sum_{(j',l')} m_{(j',l')} =s } \binom{s}{ \{m_{(j',l')}\} } \left( \prod_{(j',l')}^\leftarrow (-i a_{(j',l')} \ad_{H_{\pi_{j'}(l')}} )^{m_{(j',l')}} \right) H \\
\end{aligned}
\end{equation}
\end{widetext}
which is the nested-commutator form. Here, $\{m_{(j',l')}\}$ are a group of integers whose summation is $s$. Their corresponding multinomial coefficient is given by
\begin{equation}
\binom{s}{ \{m_{(j',l')}\} } := \frac{s!}{\prod_{(j',l')} m_{(j',l')}!}.
\end{equation}

If we define
\begin{equation}
A_{(j,l)} = A_\gamma = -i a_\gamma H_\gamma = -i a_{(j,l)} H_{\pi_j(l)},
\end{equation}
we can simplify \autoref{eq:AppGenCsNC} as
\begin{equation} \label{eq:AppGenCsSimp}
\begin{aligned}
C_s &= - \sum_{\gamma=1}^\Upsilon \sum_{\substack{m_{\gamma+1}...m_{\Upsilon}\\ \sum_{\upsilon=\gamma+1}^\Upsilon m_{\upsilon} =s}} \binom{s}{m_{\gamma+1},...,m_{\Upsilon}}  \prod_{\gamma'=\gamma+1}^\Upsilon \ad_{A_{\gamma'}}^{m_{\gamma'}} A_\gamma \\
&+ \sum_{\substack{m_1...m_{\Upsilon}\\ \sum_{\upsilon=1}^\Upsilon m_{\upsilon} =s}} \binom{s}{m_1,...,m_{\Upsilon}} \prod_{\gamma'=1}^\Upsilon \ad_{A_{\gamma'}}^{m_{\gamma'}} (-iH),
\end{aligned}
\end{equation} 
which is the summation of nested commutators. Therefore, one can still pair the leading-order terms $F_L(x)$ defined in \autoref{eq:FLxFKres} with $I$ to suppress the $1$-norm, shown in \autoref{eq:tildeVGen}.

\subsection{Norm bounds for the nested-commutator expansion} \label{ssc:AppGenNCBound}

Now, we are going to bound the $1$-norm $\mu^{(nc)}_K(x)$ and distance $\varepsilon^{(nc)}_K(x)$ of the truncated LCU formula $\tilde{V}_K^{(nc)}(x)$ in \autoref{eq:tildeVKxNCpair} for a general Hamiltonian. We will use the following formula in the derivation.
\begin{lemma}[Theorem~5 in \cite{childs2021theory}] \label{lem:tightNCbound}
Let $A_1, A_2$, ..., $A_r$ and $B$ be operators. Then, the conjugation has the expansion,
\begin{equation}
\begin{aligned}
&e^{\tau \ad_{A_r}} ... e^{\tau \ad_{A_2}} e^{\tau \ad_{A_1}} B \\
&\quad =G_0 + G_1 \tau + ... + G_{s-1}\frac{\tau^{s-1}}{(s-1)!} + G_{res,s}(\tau).
\end{aligned}
\end{equation}
Here, $G_0, G_1,...,G_{s-1}$ are operators independent of $\tau$. The operator-valued function $G_{res,s}(\tau)$ is given by
\begin{equation}
\begin{aligned}
&G_{res,s}(\tau) := \sum_{\gamma=1}^r \sum_{ \substack{m_1+...+m_\gamma=s \\ m_\gamma\neq 0} } e^{\tau \ad_{A_r}} ... e^{\tau \ad_{A_{\gamma+1}}}\cdot \\
&\int_0^\tau d\tau_2  \frac{(\tau-\tau_2)^{m_\gamma-1} \tau^{m_1+...+m_{\gamma-1}}}{(m_\gamma-1)!m_{\gamma-1}!...m_1!} e^{\tau_2 \ad_{A_\gamma}} \ad_{A_\gamma}^{m_\gamma} ... \ad_{A_1}^{m_1} B. 
\end{aligned}
\end{equation}
Furthermore, we have the spectral-norm bound,
\begin{equation}
\|G_{res,s}(\tau)\| \leq \alpha_{\mr{com}}^{(s)}(A_r,...,A_1;B) \frac{|\tau|^s}{s!} e^{2|\tau|\sum_{\gamma=1}^r \|A_r\|},
\end{equation}
for general operators and 
\begin{equation}
\|G_{res,s}(\tau)\| \leq \alpha_{\mr{com}}^{(s)}(A_r,...,A_1;B) \frac{|\tau|^s}{s!},
\end{equation}
when $A_1,...,A_r$ are anti-Hermitian. $\alpha_{\mr{com}}^{(s)}(A_r,...,A_1;B)$ is defined in \autoref{eq:alphaHalphaCom}.
\end{lemma}

We have the following proposition.
\begin{proposition}[Trotter-LCU formula by nested-commutator compensation for general Hamiltonians] \label{prop:NCGenH}
For $x>0$, $\tilde{V}_K(x)$ in \autoref{eq:tildeVGen} is a $(\mu^{(nc)}(x),\varepsilon^{(nc)}(x))$-LCU formula of $K$th-order Trotter remainder $V_K(x)=U(x)S_K(x)$ with
\begin{equation}
\begin{aligned}
\mu^{(nc)} &\leq \sqrt{1 + 4\kappa^2 \left(\sum_{s=K}^{2K} \alpha_{H,1}^{(s)} \frac{x^{s+1}}{(s+1)!} \right)^2}, \\
\varepsilon^{(nc)} &\leq 4\kappa^2 \alpha_H^{(K)} \sum_{s=K}^{2K} \alpha_H^{(s)} \frac{x^{s+K+2}}{s!(K+1)!(s+K+2)} \\
&\quad + 2\kappa \frac{x^{2K+2}}{(2K+2)!} \alpha^{(2K+1)}_H.
\end{aligned}
\end{equation}
Here, we define the nested-commutator norms to be
\begin{equation} \label{eq:alphaHalphaCom}
\begin{aligned}
&\alpha^{(s)}_H := \sum_{l=1}^L \alpha_{\mr{com}}^{(s)}( \overleftarrow{ \{ H_{\pi_{j'}(l')} \}} ; H_l), \\
&\alpha_{\mr{com}}^{(s)}(A_r,...,A_1;B) := \\
&\quad \sum_{m_1+...+m_r=s} \binom{s}{m_1,...,m_r} \|\ad_{A_r}^{m_r} ... \ad_{A_1}^{m_1} B\|,
\end{aligned}
\end{equation}
where $\overleftarrow{ \{ H_{\pi_{j'}(l')} \}}$ indicates a sequence of $\Upsilon=\kappa L$ summands with the lexicographical order given by the $K$th-order Trotter formula in \autoref{eq:SxGenOrder}. $\alpha^{(s)}_{H,1}$ is defined similarly by replacing the spectral norm to $1$-norm in \autoref{eq:alphaHalphaCom}.
\end{proposition}

\begin{proof}

From \autoref{eq:AppGenCsSimp} we can bound the $1$-norm and spectral norm of $C_s$ by
\begin{equation} \label{eq:GenCsbound}
\begin{aligned}
\|C_s\| &\leq \sum_{\gamma=1}^\Upsilon \alpha_{\mr{com}}^{(s)}(A_\Upsilon,...,A_{\gamma+1}; A_\gamma) \\
&\quad + \alpha_{\mr{com}}^{(s)}(A_\Upsilon,...,A_1;H), \\
\|C_s\|_1 &\leq \sum_{\gamma=1}^\Upsilon \alpha_{com,1}^{(s)}(A_\Upsilon,...,A_{\gamma+1}; A_\gamma) \\
&\quad + \alpha_{com,1}^{(s)}(A_\Upsilon,...,A_1;H).
\end{aligned}
\end{equation}

Now, we are going to bound $\|J_K(x)\|$ and $\|J_{K,res,2K}(x)\|$. From \autoref{eq:GenJx} and using Lemma~\ref{lem:tightNCbound} we have
\begin{equation} \label{eq:GenJxBound}
\begin{aligned}
J_K(x) &= - \sum_{\gamma=1}^\Upsilon \prod_{\gamma'=\gamma+1}^\Upsilon e^{\ad_{A_{\gamma'}}}(A_\gamma) + \prod_{\gamma'=1}^\Upsilon e^{\ad_{A_{\gamma'}}}(-iH)  \\
&= \sum_{\gamma=1}^\Upsilon G^{(\gamma)}_{res,K}(x) + G^{(0)}_{res,K}(x).
\end{aligned}
\end{equation}
In the second line, we use Lemma~\ref{lem:tightNCbound} to expand all the conjugate matrix exponentials to the following form:
\begin{equation}
\begin{aligned}
&\prod_{\gamma'=\gamma+1}^\Upsilon e^{\ad_{A_{\gamma'}}}(A_\gamma) = G^{(\gamma)}_0 + G^{(\gamma)}_1 x + .... \\
&\quad\quad + G^{(\gamma)}_{K-1} \frac{x^{K-1}}{(K-1)!} + G^{(\gamma)}_{res,K}, \quad \gamma=1,2...,\Upsilon, \\
&\prod_{\gamma'=1}^\Upsilon e^{\ad_{A_{\gamma'}}}(-iH) = G^{(0)}_0 + G^{(0)}_1 x + .... \\
&\quad\quad + G^{(0)}_{K-1} \frac{x^{K-1}}{(K-1)!} + G^{(0)}_{res,K}.
\end{aligned}
\end{equation}
The low-order terms $G_s^{(\gamma)}$ with $s=0,1,...,K-1$ in \autoref{eq:GenJxBound} cancel out due to the order condition in Proposition~\ref{prop:order}. Therefore, we have
\begin{equation} \label{eq:GenJxBound2}
\begin{aligned}
&\|J_K(x)\| \leq \sum_{\gamma=1}^\Upsilon \|G^{(\gamma)}_{res,K}(x)\| + \|G^{(0)}_{res,K}(x)\| \\
&\leq \frac{x^K}{K!} \Big( \sum_{\gamma=1}^\Upsilon \alpha_{\mr{com}}^{(K)}(A_\Upsilon,...,A_{\gamma+1}; A_\gamma)  + \alpha_{\mr{com}}^{(K)}(A_\Upsilon,...,A_1;H) \Big) \\
&\leq 2\frac{x^K}{K!} \sum_{\gamma=1}^\Upsilon \alpha_{\mr{com}}^{(K)}(A_\Upsilon,...,A_1; H_\gamma) \\
&= 2\kappa \frac{x^K}{K!} \sum_{l=1}^L \alpha_{\mr{com}}^{(K)}( \overleftarrow{ \{ a_{(j',l')} H_{\pi_{j'}(l')} \}} ; H_l) \\
&\leq 2\kappa \frac{x^K}{K!} \sum_{l=1}^L \alpha_{\mr{com}}^{(K)}( \overleftarrow{ \{ H_{\pi_{j'}(l')} \}} ; H_l) =: 2\kappa \frac{x^K}{K!} \alpha^{(K)}_H.
\end{aligned}
\end{equation}
Following the same way, we can bound $\|C_s\|$ and $\|C_s\|_1$ in \autoref{eq:GenCsbound} as
\begin{equation} \label{eq:GenCsbound2}
\begin{aligned}
\|C_s\| &\leq 2\kappa \alpha^{(s)}_H, \\
\|C_s\|_1 &\leq 2\kappa \alpha^{(s)}_{H,1}.
\end{aligned}
\end{equation}

Now, we are going to bound $\|J_{K,res,2K}(x)\|$. Similar to \autoref{eq:GenJxBound}, we expand all the conjugate matrix exponentials to $2K$th-order,
\begin{equation} \label{eq:GenJresBound}
\begin{aligned}
J_K(x) &= \sum_{\gamma=1}^\Upsilon \prod_{\gamma'=\gamma+1}^\Upsilon e^{\ad_{A_{\gamma'}}}(A_\gamma)  + \prod_{\gamma'=1}^\Upsilon e^{\ad_{A_{\gamma'}}}(-iH)  \\
&= \sum_{s=K}^{2K} C_s \frac{x^s}{s!} + \sum_{\gamma=1}^\Upsilon G^{(\gamma)}_{res,2K+1}(x) + G^{(0)}_{res,2K+1}(x),
\end{aligned}
\end{equation}
in the second line, we use the expansion of $J_K(x)$ in \autoref{eq:JKxExpand}. Based on \autoref{eq:GenJxBound2} we then have
\begin{equation} \label{eq:GenJresBound2}
\begin{aligned}
J_{K,res,2K}(x) = \sum_{\gamma=1}^\Upsilon & G^{(\gamma)}_{res,2K+1}(x) + G^{(0)}_{res,2K+1}(x) \\
\Rightarrow \|J_{K,res,2K}(x)\| &\leq \sum_{\gamma=1}^\Upsilon \|G^{(\gamma)}_{res,2K+1}(x)\| + \|G^{(0)}_{res,2K+1}(x)\| \\
&\leq 2\kappa \frac{x^{2K+1}}{(2K+1)!} \alpha^{(2K+1)}_H.
\end{aligned}
\end{equation}
We omit the derivation to the second line since it is the same as the one in \autoref{eq:GenJxBound2}.

Then we can bound $\|M_K(x)\|$ by
\begin{equation} \label{eq:GenMxBound}
\begin{aligned}
\|M_K(x)\| \leq \int_0^x d\tau \|J_K(\tau)\| \leq 2\kappa \frac{x^{K+1}}{(K+1)!} \alpha_H^{(K)}.
\end{aligned}
\end{equation}

Finally, by applying Proposition~\ref{prop:NCnormBound} and using \autoref{eq:GenCsbound2}, we can bound the $1$-norm $\mu^{(nc)}_K(x)$ as follows:
\begin{equation} \label{eq:GenMuBound}
\mu^{(nc)}_K(x) \leq \sqrt{1 + 4\kappa^2 \left(\sum_{s=K}^{2K} \alpha_{H,1}^{(s)} \frac{x^{s+1}}{(s+1)!} \right)^2}, 
\end{equation}
while the accuracy $\varepsilon^{(nc)}_K(x)$ can be bounded using \autoref{eq:GenCsbound}, \autoref{eq:GenJresBound2}, and \autoref{eq:GenMxBound},
\begin{widetext}
\begin{equation} \label{eq:GenEpsBound}
\begin{aligned}
& \varepsilon^{(nc)}_K(x) = \|F_{K,\mc{res}}(x)\| \leq \int_0^x d\tau \left( \|M_K(\tau)\| \|J_L(\tau)\| + \|J_{K,res,2K}(\tau)\|  \right) \\
&\leq \int_0^x d\tau \left( 4\kappa^2 \frac{\tau^{K+1}}{(K+1)!} \alpha_H^{(K)} \sum_{s=K}^{2K} \alpha_H^{(s)} \frac{\tau^s}{s!} + 2\kappa \frac{\tau^{2K+1}}{(2K+1)!} \alpha^{(2K+1)}_H \right) \\
&\leq 4\kappa^2 \alpha_H^{(K)} \sum_{s=K}^{2K} \alpha_H^{(s)} \frac{x^{s+K+2}}{s!(K+1)!(s+K+2)} + 2\kappa \frac{x^{2K+2}}{(2K+2)!} \alpha^{(2K+1)}_H.
\end{aligned}
\end{equation}
\end{widetext}
\end{proof}

From Proposition~\ref{prop:NCGenH} we can see that, to characterize the performance of the Trotter-LCU algorithm, we only need to estimate the values of $\alpha^{(s)}_H$ and $\alpha^{(s)}_{H,1}$ for a given Hamiltonian. We can further simplify the form of $\alpha_H^{(s)}$ by the following upper bound~\cite{childs2021theory},
\begin{equation}
\begin{aligned}
\alpha_H^{(s)} &\leq \kappa^s \sum_{l_{s+1}=1}^L ...\sum_{l_{2}=1}^L \| [H_{l_{s+1}},...[H_{l_2},H_l]]...] \| \\
&=: \kappa^s \tilde{\alpha}_{\mr{com}}(H),
\end{aligned}
\end{equation}
which is because the commutator terms in the left-hand side must be of the form on the right. Moreover, if we fix one term $\| [H_{l_{s+1}},...[H_{l_2},H_l]]...] \|$ on the right-hand side, we can find at most $\kappa^s$ times of this term on the left-hand side.

\section{SOME USEFUL FORMULAS IN THE PROOF} \label{sec:formulas}

\begin{lemma}[Tail bound for the Poisson distribution (Theorem~1 in \cite{canonne2016poisson})] \label{lem:PoissonTail}
Suppose $\hat{X}$ is a random variable with Poisson distribution so that $\Pr(\hat{X}=s) = \mr{Poi}(s;x) = e^{-x}\frac{x^s}{s!}$, where $x>0$ is the expectation value. Then, for any $\epsilon>0$, we have,
\begin{equation}
\Pr(\hat{X}\geq x+\epsilon) \leq e^{- \frac{\epsilon^2}{2x} h(\frac{\epsilon}{x})},
\end{equation}
and, for any $0<\epsilon<x$,
\begin{equation}
\Pr(\hat{X}\leq x-\epsilon) \leq e^{- \frac{\epsilon^2}{2x} h(-\frac{\epsilon}{x})}.
\end{equation}
Here, $h(u):= 2\frac{(1+u)\ln(1+u)-u}{u^2}$ for $u\geq -1$.
\end{lemma}

From Lemma~\ref{lem:PoissonTail} we have the following corollaries.
\begin{corollary} \label{coro:ExpTailPower}
For $x>0$ and positive integer $k$ such that $x<k+1$, we have
\begin{equation}
\sum_{s=k+1}^\infty \frac{x^s}{s!} \leq \left(\frac{e x}{k+1}\right)^{k+1}.
\end{equation}
\end{corollary}

\begin{proof}
We set $\epsilon = (k+1)-x$. From Lemma~\ref{lem:PoissonTail} we have
\begin{equation}
\begin{aligned}
\Pr(\hat{X}\geq k+1) &\leq \exp\left( -\frac{(k+1-x)^2}{2x} h(\frac{k+1-x}{x}) \right) \\
&= e^{-x} \left( \frac{e x}{k+1} \right)^{k+1}.
\end{aligned}
\end{equation}
Therefore,
\begin{equation}
\begin{aligned}
&\quad \Pr(\hat{X}\geq k+1) = e^{-x} \sum_{s=k+1}^{\infty} \frac{x^s}{s!} \leq e^{-x} \left( \frac{e x}{k+1} \right)^{k+1} \\
&\Rightarrow \sum_{s=k+1}^{\infty} \frac{x^s}{s!} \leq \left( \frac{e x}{k+1} \right)^{k+1}.
\end{aligned}
\end{equation}

\end{proof}

\begin{corollary} \label{coro:ExpTailExp}
For $x>0$ and positive integer $k$ such that $x<k+1$, we have
\begin{equation}
e^x - \sum_{s=1}^{k} \frac{x^s}{s!} \leq e^{e x^{k+1}}.
\end{equation}
\end{corollary}

\begin{proof}
When $x<k+1$, from Corollary~\ref{coro:ExpTailPower} we have
\begin{equation}
\begin{aligned}
&\quad e^{-x}\sum_{s=k+1}^{\infty} \frac{x^s}{s!} \leq e^{-x}\left( \frac{e x}{k+1} \right)^{k+1} \\
&\Rightarrow 1 - e^{-x} - e^{-x} \sum_{s=1}^{k} \frac{x^s}{s!} \leq e^{-x} \left( \frac{e x}{k+1} \right)^{k+1} \\
&\Rightarrow e^x - \sum_{s=1}^k \frac{x^s}{s!} \leq 1 + \left(\frac{ex}{k+1}\right)^{k+1} \\
&\quad \leq 1 + \frac{(ex)^{k+1}}{(k+1)!} \leq e^{e x^{k+1}}.
\end{aligned}
\end{equation}

\end{proof}

\begin{lemma}[Proposition~9 in \cite{wan2021randomized}] \label{lem:ExpTailSolu}
For any $\beta>0$, $1>\epsilon>0$, we have
$\left( \frac{e\beta}{s} \right)^s \leq \epsilon$,
for all 
$s \geq f(\beta,\epsilon) := \frac{\ln(\frac{1}{\epsilon})}{W_0\left( \frac{1}{e\beta}\ln(\frac{1}{\epsilon}) \right)}$.
Here, $W_0(y)$ is the principle branch of the Lambert $W$ function. 
\end{lemma}

\begin{lemma}[Theorem~2.7 in \cite{hoorfar2008inequalities}] \label{lem:WfuncBound}
When $y\geq e$ we have
\begin{equation}
\begin{aligned}
 &\ln(y) - \ln\ln(y) + \frac{1}{2} \frac{\ln\ln(y)}{\ln(y)} \leq W_0(y) \\
 & \quad \leq \ln(y) - \ln\ln(y) + \frac{e}{e-1} \frac{\ln\ln(y)}{\ln(y)}.
\end{aligned}
\end{equation}
\end{lemma}

\section{ADDITIONAL NUMERICAL RESULTS} \label{sec:AppNumerics}

In this section, we provide more numerical results by comparing the gate costs of Trotter-LCU algorithms with other typical algorithms, especially the Trotter algorithm and ``post-Trotter'' algorithm with best performance, i.e., the fourth-order Trotter algorithm and the quantum signal processing (QSP) algorithm. 
We will mainly consider two different scenarios---generic $L$-sparse Hamiltonians and lattice Hamiltonians. In the former case, the previously known best result is given by the QSP algorithm~\cite{low2017optimal}; in the latter case, since we can take advantage of the spacial locality and commutator information, the fourth-order Trotter algorithm is known to have the best performance~\cite{childs2018toward,childs2019nearly}.


For the Trotter algorithms and Trotter-LCU algorithms, we first compile the circuit to $\mc{CNOT}$ gates, single-qubit Clifford gates and non-Clifford $Z$-axis rotation gate $R_z(\theta)=e^{i\theta Z}$. 
On the other hand, the circuit compilation for the QSP algorithm is more complicated: we need to decompose the state-preparation oracles and the select-$H$ gates. We follow the qROM design in Ref.~\cite{babbush2018encoding} based on a sawtooth structure. The detailed gate resource analysis can be found in a companion work~\cite{Sun2022high}.
Based on the qROM design, we decompose all the state preparation oracles and the select-$H$ gates to Toffoli gates, which can be further decomposed to $\mc{CNOT}$ gates, single-qubit Clifford gates and $T$ gates.

For a fair comparison between the gate costs of Trotter, Trotter-LCU and QSP algorithms, we need to further compile the $Z$-axis rotation gate $R_z(\theta)$ to $T$ gate. We follow the gate compilation work in Ref.~\cite{bocharov2015efficient}, where the expected $T$-gate number to compile $R_z(\theta)$ with random $\theta$ is about
\begin{equation}
c_T = 1.149 \log_2(1/\epsilon) + 9.2,
\end{equation}
where $\epsilon$ is the compilation accuracy. We set the accuracy $\epsilon=10^{-15}$ in the later resource estimation. In this case, $c_T \approx 66$.

\subsection{Generic $L$-sparse Hamiltonians}

For the generic $L$-sparse Hamiltonian, we choose the $2$-local Hamiltonian 
\begin{equation} \label{eq:2localH}
H = \sum_{i,j} J_{ij} X_i X_j + \sum_i Z_i, 
\end{equation}
with $J_{ij} = 1$ as an example, in which case we ignore the commutator information between different Hamiltonian summands. 
This simple model works as a representative of many generic Hamiltonians where the commutator information is not helpful or too complicated to count on, e.g., quantum chemistry Hamiltonian for the molecules.
Without the commutator information, the former best Hamiltonian simulation algorithm is QSP~\cite{low2017optimal}.

In \autoref{fig:t_CNOT_Power0} and \autoref{fig:t_T_Power0}, we estimate and compare the $\mc{CNOT}$ and $T$ gate counts for the PTSC algorithms, QSP and fourth-order Trotter algorithm with respect to the evolution time. 
We choose fourth-order Trotter algorithm with random permutation since it performs the best over all the Trotter algorithms. The gate counting method for fourth-order Trotter with random permutation is based on the analytical bounds in Ref.~\cite{childs2019fasterquantum}. For the QSP algorithm, we estimate the gate count based on the QROM construction in Ref.~\cite{babbush2018encoding}. 
Only the high-accuracy results $\varepsilon = 10^{-5}$ for our method and quantum signal processing are presented since our method and quantum signal processing has a logarithmic dependence and hence are less prone to the accuracy while the Trotter formulae has polynomial dependence on the accuracy.

\begin{figure}[htbp]
\centering
\includegraphics[width=0.95\columnwidth]{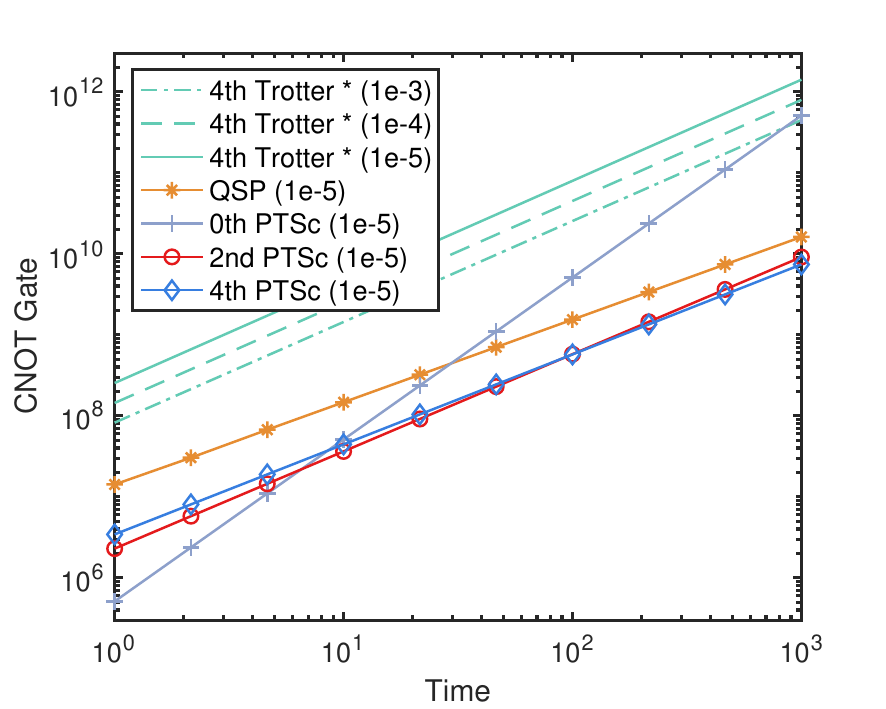} 
\caption{ 
CNOT-gate number estimation for simulating the generic $L$-sparse Hamiltonian with an increasing time. The system size is set as $n = 20$. The simulation is exemplified with the $2$-local Hamiltonian in \autoref{eq:2localH}.
}
\label{fig:t_CNOT_Power0}
\end{figure}

\begin{figure}[htbp]
\centering
\includegraphics[width=0.95\columnwidth]{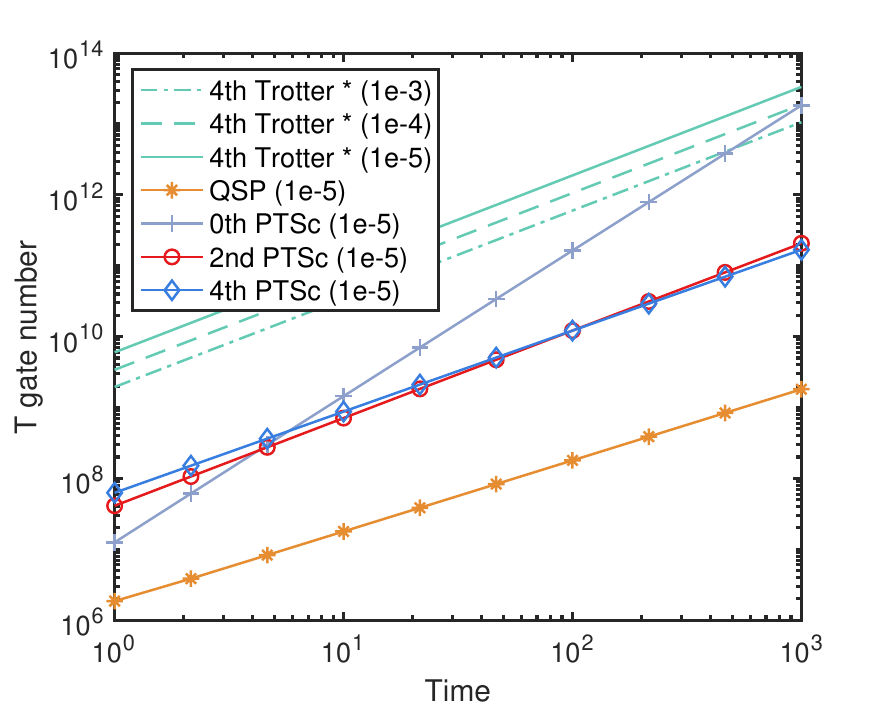} 
\caption{ 
$T$-gate number estimation for simulating the generic $L$-sparse Hamiltonian with an increasing time. The system size is set as $n = 20$. The simulation is exemplified with the $2$-local Hamiltonian in \autoref{eq:2localH}.
}
\label{fig:t_T_Power0}
\end{figure}

From \autoref{fig:t_CNOT_Power0} and \autoref{fig:t_T_Power0} we can see that, the PTSC algorithm owns a lower $\mc{CNOT}$ gate cost than the QSP algorithm because it does not need the qROM for classical data loading. On the other hand, the PTSC algorithms have a larger T-gate cost than QSP algorithm. This is mainly due to the compilation cost of arbitrary $Z$-axis rotation gate $R_z(\theta)$ to $T$ gates: for each $R_z(\theta)$ gate with a random phase $\theta$, we need $c_T=66$ $T$ gates on average to compile it to the accuracy of $\varepsilon=10^{-15}$. In the Trotter or Trotter-LCU algorithms, the non-Clifford gate cost mainly originates from the $R_z(\theta)$ gates: in each segment, there are roughly $\kappa_K L$ $R_z(\theta)$ gates, leading to about $\kappa_K c_T L$ compiled $T$ gates.
As a comparison, there are only few $R_z(\theta)$ gates in QSP algorithm used for the phase iteraction procedure to realize the Jacobi-Anger polynomials. The major non-Clifford gate cost lies in the compilation of state-preparation oracles and the select-$H$ gates, each of which can be realized by about $4L$ $T$ gates based on the qROM design in Ref.~\cite{babbush2018encoding}.

In \autoref{fig:n_CNOT_Power0} we compare the qubit number required to implement the QSP and the PTSC algorithms, based on the $2$-local Hamiltonian model. We can see a clear advantage of the spacial resource cost of PTSC to QSP algorithm.

\begin{figure}[htbp]
\centering
\includegraphics[width=0.95\columnwidth]{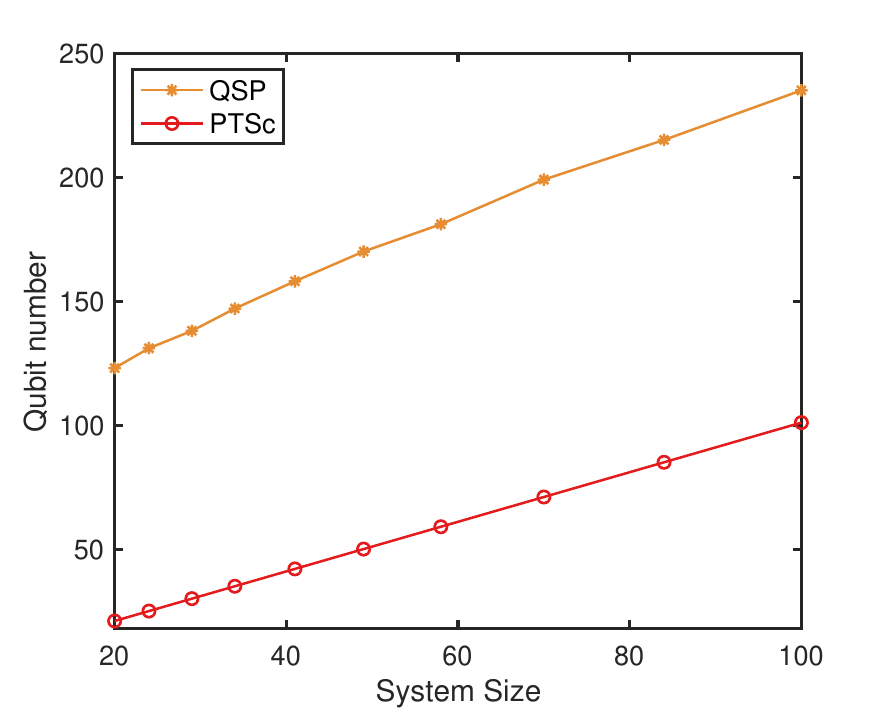}
\caption{ 
Qubit number required for simulating the generic $L$-sparse Hamiltonian with an increasing system size. The simulation is exemplified with the $2$-local Hamiltonian in \autoref{eq:2localH}.
}
\label{fig:n_CNOT_Power0}
\end{figure}

\subsection{Lattice Hamiltonians}

Now, we consider the case of lattice Hamiltonians. We consider the Heisenberg model, $H = J \sum_{i} \vec \sigma_i \vec \sigma_{i+1} + h \sum_{i} Z_i$, where $\vec \sigma_i:= (X_i, Y_i, Z_i)$ is the vector of Pauli operators on the $i$th qubit and $J=h=1$.

We compare the $\mc{CNOT}$ and $T$ gate number of the second-order NCC algorithm, QSP and fourth-order Trotter algorithm in \autoref{fig:NC_CNOT_Heis} and \autoref{fig:NC_T_Heis}, respectively. The fourth-order Trotter error analysis is based on the nested-commutator bound (Proposition~M.1 in Ref.~\cite{childs2021theory}), which is currently the tightest Trotter error bound. The performance of our second-order NCC algorithm is analyzed based on the detailed analysis in Section~\ref{sec:AppNested2nd}. We explicitly calculate the $1$-norm of the LCU formula and use the analytical bound for the accuracy analysis. Here, we do not introduce the empirical analysis for a fair comparison. Since the explicit evaluation of our fourth-order NCC algorithm is complicated, we mainly present the results for our second-order algorithm.

As addressed in \cite{childs2021theory,childs2019nearly}, fourth-order Trotter formula shows a near-optimal scaling with respect to the system size for the lattice model, which is clearly shown in \autoref{fig:NC_CNOT_Heis}.
From \autoref{fig:NC_CNOT_Heis} we can see that, our second-order NCC algorithm shows an advantage over the fourth-order Trotter algorithm. Furthermore, the $n$ and $t$ scaling of our second-order NCC algorithm is similar to the fourth-order Trotter algorithm, which is near optimal.

\begin{figure}[htbp]
\centering
\includegraphics[width=0.95\columnwidth]{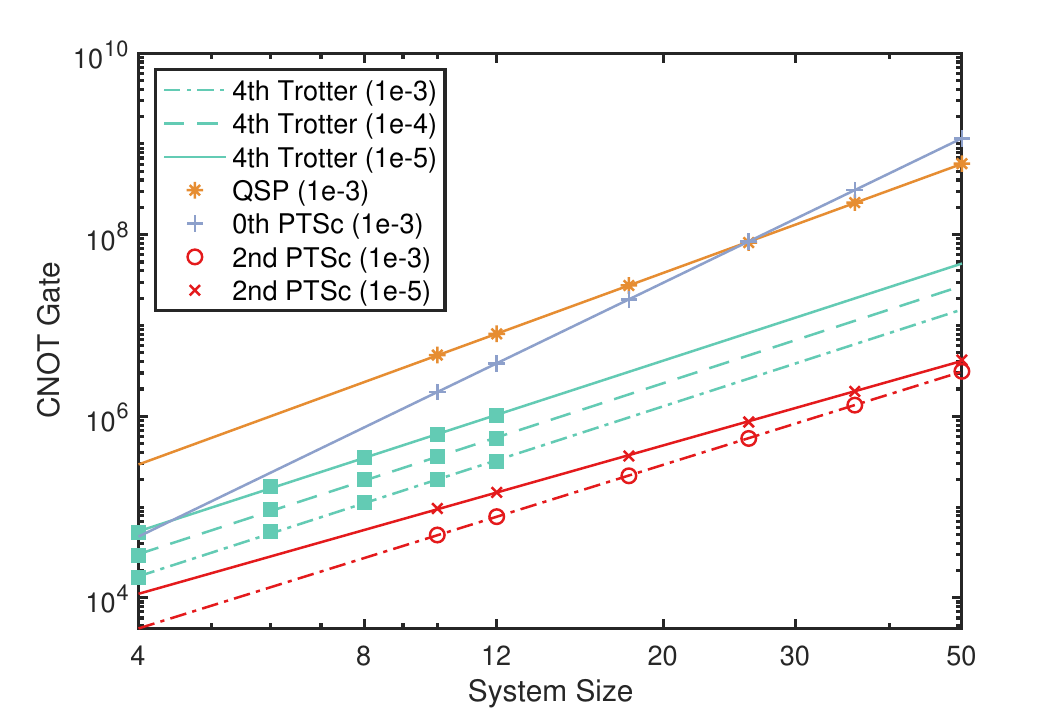}
\caption{  
CNOT gate-number estimation for simulating the Heisenberg Hamiltonian using the nested-commutator bound with an increasing system size $n$. The simulation time is set as $t = n$.
The fourth-order Trotter formula uses the nested-commutator bound shown in Proposition M.1 in \cite{childs2021theory}. Our result is based on second-order NCC algorithm with the analysis in Section~\ref{sec:AppNested2nd}.
}
\label{fig:NC_CNOT_Heis}
\end{figure}

\begin{figure}[htbp]
\centering
\includegraphics[width=0.95\columnwidth]{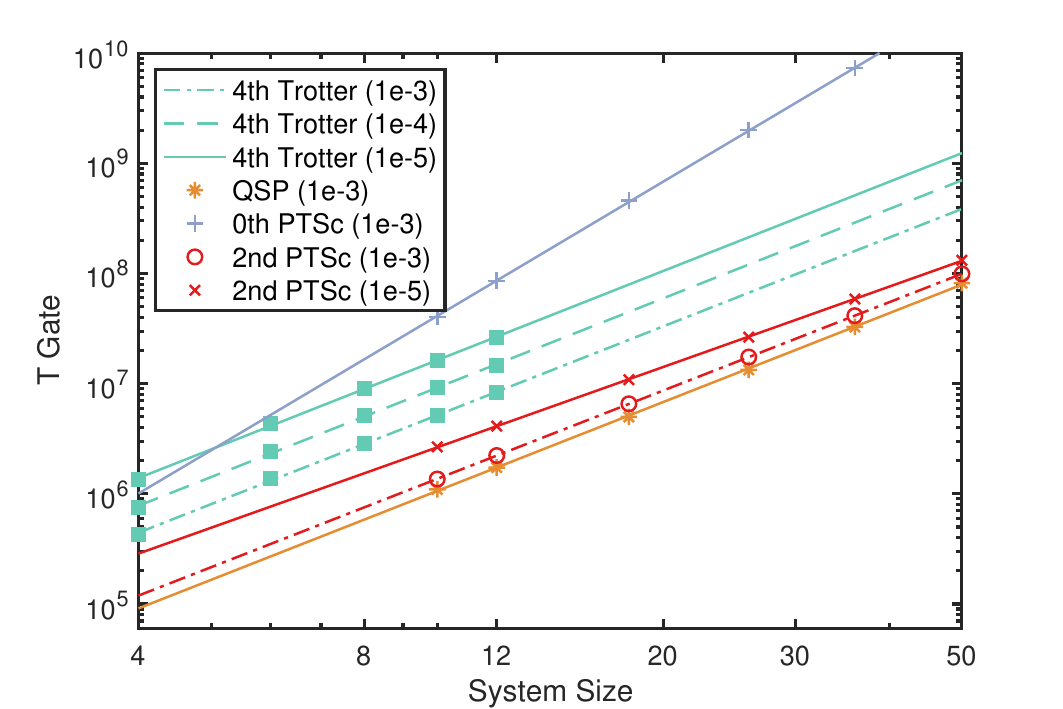}  
\caption{  
$T$ gate-number estimation for simulating the Heisenberg Hamiltonian using the nested-commutator bound with an increasing system size $n$.  The stage is set up the same as that in \autoref{fig:NC_CNOT_Heis}.
}
\label{fig:NC_T_Heis}
\end{figure}

\section{COHERENT IMPLEMENTATION} \label{sec:AppCoherent}

In the main text, we primarily focus on the random-sampling implementation of the Trotter-LCU algorithms due to its simplicity. When the block encoding of an LCU formula of $V$ is feasible on a fault-tolerant quantum computer, we can also consider the coherent implementation of the Trotter-LCU algorithms.

The gate complexity of the coherent implementation of the Trotter-LCU algorithm is determined by the segment number $\nu$, the Trotter order $K$, the number of elementary unitaries $\Gamma$ and the gate complexity of each elementary unitary in the LCU formula. The value $\Gamma$ is related to the specific compensation method we use, which is usually proportional to $L$ and the truncation order $s_c$. The gate complexity is
\begin{equation} \label{eq:gateNumCoherent}
\begin{aligned}
N_{K} = \mc{O}(\nu(\kappa_K L + \Gamma)),
\end{aligned}
\end{equation}
To estimate the segment number $\nu$, we first find a proper evolution time $x$ such that the $1$-norm $\mu_x\leq 2$. The segment number is then $\nu = t/x$.

For the PTSC formula, based on the $1$-norm bound in \autoref{thm:PTS}, we find that to ensure $\mu_x\leq 2$, it is sufficient to set $
x = \frac{1}{2 \lambda} (\frac{\ln 2}{e+c_k})^{\frac{1}{2K+2}} = \mc{O}(\frac{1}{\lambda}) $.
As a result, the number of segment $\nu=t/x = \mc{O}(\lambda t)$. From \autoref{eq:gateNumCoherent}, the overall gate complexity of the algorithm is 
\begin{equation}
N_{K}^{(p)} = \mc{O}(\nu(\kappa_K L + \Gamma)) = \mc{O}\left(\lambda t L \frac{\log(1/\varepsilon)}{\log\log(1/\varepsilon)} \right),
\end{equation}
which is the same as the case when no Trotter formula is applied~\cite{berry2015simulating}. 

For the NCC formula, based on the $1$-norm bound in \autoref{thm:NC}, we find that to ensure $\mu_x\leq 2$, we need to set $x = \frac{1}{\beta\Lambda_1}\left(\frac{2\ln2 (K!)^2}{(n\kappa)^2}\right)^{\frac{1}{2K+2}} \Rightarrow \nu = \mc{O}(n^{\frac{1}{K+1}} t)$. To achieve an overall simulation accuracy $\varepsilon$, we need
\begin{equation}
\nu \varepsilon_K^{(nc)}(x) \leq \varepsilon \Rightarrow \nu = \mc{O}( n^{\frac{2}{2K+1}} t^{1+\frac{1}{2K+1}} \varepsilon^{-\frac{1}{2K+1}} ).
\end{equation}
Therefore, it suffice to choose $\nu$ to be
\begin{equation}
\nu = \mc{O}(n^{\frac{2}{2K+1}}t^{1+\frac{1}{2K+1}}\varepsilon^{-\frac{1}{2K+1}}).
\end{equation}
In each segment, we need to implement $K$th-order Trotter formula and the $(K+1)$th- to $(2K+1)$th-order LCU formula. The number of elementary unitaries is $\mc{O}(n)$. Therefore, the overall gate complexity is
\begin{equation}
N_{K}^{(nc)} = \mc{O}(n^{1+\frac{2}{2K+1}}t^{1+\frac{1}{2K+1}}\varepsilon^{-\frac{1}{2K+1}}).
\end{equation}
Although it does not achieve the logarithmic accuracy dependence seen in standard post-Trotter algorithms~\cite{berry2015simulating,low2019Hamiltonian}, it has the unique advantage of a system-size dependence of $\mc{O}(n)$, which could be advantageous in large-scale coherent Hamiltonian simulations.

\end{appendix}

\bibliographystyle{apsrev}

\begin{thebibliography}{54}
\expandafter\ifx\csname natexlab\endcsname\relax\def\natexlab#1{#1}\fi
\expandafter\ifx\csname bibnamefont\endcsname\relax
  \def\bibnamefont#1{#1}\fi
\expandafter\ifx\csname bibfnamefont\endcsname\relax
  \def\bibfnamefont#1{#1}\fi
\expandafter\ifx\csname citenamefont\endcsname\relax
  \def\citenamefont#1{#1}\fi
\expandafter\ifx\csname url\endcsname\relax
  \def\url#1{\texttt{#1}}\fi
\expandafter\ifx\csname urlprefix\endcsname\relax\def\urlprefix{URL }\fi
\providecommand{\bibinfo}[2]{#2}
\providecommand{\eprint}[2][]{\url{#2}}

\bibitem[{\citenamefont{Feynman}(1982)}]{feynman1982simulating}
\bibinfo{author}{\bibfnamefont{R.~P.} \bibnamefont{Feynman}},
  \bibinfo{journal}{International Journal of Theoretical Physics}
  \textbf{\bibinfo{volume}{21}}, \bibinfo{pages}{467} (\bibinfo{year}{1982}),
  ISSN \bibinfo{issn}{1572-9575},
  \urlprefix\url{https://doi.org/10.1007/BF02650179}.

\bibitem[{\citenamefont{Abrams and Lloyd}(1999)}]{Lloyd1999ground}
\bibinfo{author}{\bibfnamefont{D.~S.} \bibnamefont{Abrams}} \bibnamefont{and}
  \bibinfo{author}{\bibfnamefont{S.}~\bibnamefont{Lloyd}},
  \bibinfo{journal}{Phys. Rev. Lett.} \textbf{\bibinfo{volume}{83}},
  \bibinfo{pages}{5162} (\bibinfo{year}{1999}),
  \urlprefix\url{https://link.aps.org/doi/10.1103/PhysRevLett.83.5162}.

\bibitem[{\citenamefont{Aspuru-Guzik et~al.}(2005)\citenamefont{Aspuru-Guzik,
  Dutoi, Love, and Head-Gordon}}]{Alan05Science}
\bibinfo{author}{\bibfnamefont{A.}~\bibnamefont{Aspuru-Guzik}},
  \bibinfo{author}{\bibfnamefont{A.~D.} \bibnamefont{Dutoi}},
  \bibinfo{author}{\bibfnamefont{P.~J.} \bibnamefont{Love}}, \bibnamefont{and}
  \bibinfo{author}{\bibfnamefont{M.}~\bibnamefont{Head-Gordon}},
  \bibinfo{journal}{Science} \textbf{\bibinfo{volume}{309}},
  \bibinfo{pages}{1704} (\bibinfo{year}{2005}),
  \eprint{https://www.science.org/doi/pdf/10.1126/science.1113479},
  \urlprefix\url{https://www.science.org/doi/abs/10.1126/science.1113479}.

\bibitem[{\citenamefont{Farhi et~al.}(2014)\citenamefont{Farhi, Goldstone, and
  Gutmann}}]{farhi2014quantum}
\bibinfo{author}{\bibfnamefont{E.}~\bibnamefont{Farhi}},
  \bibinfo{author}{\bibfnamefont{J.}~\bibnamefont{Goldstone}},
  \bibnamefont{and} \bibinfo{author}{\bibfnamefont{S.}~\bibnamefont{Gutmann}},
  \emph{\bibinfo{title}{A quantum approximate optimization algorithm}}
  (\bibinfo{year}{2014}), \urlprefix\url{https://arxiv.org/abs/1411.4028}.

\bibitem[{\citenamefont{Zhou et~al.}(2020)\citenamefont{Zhou, Wang, Choi,
  Pichler, and Lukin}}]{zhou2020QAOA}
\bibinfo{author}{\bibfnamefont{L.}~\bibnamefont{Zhou}},
  \bibinfo{author}{\bibfnamefont{S.-T.} \bibnamefont{Wang}},
  \bibinfo{author}{\bibfnamefont{S.}~\bibnamefont{Choi}},
  \bibinfo{author}{\bibfnamefont{H.}~\bibnamefont{Pichler}}, \bibnamefont{and}
  \bibinfo{author}{\bibfnamefont{M.~D.} \bibnamefont{Lukin}},
  \bibinfo{journal}{Phys. Rev. X} \textbf{\bibinfo{volume}{10}},
  \bibinfo{pages}{021067} (\bibinfo{year}{2020}),
  \urlprefix\url{https://link.aps.org/doi/10.1103/PhysRevX.10.021067}.

\bibitem[{\citenamefont{Harrow et~al.}(2009)\citenamefont{Harrow, Hassidim, and
  Lloyd}}]{HHL}
\bibinfo{author}{\bibfnamefont{A.~W.} \bibnamefont{Harrow}},
  \bibinfo{author}{\bibfnamefont{A.}~\bibnamefont{Hassidim}}, \bibnamefont{and}
  \bibinfo{author}{\bibfnamefont{S.}~\bibnamefont{Lloyd}},
  \bibinfo{journal}{Phys. Rev. Lett.} \textbf{\bibinfo{volume}{103}},
  \bibinfo{pages}{150502} (\bibinfo{year}{2009}),
  \urlprefix\url{https://link.aps.org/doi/10.1103/PhysRevLett.103.150502}.

\bibitem[{\citenamefont{Lloyd}(1996)}]{lloyd1996universal}
\bibinfo{author}{\bibfnamefont{S.}~\bibnamefont{Lloyd}},
  \bibinfo{journal}{Science} \textbf{\bibinfo{volume}{273}},
  \bibinfo{pages}{1073} (\bibinfo{year}{1996}),
  \eprint{https://www.science.org/doi/pdf/10.1126/science.273.5278.1073},
  \urlprefix\url{https://www.science.org/doi/abs/10.1126/science.273.5278.1073}.

\bibitem[{\citenamefont{Suzuki}(1990)}]{suzuki1990fractal}
\bibinfo{author}{\bibfnamefont{M.}~\bibnamefont{Suzuki}},
  \bibinfo{journal}{Physics Letters A} \textbf{\bibinfo{volume}{146}},
  \bibinfo{pages}{319} (\bibinfo{year}{1990}), ISSN \bibinfo{issn}{0375-9601},
  \urlprefix\url{https://www.sciencedirect.com/science/article/pii/037596019090962N}.

\bibitem[{\citenamefont{Suzuki}(1991)}]{suzuki1991general}
\bibinfo{author}{\bibfnamefont{M.}~\bibnamefont{Suzuki}},
  \bibinfo{journal}{Journal of Mathematical Physics}
  \textbf{\bibinfo{volume}{32}}, \bibinfo{pages}{400} (\bibinfo{year}{1991}),
  \urlprefix\url{https://aip.scitation.org/doi/10.1063/1.529425}.

\bibitem[{\citenamefont{Berry et~al.}(2007)\citenamefont{Berry, Ahokas, Cleve,
  and Sanders}}]{berry2007efficient}
\bibinfo{author}{\bibfnamefont{D.~W.} \bibnamefont{Berry}},
  \bibinfo{author}{\bibfnamefont{G.}~\bibnamefont{Ahokas}},
  \bibinfo{author}{\bibfnamefont{R.}~\bibnamefont{Cleve}}, \bibnamefont{and}
  \bibinfo{author}{\bibfnamefont{B.~C.} \bibnamefont{Sanders}},
  \bibinfo{journal}{Communications in Mathematical Physics}
  \textbf{\bibinfo{volume}{270}}, \bibinfo{pages}{359} (\bibinfo{year}{2007}),
  \urlprefix\url{https://doi.org/10.1007/s00220-006-0150-x}.

\bibitem[{\citenamefont{Campbell}(2019)}]{campbell2019random}
\bibinfo{author}{\bibfnamefont{E.}~\bibnamefont{Campbell}},
  \bibinfo{journal}{Phys. Rev. Lett.} \textbf{\bibinfo{volume}{123}},
  \bibinfo{pages}{070503} (\bibinfo{year}{2019}),
  \urlprefix\url{https://link.aps.org/doi/10.1103/PhysRevLett.123.070503}.

\bibitem[{\citenamefont{Childs et~al.}(2019)\citenamefont{Childs, Ostrander,
  and Su}}]{childs2019fasterquantum}
\bibinfo{author}{\bibfnamefont{A.~M.} \bibnamefont{Childs}},
  \bibinfo{author}{\bibfnamefont{A.}~\bibnamefont{Ostrander}},
  \bibnamefont{and} \bibinfo{author}{\bibfnamefont{Y.}~\bibnamefont{Su}},
  \bibinfo{journal}{{Quantum}} \textbf{\bibinfo{volume}{3}},
  \bibinfo{pages}{182} (\bibinfo{year}{2019}), ISSN \bibinfo{issn}{2521-327X},
  \urlprefix\url{https://doi.org/10.22331/q-2019-09-02-182}.

\bibitem[{\citenamefont{Childs and Su}(2019)}]{childs2019nearly}
\bibinfo{author}{\bibfnamefont{A.~M.} \bibnamefont{Childs}} \bibnamefont{and}
  \bibinfo{author}{\bibfnamefont{Y.}~\bibnamefont{Su}}, \bibinfo{journal}{Phys.
  Rev. Lett.} \textbf{\bibinfo{volume}{123}}, \bibinfo{pages}{050503}
  (\bibinfo{year}{2019}),
  \urlprefix\url{https://link.aps.org/doi/10.1103/PhysRevLett.123.050503}.

\bibitem[{\citenamefont{Endo et~al.}(2019)\citenamefont{Endo, Zhao, Li,
  Benjamin, and Yuan}}]{PhysRevA.99.012334}
\bibinfo{author}{\bibfnamefont{S.}~\bibnamefont{Endo}},
  \bibinfo{author}{\bibfnamefont{Q.}~\bibnamefont{Zhao}},
  \bibinfo{author}{\bibfnamefont{Y.}~\bibnamefont{Li}},
  \bibinfo{author}{\bibfnamefont{S.}~\bibnamefont{Benjamin}}, \bibnamefont{and}
  \bibinfo{author}{\bibfnamefont{X.}~\bibnamefont{Yuan}},
  \bibinfo{journal}{Phys. Rev. A} \textbf{\bibinfo{volume}{99}},
  \bibinfo{pages}{012334} (\bibinfo{year}{2019}),
  \urlprefix\url{https://link.aps.org/doi/10.1103/PhysRevA.99.012334}.

\bibitem[{\citenamefont{Heyl et~al.}(2019)\citenamefont{Heyl, Hauke, and
  Zoller}}]{heyl2019quantum}
\bibinfo{author}{\bibfnamefont{M.}~\bibnamefont{Heyl}},
  \bibinfo{author}{\bibfnamefont{P.}~\bibnamefont{Hauke}}, \bibnamefont{and}
  \bibinfo{author}{\bibfnamefont{P.}~\bibnamefont{Zoller}},
  \bibinfo{journal}{Science Advances} \textbf{\bibinfo{volume}{5}},
  \bibinfo{pages}{eaau8342} (\bibinfo{year}{2019}).

\bibitem[{\citenamefont{Chen et~al.}(2021)\citenamefont{Chen, Huang, Kueng, and
  Tropp}}]{chen2020quantum}
\bibinfo{author}{\bibfnamefont{C.-F.} \bibnamefont{Chen}},
  \bibinfo{author}{\bibfnamefont{H.-Y.} \bibnamefont{Huang}},
  \bibinfo{author}{\bibfnamefont{R.}~\bibnamefont{Kueng}}, \bibnamefont{and}
  \bibinfo{author}{\bibfnamefont{J.~A.} \bibnamefont{Tropp}},
  \bibinfo{journal}{PRX Quantum} \textbf{\bibinfo{volume}{2}},
  \bibinfo{pages}{040305} (\bibinfo{year}{2021}),
  \urlprefix\url{https://link.aps.org/doi/10.1103/PRXQuantum.2.040305}.

\bibitem[{\citenamefont{{\c{S}}ahino{\u{g}}lu and
  Somma}(2021)}]{sahinoglu2020hamiltonian}
\bibinfo{author}{\bibfnamefont{B.}~\bibnamefont{{\c{S}}ahino{\u{g}}lu}}
  \bibnamefont{and} \bibinfo{author}{\bibfnamefont{R.~D.} \bibnamefont{Somma}},
  \bibinfo{journal}{npj Quantum Information} \textbf{\bibinfo{volume}{7}},
  \bibinfo{pages}{1} (\bibinfo{year}{2021}),
  \urlprefix\url{https://www.nature.com/articles/s41534-021-00451-w}.

\bibitem[{\citenamefont{Su et~al.}(2021)\citenamefont{Su, Huang, and
  Campbell}}]{su2020nearly}
\bibinfo{author}{\bibfnamefont{Y.}~\bibnamefont{Su}},
  \bibinfo{author}{\bibfnamefont{H.-Y.} \bibnamefont{Huang}}, \bibnamefont{and}
  \bibinfo{author}{\bibfnamefont{E.~T.} \bibnamefont{Campbell}},
  \bibinfo{journal}{Quantum} \textbf{\bibinfo{volume}{5}}, \bibinfo{pages}{495}
  (\bibinfo{year}{2021}),
  \urlprefix\url{https://quantum-journal.org/papers/q-2021-07-05-495/}.

\bibitem[{\citenamefont{Tran et~al.}(2020)\citenamefont{Tran, Chu, Su, Childs,
  and Gorshkov}}]{Tran_2020}
\bibinfo{author}{\bibfnamefont{M.~C.} \bibnamefont{Tran}},
  \bibinfo{author}{\bibfnamefont{S.-K.} \bibnamefont{Chu}},
  \bibinfo{author}{\bibfnamefont{Y.}~\bibnamefont{Su}},
  \bibinfo{author}{\bibfnamefont{A.~M.} \bibnamefont{Childs}},
  \bibnamefont{and} \bibinfo{author}{\bibfnamefont{A.~V.}
  \bibnamefont{Gorshkov}}, \bibinfo{journal}{Phys. Rev. Lett.}
  \textbf{\bibinfo{volume}{124}}, \bibinfo{pages}{220502}
  (\bibinfo{year}{2020}),
  \urlprefix\url{https://link.aps.org/doi/10.1103/PhysRevLett.124.220502}.

\bibitem[{\citenamefont{Childs et~al.}(2021)\citenamefont{Childs, Su, Tran,
  Wiebe, and Zhu}}]{childs2021theory}
\bibinfo{author}{\bibfnamefont{A.~M.} \bibnamefont{Childs}},
  \bibinfo{author}{\bibfnamefont{Y.}~\bibnamefont{Su}},
  \bibinfo{author}{\bibfnamefont{M.~C.} \bibnamefont{Tran}},
  \bibinfo{author}{\bibfnamefont{N.}~\bibnamefont{Wiebe}}, \bibnamefont{and}
  \bibinfo{author}{\bibfnamefont{S.}~\bibnamefont{Zhu}},
  \bibinfo{journal}{Phys. Rev. X} \textbf{\bibinfo{volume}{11}},
  \bibinfo{pages}{011020} (\bibinfo{year}{2021}),
  \urlprefix\url{https://link.aps.org/doi/10.1103/PhysRevX.11.011020}.

\bibitem[{\citenamefont{Layden}(2021)}]{layden2021first}
\bibinfo{author}{\bibfnamefont{D.}~\bibnamefont{Layden}},
  \bibinfo{journal}{arXiv preprint arXiv:2107.08032}  (\bibinfo{year}{2021}).

\bibitem[{\citenamefont{Zhao et~al.}(2022)\citenamefont{Zhao, Zhou, Shaw, Li,
  and Childs}}]{PhysRevLett.129.270502}
\bibinfo{author}{\bibfnamefont{Q.}~\bibnamefont{Zhao}},
  \bibinfo{author}{\bibfnamefont{Y.}~\bibnamefont{Zhou}},
  \bibinfo{author}{\bibfnamefont{A.~F.} \bibnamefont{Shaw}},
  \bibinfo{author}{\bibfnamefont{T.}~\bibnamefont{Li}}, \bibnamefont{and}
  \bibinfo{author}{\bibfnamefont{A.~M.} \bibnamefont{Childs}},
  \bibinfo{journal}{Phys. Rev. Lett.} \textbf{\bibinfo{volume}{129}},
  \bibinfo{pages}{270502} (\bibinfo{year}{2022}),
  \urlprefix\url{https://link.aps.org/doi/10.1103/PhysRevLett.129.270502}.

\bibitem[{\citenamefont{Reiher et~al.}(2017)\citenamefont{Reiher, Wiebe, Svore,
  Wecker, and Troyer}}]{markus2017elucidating}
\bibinfo{author}{\bibfnamefont{M.}~\bibnamefont{Reiher}},
  \bibinfo{author}{\bibfnamefont{N.}~\bibnamefont{Wiebe}},
  \bibinfo{author}{\bibfnamefont{K.~M.} \bibnamefont{Svore}},
  \bibinfo{author}{\bibfnamefont{D.}~\bibnamefont{Wecker}}, \bibnamefont{and}
  \bibinfo{author}{\bibfnamefont{M.}~\bibnamefont{Troyer}},
  \bibinfo{journal}{Proceedings of the National Academy of Sciences}
  \textbf{\bibinfo{volume}{114}}, \bibinfo{pages}{7555} (\bibinfo{year}{2017}),
  \eprint{https://www.pnas.org/doi/pdf/10.1073/pnas.1619152114},
  \urlprefix\url{https://www.pnas.org/doi/abs/10.1073/pnas.1619152114}.

\bibitem[{\citenamefont{Berry et~al.}(2014)\citenamefont{Berry, Childs, Cleve,
  Kothari, and Somma}}]{berry2014exponential}
\bibinfo{author}{\bibfnamefont{D.~W.} \bibnamefont{Berry}},
  \bibinfo{author}{\bibfnamefont{A.~M.} \bibnamefont{Childs}},
  \bibinfo{author}{\bibfnamefont{R.}~\bibnamefont{Cleve}},
  \bibinfo{author}{\bibfnamefont{R.}~\bibnamefont{Kothari}}, \bibnamefont{and}
  \bibinfo{author}{\bibfnamefont{R.~D.} \bibnamefont{Somma}}, in
  \emph{\bibinfo{booktitle}{Proceedings of the Forty-Sixth Annual ACM Symposium
  on Theory of Computing}} (\bibinfo{publisher}{Association for Computing
  Machinery}, \bibinfo{address}{New York, NY, USA}, \bibinfo{year}{2014}), STOC
  '14, p. \bibinfo{pages}{283–292}, ISBN \bibinfo{isbn}{9781450327107},
  \urlprefix\url{https://doi.org/10.1145/2591796.2591854}.

\bibitem[{\citenamefont{Berry et~al.}(2015)\citenamefont{Berry, Childs, Cleve,
  Kothari, and Somma}}]{berry2015simulating}
\bibinfo{author}{\bibfnamefont{D.~W.} \bibnamefont{Berry}},
  \bibinfo{author}{\bibfnamefont{A.~M.} \bibnamefont{Childs}},
  \bibinfo{author}{\bibfnamefont{R.}~\bibnamefont{Cleve}},
  \bibinfo{author}{\bibfnamefont{R.}~\bibnamefont{Kothari}}, \bibnamefont{and}
  \bibinfo{author}{\bibfnamefont{R.~D.} \bibnamefont{Somma}},
  \bibinfo{journal}{Phys. Rev. Lett.} \textbf{\bibinfo{volume}{114}},
  \bibinfo{pages}{090502} (\bibinfo{year}{2015}),
  \urlprefix\url{https://link.aps.org/doi/10.1103/PhysRevLett.114.090502}.

\bibitem[{\citenamefont{{Berry} et~al.}(2015)\citenamefont{{Berry}, {Childs},
  and {Kothari}}}]{Berry15optimal}
\bibinfo{author}{\bibfnamefont{D.~W.} \bibnamefont{{Berry}}},
  \bibinfo{author}{\bibfnamefont{A.~M.} \bibnamefont{{Childs}}},
  \bibnamefont{and}
  \bibinfo{author}{\bibfnamefont{R.}~\bibnamefont{{Kothari}}}, in
  \emph{\bibinfo{booktitle}{56th Annual IEEE Symposium on Foundations of
  Computer Science}} (\bibinfo{year}{2015}), pp. \bibinfo{pages}{792--809},
  \urlprefix\url{https://ieeexplore.ieee.org/abstract/document/7354428}.

\bibitem[{\citenamefont{Low and Chuang}(2019)}]{low2019Hamiltonian}
\bibinfo{author}{\bibfnamefont{G.~H.} \bibnamefont{Low}} \bibnamefont{and}
  \bibinfo{author}{\bibfnamefont{I.~L.} \bibnamefont{Chuang}},
  \bibinfo{journal}{{Quantum}} \textbf{\bibinfo{volume}{3}},
  \bibinfo{pages}{163} (\bibinfo{year}{2019}), ISSN \bibinfo{issn}{2521-327X},
  \urlprefix\url{https://doi.org/10.22331/q-2019-07-12-163}.

\bibitem[{\citenamefont{Low and Chuang}(2017)}]{low2017optimal}
\bibinfo{author}{\bibfnamefont{G.~H.} \bibnamefont{Low}} \bibnamefont{and}
  \bibinfo{author}{\bibfnamefont{I.~L.} \bibnamefont{Chuang}},
  \bibinfo{journal}{Phys. Rev. Lett.} \textbf{\bibinfo{volume}{118}},
  \bibinfo{pages}{010501} (\bibinfo{year}{2017}),
  \urlprefix\url{https://link.aps.org/doi/10.1103/PhysRevLett.118.010501}.

\bibitem[{\citenamefont{Low}(2019)}]{low2019STOC}
\bibinfo{author}{\bibfnamefont{G.~H.} \bibnamefont{Low}}, in
  \emph{\bibinfo{booktitle}{51st Annual ACM Symposium on Theory of Computing}}
  (\bibinfo{year}{2019}), pp. \bibinfo{pages}{491--502}, ISBN
  \bibinfo{isbn}{9781450367059},
  \urlprefix\url{https://doi.org/10.1145/3313276.3316386}.

\bibitem[{\citenamefont{Childs and Wiebe}(2012)}]{childs2012Hamiltonian}
\bibinfo{author}{\bibfnamefont{A.~M.} \bibnamefont{Childs}} \bibnamefont{and}
  \bibinfo{author}{\bibfnamefont{N.}~\bibnamefont{Wiebe}},
  \bibinfo{journal}{Quantum Information and Computation}
  \textbf{\bibinfo{volume}{12}}, \bibinfo{pages}{0901} (\bibinfo{year}{2012}),
  ISSN \bibinfo{issn}{1533-7146},
  \urlprefix\url{http://dx.doi.org/10.26421/QIC12.11-12}.

\bibitem[{\citenamefont{Long}(2011)}]{long2011}
\bibinfo{author}{\bibfnamefont{G.~L.} \bibnamefont{Long}},
  \bibinfo{journal}{International Journal of Theoretical Physics}
  \textbf{\bibinfo{volume}{50}}, \bibinfo{pages}{1305} (\bibinfo{year}{2011}),
  ISSN \bibinfo{issn}{1572-9575},
  \urlprefix\url{https://doi.org/10.1007/s10773-010-0603-z}.

\bibitem[{\citenamefont{Lin and Tong}(2022)}]{lin2021heisenberglimited}
\bibinfo{author}{\bibfnamefont{L.}~\bibnamefont{Lin}} \bibnamefont{and}
  \bibinfo{author}{\bibfnamefont{Y.}~\bibnamefont{Tong}}, \bibinfo{journal}{PRX
  Quantum} \textbf{\bibinfo{volume}{3}}, \bibinfo{pages}{010318}
  (\bibinfo{year}{2022}),
  \urlprefix\url{https://link.aps.org/doi/10.1103/PRXQuantum.3.010318}.

\bibitem[{\citenamefont{Yang et~al.}(2021)\citenamefont{Yang, Lu, and
  Li}}]{yang2021accelerated}
\bibinfo{author}{\bibfnamefont{Y.}~\bibnamefont{Yang}},
  \bibinfo{author}{\bibfnamefont{B.-N.} \bibnamefont{Lu}}, \bibnamefont{and}
  \bibinfo{author}{\bibfnamefont{Y.}~\bibnamefont{Li}}, \bibinfo{journal}{PRX
  Quantum} \textbf{\bibinfo{volume}{2}}, \bibinfo{pages}{040361}
  (\bibinfo{year}{2021}),
  \urlprefix\url{https://link.aps.org/doi/10.1103/PRXQuantum.2.040361}.

\bibitem[{\citenamefont{Wan et~al.}(2022)\citenamefont{Wan, Berta, and
  Campbell}}]{wan2021randomized}
\bibinfo{author}{\bibfnamefont{K.}~\bibnamefont{Wan}},
  \bibinfo{author}{\bibfnamefont{M.}~\bibnamefont{Berta}}, \bibnamefont{and}
  \bibinfo{author}{\bibfnamefont{E.~T.} \bibnamefont{Campbell}},
  \bibinfo{journal}{Phys. Rev. Lett.} \textbf{\bibinfo{volume}{129}},
  \bibinfo{pages}{030503} (\bibinfo{year}{2022}),
  \urlprefix\url{https://link.aps.org/doi/10.1103/PhysRevLett.129.030503}.

\bibitem[{\citenamefont{Faehrmann et~al.}(2022)\citenamefont{Faehrmann,
  Steudtner, Kueng, Kieferova, and Eisert}}]{faehrmann2021randomizing}
\bibinfo{author}{\bibfnamefont{P.~K.} \bibnamefont{Faehrmann}},
  \bibinfo{author}{\bibfnamefont{M.}~\bibnamefont{Steudtner}},
  \bibinfo{author}{\bibfnamefont{R.}~\bibnamefont{Kueng}},
  \bibinfo{author}{\bibfnamefont{M.}~\bibnamefont{Kieferova}},
  \bibnamefont{and} \bibinfo{author}{\bibfnamefont{J.}~\bibnamefont{Eisert}},
  \bibinfo{journal}{{Quantum}} \textbf{\bibinfo{volume}{6}},
  \bibinfo{pages}{806} (\bibinfo{year}{2022}), ISSN \bibinfo{issn}{2521-327X},
  \urlprefix\url{https://doi.org/10.22331/q-2022-09-19-806}.

\bibitem[{\citenamefont{Lin and Tong}(2020)}]{lin2020nearoptimalground}
\bibinfo{author}{\bibfnamefont{L.}~\bibnamefont{Lin}} \bibnamefont{and}
  \bibinfo{author}{\bibfnamefont{Y.}~\bibnamefont{Tong}},
  \bibinfo{journal}{{Quantum}} \textbf{\bibinfo{volume}{4}},
  \bibinfo{pages}{372} (\bibinfo{year}{2020}), ISSN \bibinfo{issn}{2521-327X},
  \urlprefix\url{https://doi.org/10.22331/q-2020-12-14-372}.

\bibitem[{\citenamefont{Zeng et~al.}(2021)\citenamefont{Zeng, Sun, and
  Yuan}}]{zeng2021universal}
\bibinfo{author}{\bibfnamefont{P.}~\bibnamefont{Zeng}},
  \bibinfo{author}{\bibfnamefont{J.}~\bibnamefont{Sun}}, \bibnamefont{and}
  \bibinfo{author}{\bibfnamefont{X.}~\bibnamefont{Yuan}},
  \emph{\bibinfo{title}{Universal quantum algorithmic cooling on a quantum
  computer}} (\bibinfo{year}{2021}),
  \urlprefix\url{https://arxiv.org/abs/2109.15304}.

\bibitem[{\citenamefont{Zhang et~al.}(2022)\citenamefont{Zhang, Wang, and
  Johnson}}]{zhang2022computingground}
\bibinfo{author}{\bibfnamefont{R.}~\bibnamefont{Zhang}},
  \bibinfo{author}{\bibfnamefont{G.}~\bibnamefont{Wang}}, \bibnamefont{and}
  \bibinfo{author}{\bibfnamefont{P.}~\bibnamefont{Johnson}},
  \bibinfo{journal}{{Quantum}} \textbf{\bibinfo{volume}{6}},
  \bibinfo{pages}{761} (\bibinfo{year}{2022}), ISSN \bibinfo{issn}{2521-327X},
  \urlprefix\url{https://doi.org/10.22331/q-2022-07-11-761}.

\bibitem[{\citenamefont{Kitaev}(1995)}]{kitaev1995quantum}
\bibinfo{author}{\bibfnamefont{A.~Y.} \bibnamefont{Kitaev}},
  \emph{\bibinfo{title}{Quantum measurements and the abelian stabilizer
  problem}} (\bibinfo{year}{1995}),
  \urlprefix\url{https://arxiv.org/abs/quant-ph/9511026}.

\bibitem[{\citenamefont{Litinski}(2019)}]{Litinski2019gameofsurfacecodes}
\bibinfo{author}{\bibfnamefont{D.}~\bibnamefont{Litinski}},
  \bibinfo{journal}{{Quantum}} \textbf{\bibinfo{volume}{3}},
  \bibinfo{pages}{128} (\bibinfo{year}{2019}), ISSN \bibinfo{issn}{2521-327X},
  \urlprefix\url{https://doi.org/10.22331/q-2019-03-05-128}.

\bibitem[{\citenamefont{Canonne}(2016)}]{canonne2016poisson}
\bibinfo{author}{\bibfnamefont{C.}~\bibnamefont{Canonne}},
  \emph{\bibinfo{title}{A short note on poisson tail bounds}}
  (\bibinfo{year}{2016}),
  \urlprefix\url{http://www.cs.columbia.edu/~ccanonne/files/misc/2017-poissonconcentration.pdf}.

\bibitem[{\citenamefont{Babbush
  et~al.}(2018{\natexlab{a}})\citenamefont{Babbush, Wiebe, McClean, McClain,
  Neven, and Chan}}]{babbush2018lowdepth}
\bibinfo{author}{\bibfnamefont{R.}~\bibnamefont{Babbush}},
  \bibinfo{author}{\bibfnamefont{N.}~\bibnamefont{Wiebe}},
  \bibinfo{author}{\bibfnamefont{J.}~\bibnamefont{McClean}},
  \bibinfo{author}{\bibfnamefont{J.}~\bibnamefont{McClain}},
  \bibinfo{author}{\bibfnamefont{H.}~\bibnamefont{Neven}}, \bibnamefont{and}
  \bibinfo{author}{\bibfnamefont{G.~K.-L.} \bibnamefont{Chan}},
  \bibinfo{journal}{Phys. Rev. X} \textbf{\bibinfo{volume}{8}},
  \bibinfo{pages}{011044} (\bibinfo{year}{2018}{\natexlab{a}}),
  \urlprefix\url{https://link.aps.org/doi/10.1103/PhysRevX.8.011044}.

\bibitem[{\citenamefont{Lee et~al.}(2021)\citenamefont{Lee, Berry, Gidney,
  Huggins, McClean, Wiebe, and Babbush}}]{lee2021even}
\bibinfo{author}{\bibfnamefont{J.}~\bibnamefont{Lee}},
  \bibinfo{author}{\bibfnamefont{D.~W.} \bibnamefont{Berry}},
  \bibinfo{author}{\bibfnamefont{C.}~\bibnamefont{Gidney}},
  \bibinfo{author}{\bibfnamefont{W.~J.} \bibnamefont{Huggins}},
  \bibinfo{author}{\bibfnamefont{J.~R.} \bibnamefont{McClean}},
  \bibinfo{author}{\bibfnamefont{N.}~\bibnamefont{Wiebe}}, \bibnamefont{and}
  \bibinfo{author}{\bibfnamefont{R.}~\bibnamefont{Babbush}},
  \bibinfo{journal}{PRX Quantum} \textbf{\bibinfo{volume}{2}},
  \bibinfo{pages}{030305} (\bibinfo{year}{2021}),
  \urlprefix\url{https://link.aps.org/doi/10.1103/PRXQuantum.2.030305}.

\bibitem[{\citenamefont{Berry et~al.}(2019)\citenamefont{Berry, Gidney, Motta,
  McClean, and Babbush}}]{Berry2019qubitizationof}
\bibinfo{author}{\bibfnamefont{D.~W.} \bibnamefont{Berry}},
  \bibinfo{author}{\bibfnamefont{C.}~\bibnamefont{Gidney}},
  \bibinfo{author}{\bibfnamefont{M.}~\bibnamefont{Motta}},
  \bibinfo{author}{\bibfnamefont{J.~R.} \bibnamefont{McClean}},
  \bibnamefont{and} \bibinfo{author}{\bibfnamefont{R.}~\bibnamefont{Babbush}},
  \bibinfo{journal}{{Quantum}} \textbf{\bibinfo{volume}{3}},
  \bibinfo{pages}{208} (\bibinfo{year}{2019}), ISSN \bibinfo{issn}{2521-327X},
  \urlprefix\url{https://doi.org/10.22331/q-2019-12-02-208}.

\bibitem[{\citenamefont{von Burg et~al.}(2021)\citenamefont{von Burg, Low,
  H\"aner, Steiger, Reiher, Roetteler, and Troyer}}]{vonBurg2021quantum}
\bibinfo{author}{\bibfnamefont{V.}~\bibnamefont{von Burg}},
  \bibinfo{author}{\bibfnamefont{G.~H.} \bibnamefont{Low}},
  \bibinfo{author}{\bibfnamefont{T.}~\bibnamefont{H\"aner}},
  \bibinfo{author}{\bibfnamefont{D.~S.} \bibnamefont{Steiger}},
  \bibinfo{author}{\bibfnamefont{M.}~\bibnamefont{Reiher}},
  \bibinfo{author}{\bibfnamefont{M.}~\bibnamefont{Roetteler}},
  \bibnamefont{and} \bibinfo{author}{\bibfnamefont{M.}~\bibnamefont{Troyer}},
  \bibinfo{journal}{Phys. Rev. Res.} \textbf{\bibinfo{volume}{3}},
  \bibinfo{pages}{033055} (\bibinfo{year}{2021}),
  \urlprefix\url{https://link.aps.org/doi/10.1103/PhysRevResearch.3.033055}.

\bibitem[{\citenamefont{Hagan and Wiebe}(2023)}]{hagan2022composite}
\bibinfo{author}{\bibfnamefont{M.}~\bibnamefont{Hagan}} \bibnamefont{and}
  \bibinfo{author}{\bibfnamefont{N.}~\bibnamefont{Wiebe}},
  \bibinfo{journal}{{Quantum}} \textbf{\bibinfo{volume}{7}},
  \bibinfo{pages}{1181} (\bibinfo{year}{2023}), ISSN \bibinfo{issn}{2521-327X},
  \urlprefix\url{https://doi.org/10.22331/q-2023-11-14-1181}.

\bibitem[{\citenamefont{Cho et~al.}(2024)\citenamefont{Cho, Berry, and
  Hsieh}}]{cho2022doubling}
\bibinfo{author}{\bibfnamefont{C.-H.} \bibnamefont{Cho}},
  \bibinfo{author}{\bibfnamefont{D.~W.} \bibnamefont{Berry}}, \bibnamefont{and}
  \bibinfo{author}{\bibfnamefont{M.-H.} \bibnamefont{Hsieh}},
  \bibinfo{journal}{Phys. Rev. A} \textbf{\bibinfo{volume}{109}},
  \bibinfo{pages}{062431} (\bibinfo{year}{2024}),
  \urlprefix\url{https://link.aps.org/doi/10.1103/PhysRevA.109.062431}.

\bibitem[{\citenamefont{Zhao and Yuan}(2021)}]{Zhao2021exploiting}
\bibinfo{author}{\bibfnamefont{Q.}~\bibnamefont{Zhao}} \bibnamefont{and}
  \bibinfo{author}{\bibfnamefont{X.}~\bibnamefont{Yuan}},
  \bibinfo{journal}{{Quantum}} \textbf{\bibinfo{volume}{5}},
  \bibinfo{pages}{534} (\bibinfo{year}{2021}), ISSN \bibinfo{issn}{2521-327X},
  \urlprefix\url{https://doi.org/10.22331/q-2021-08-31-534}.

\bibitem[{\citenamefont{Childs et~al.}(2018)\citenamefont{Childs, Maslov, Nam,
  Ross, and Su}}]{childs2018toward}
\bibinfo{author}{\bibfnamefont{A.~M.} \bibnamefont{Childs}},
  \bibinfo{author}{\bibfnamefont{D.}~\bibnamefont{Maslov}},
  \bibinfo{author}{\bibfnamefont{Y.}~\bibnamefont{Nam}},
  \bibinfo{author}{\bibfnamefont{N.~J.} \bibnamefont{Ross}}, \bibnamefont{and}
  \bibinfo{author}{\bibfnamefont{Y.}~\bibnamefont{Su}},
  \bibinfo{journal}{Proceedings of the National Academy of Sciences}
  \textbf{\bibinfo{volume}{115}}, \bibinfo{pages}{9456} (\bibinfo{year}{2018}),
  \urlprefix\url{https://www.pnas.org/content/115/38/9456}.

\bibitem[{\citenamefont{Bringmann and
  Panagiotou}(2012)}]{bringmann2012efficient}
\bibinfo{author}{\bibfnamefont{K.}~\bibnamefont{Bringmann}} \bibnamefont{and}
  \bibinfo{author}{\bibfnamefont{K.}~\bibnamefont{Panagiotou}}, in
  \emph{\bibinfo{booktitle}{Automata, Languages, and Programming}}, edited by
  \bibinfo{editor}{\bibfnamefont{A.}~\bibnamefont{Czumaj}},
  \bibinfo{editor}{\bibfnamefont{K.}~\bibnamefont{Mehlhorn}},
  \bibinfo{editor}{\bibfnamefont{A.}~\bibnamefont{Pitts}}, \bibnamefont{and}
  \bibinfo{editor}{\bibfnamefont{R.}~\bibnamefont{Wattenhofer}}
  (\bibinfo{publisher}{Springer Berlin Heidelberg}, \bibinfo{address}{Berlin,
  Heidelberg}, \bibinfo{year}{2012}), pp. \bibinfo{pages}{133--144}, ISBN
  \bibinfo{isbn}{978-3-642-31594-7}.

\bibitem[{\citenamefont{Hoorfar and Hassani}(2008)}]{hoorfar2008inequalities}
\bibinfo{author}{\bibfnamefont{A.}~\bibnamefont{Hoorfar}} \bibnamefont{and}
  \bibinfo{author}{\bibfnamefont{M.}~\bibnamefont{Hassani}},
  \bibinfo{journal}{J. Inequal. Pure and Appl. Math}
  \textbf{\bibinfo{volume}{9}}, \bibinfo{pages}{5} (\bibinfo{year}{2008}),
  \urlprefix\url{https://www.emis.de/journals/JIPAM/article983.html?sid=983}.

\bibitem[{\citenamefont{Babbush
  et~al.}(2018{\natexlab{b}})\citenamefont{Babbush, Gidney, Berry, Wiebe,
  McClean, Paler, Fowler, and Neven}}]{babbush2018encoding}
\bibinfo{author}{\bibfnamefont{R.}~\bibnamefont{Babbush}},
  \bibinfo{author}{\bibfnamefont{C.}~\bibnamefont{Gidney}},
  \bibinfo{author}{\bibfnamefont{D.~W.} \bibnamefont{Berry}},
  \bibinfo{author}{\bibfnamefont{N.}~\bibnamefont{Wiebe}},
  \bibinfo{author}{\bibfnamefont{J.}~\bibnamefont{McClean}},
  \bibinfo{author}{\bibfnamefont{A.}~\bibnamefont{Paler}},
  \bibinfo{author}{\bibfnamefont{A.}~\bibnamefont{Fowler}}, \bibnamefont{and}
  \bibinfo{author}{\bibfnamefont{H.}~\bibnamefont{Neven}},
  \bibinfo{journal}{Phys. Rev. X} \textbf{\bibinfo{volume}{8}},
  \bibinfo{pages}{041015} (\bibinfo{year}{2018}{\natexlab{b}}),
  \urlprefix\url{https://link.aps.org/doi/10.1103/PhysRevX.8.041015}.

\bibitem[{\citenamefont{Sun et~al.}(2024)\citenamefont{Sun, Zeng, Gur, and
  Kim}}]{Sun2022high}
\bibinfo{author}{\bibfnamefont{J.}~\bibnamefont{Sun}},
  \bibinfo{author}{\bibfnamefont{P.}~\bibnamefont{Zeng}},
  \bibinfo{author}{\bibfnamefont{T.}~\bibnamefont{Gur}}, \bibnamefont{and}
  \bibinfo{author}{\bibfnamefont{M.}~\bibnamefont{Kim}},
  \bibinfo{journal}{arXiv preprint arXiv:2406.04307}  (\bibinfo{year}{2024}).

\bibitem[{\citenamefont{Bocharov et~al.}(2015)\citenamefont{Bocharov,
  Roetteler, and Svore}}]{bocharov2015efficient}
\bibinfo{author}{\bibfnamefont{A.}~\bibnamefont{Bocharov}},
  \bibinfo{author}{\bibfnamefont{M.}~\bibnamefont{Roetteler}},
  \bibnamefont{and} \bibinfo{author}{\bibfnamefont{K.~M.} \bibnamefont{Svore}},
  \bibinfo{journal}{Phys. Rev. Lett.} \textbf{\bibinfo{volume}{114}},
  \bibinfo{pages}{080502} (\bibinfo{year}{2015}),
  \urlprefix\url{https://link.aps.org/doi/10.1103/PhysRevLett.114.080502}.

\end{thebibliography}

\end{document}